\documentclass[acmsmall,screen,nonacm]{acmart}

\makeatletter
\def\@subsubsecfont{\sffamily\bfseries}
\makeatother

\AtBeginDocument{%
  }

\setcopyright{none} %

\usepackage{amsmath,amsfonts}
\usepackage{amsthm}
\usepackage{thm-restate}
\usepackage{mathtools}
\usepackage{mathpartir}
\usepackage[inference]{semantic}

\usepackage{url}

\usepackage{graphicx}
\usepackage{xcolor}
\usepackage{tikz}
\usepackage{tikz-cd}
\usepackage{quantikz}
\usepackage{listings}
\theoremstyle{plain}
\newtheorem{theorem}{Theorem}
\newtheorem{proposition}{Proposition}
\newtheorem{lemma}{Lemma}

\theoremstyle{definition}
\newtheorem{remark}{Remark}
\newtheorem{example}{Example}
\AtBeginEnvironment{definition}{%
  \pushQED{\qed}%
}
\AtEndEnvironment{definition}{\popQED\enddefinition}
\AtBeginEnvironment{remark}{%
  \pushQED{\qed}%
}
\AtEndEnvironment{remark}{\popQED\endremark}
\AtBeginEnvironment{example}{%
  \pushQED{\qed}%
}
\AtEndEnvironment{example}{\popQED\endexample}

\usepackage{bbold}
\usepackage{mathrsfs}
\usepackage{amsmath}
\usepackage{stmaryrd}
\usepackage{euscript}
\renewcommand*{\mathcal}{\CMcal}

\makeatletter
\newcommand{\xmultimap}[2][]{\ext@arrow 0359\multimapfill@{#1}{#2}}
\newcommand*{\multimapfill@}{%
	\arrowfill@\relbar\relbar\multimap}
\makeatother

\newcommand{\labOne}{\ell} \newcommand{\labTwo}{k} \newcommand{\labThree}{t} \newcommand{\labFour}{q}
\newcommand*{\lvOne}{a} \newcommand{\lvTwo}{b}
\newcommand{\wtypeOne}{w}
\newcommand{\lcOne}{Q} \newcommand{\lcTwo}{L}  
\newcommand{\gateOne}{g} 
\newcommand*{\Hadamard}{\ensuremath{\mathit{H}}}
\newcommand*{\CNOT}{\ensuremath{\mathit{CNOT}}}
\newcommand*{\circuitOne}{C} \newcommand*{\circuitTwo}{D} \newcommand*{\circuitThree}{E}
\newcommand{\mtypeOne}{T} \newcommand{\mtypeTwo}{U}

\newcommand*{\sig}{\Sigma}
\newcommand*{\sigQC}{\sig^{\mathrm{QC}}}
\newcommand{\cidentity}[1]{id_{#1}}

\newcommand{\unitt}{\mathbb{1}}
\newcommand*{\munitt}{I}
\newcommand{\qubitt}{\mathsf{Qubit}}
\newcommand{\qubitts}{\mathsf{Q}}
\newcommand{\bitt}{\mathsf{Bit}}
\newcommand{\bitts}{\mathsf{B}}
\newcommand{\natt}{\mathsf{Nat}}
\newcommand*{\vect}[2]{\mathsf{Vec} \; #1 \; #2}
\newcommand{\circjudgment}[3]{#1 : #2 \to #3}

\newcommand*{\configjudgment}[5]{#1 \vdash \config #2 #3 : #4; #5}
\newcommand*{\configjudgmentEff}[6]{#1 \vdash \config #2 #3 : #4; #5; #6}
\newcommand{\struct}[1]{\bar{#1}}

\newcommand{\varOne}{x} \newcommand{\varTwo}{y} 
  
\newcommand{\contextOne}{\Gamma}
\newcommand{\pcontextOne}{\Phi}
 
\newcommand{\termOne}{M} \newcommand{\termTwo}{N} \newcommand{\termThree}{P}
  
\newcommand{\valOne}{V} \newcommand{\valTwo}{W} \newcommand{\valThree}{X}  
\newcommand{\typeOne}{A} \newcommand{\typeTwo}{B} \newcommand{\typeThree}{C}
\newcommand{\ptypeOne}{P} \newcommand{\ptypeTwo}{R}
 
\newcommand{\emptycontext}{\emptyset}

\newcommand{\unitv}{*}
\newcommand{\tuple}[2]{\langle#1,#2\rangle}
\newcommand{\abs}[3]{\lambda #1_{#2}.#3}
\newcommand{\app}[2]{#1\,#2}
\newcommand{\liftoperator}{\operatorname{\mathsf{lift}}}
\newcommand{\lift}[1]{\liftoperator #1}
\newcommand{\force}[1]{\operatorname{\mathsf{force}} #1}
\newcommand{\boxoperator}{\operatorname{\mathsf{box}}}
\newcommand{\boxt}[2]{\boxoperator_{#1} #2}
\newcommand{\boxedCirc}[3]{(#1,#2,#3)}
\newcommand{\applyoperator}{\operatorname{\mathsf{apply}}}
\newcommand{\apply}[2]{\applyoperator(#1,#2)}
\newcommand{\letoperator}{\mathsf{let}}
\newcommand{\letin}[3]{\letoperator\; #1 = #2 \;\mathsf{in}\; #3}
\newcommand{\letinLessSpace}[3]{\letoperator \mkern2mu #1 = #2 \mkern2mu \mathsf{in}\, #3}
\newcommand{\dest}[4]{\letoperator\; \tuple{#1}{#2} = #3 \;\mathsf{in}\; #4}
\newcommand*{\ifzoperator}{\mathsf{ifz}}
\newcommand{\ifz}[3]{\ifzoperator\; #1 \; \mathsf{then} \; #2 \; \mathsf{else} \; #3}
\newcommand{\returnoperator}{\mathsf{return}}
\newcommand{\return}[1]{\returnoperator\; #1}

\newcommand{\lineararrow}{\multimap}
\newcommand*{\arrowst}[3]{#1 \multimap_{#3} #2}
\newcommand*{\arroweff}[4]{#1 \stackrel{#4}{\multimap_{#3}} #2}

\newcommand{\bang}[1]{\operatorname{!}#1}
\newcommand{\bangeff}[2]{\operatorname{!}^{#2}#1}
\newcommand{\circt}[3]{\operatorname{\mathsf{Circ}}^{#1}(#2,#3)}
\newcommand{\tensor}[2]{#1 \otimes #2}

\newcommand*{\cjudgst}[3]{#1 \vdash_c #2 : #3}
\newcommand*{\cjudgstLessSpace}[3]{#1 \mkern-1mu \vdash_c \mkern-1mu #2 \! :\! #3}
\newcommand*{\cjudgeff}[4]{#1 \vdash_c #2 : #3; #4}

\newcommand{\vjudgst}[3]{#1 \vdash_v #2 : #3}

\newcommand*{\permequiv}{\cong_{\sigma}}

\newcommand{\config}[2]{(#1,#2)}
\newcommand{\eval}{\Downarrow}
\newcommand{\freshlabelsfunction}{\operatorname{freshlabels}}
\newcommand{\freshlabels}[1]{\freshlabelsfunction(#1)}
\newcommand{\appendfunction}{\operatorname{append}}
\newcommand{\append}[3]{\appendfunction(#1,#2,#3)}
\newcommand{\sub}[2]{[#1/#2]}

\newcommand{\rcount}[1]{\#(#1)}

\newcommand*{\defeq}{\stackrel{\mathrm{def}}{=}}

\newcommand{\powerset}{\operatorname{\EuScript{P}}}

\newcommand{\PQ}{\textsf{Proto-Quipper}}
\newcommand{\PQM}{\textsf{Proto-Quipper-M}}
\newcommand{\PQC}{\textsf{Proto-Quipper-C}}

\newcommand{\PQS}{\textsf{Proto-Quipper-S}}
\newcommand{\PQD}{\textsf{Proto-Quipper-D}}

\newcommand{\PQDyn}{\textsf{Proto-Quipper-Dyn}}
\newcommand{\PQR}{\textsf{Proto-Quipper-R}}
\newcommand{\PQRA}{\textsf{Proto-Quipper-RA}}
\newcommand{\Qsharp}{\texttt{Q\#}}
\newcommand{\Quipper}{\texttt{Quipper}}
\newcommand{\Qiskit}{\texttt{Qiskit}}
\newcommand{\Cirq}{\texttt{Cirq}}
\newcommand{\Python}{\texttt{Python}}
\newcommand{\Haskell}{\texttt{Haskell}}
\newcommand*{\VQPL}{\textsf{VQPL}}
\newcommand*{\QWIRE}{\textsf{QWire}}
\newcommand*{\EWIRE}{\textsf{EWire}}

\makeatletter
\newcommand*\sem[1]{\llbracket #1 \rrbracket}
\newcommand*{\semM}[1]{\sem{#1}_{\mathcal{M}}}
\newcommand*{\semP}[1]{\sem{#1}_{P}}

\let\category=\mathcal
\newcommand\catA{\category{A}}
\newcommand\catC{\category{C}}

\newcommand\catE{\category{E}}
\newcommand*\catK{\category{K}}
\newcommand*{\catS}{\category S}
\newcommand*{\catM}{\category M}
\newcommand*{\catV}{\category V}
\newcommand*\Set{\mathbf{Set}}
\newcommand*\CPM{\mathbf{CPM}}
\newcommand*\Q{\mathbf{Q}}
\newcommand*{\Asrt}{\mathcal{A}\mkern-1mu\mathrm{srt}}

\DeclareFontFamily{U}{min}{}
\DeclareFontShape{U}{min}{m}{n}{<-> udmj30}{}
\newcommand\yo{\!\text{\usefont{U}{min}{m}{n}\symbol{'207}}\!}
\newcommand*\op{\mathrm{op}}
\newcommand*{\disc}{\mathbf{disc}}

\newcommand\obj{\mathbf{Obj}}
\newcommand*{\setObjOne}{X}
\newcommand*{\setObjTwo}{Y}
\newcommand*{\setObjThree}{Z}
\newcommand*{\mObjOne}{T}
\newcommand*{\mObjTwo}{U}
\newcommand*{\eObjOne}{\mathtt{t}}
\newcommand*{\eObjTwo}{\mathtt{u}}
\newcommand*{\eObjThree}{\mathtt{s}}
\newcommand*{\eObjUnit}{\mathtt{i}}

\newcommand*{\ltensor}{\ltimes}
\newcommand*{\rtensor}{\rtimes}
\newcommand*{\rtensorTwo}{\mathrlap{<}{\bigcirc}}
\newcommand*{\ltensorF}{\mathrlap{\,\ltensor}{\bigcirc}}
\newcommand*{\rtensorF}{\mathrlap{\,\rtensor}{\bigcirc}}

\newcommand*{\Fam}[1]{\mathbf{Fam}(#1)}

\newcommand*{\id}{\mathrm{id}}

\newcommand*{\unit}{\eta}
\newcommand*{\mult}{\mu}

\newcommand*{\Klarr}[3]{#1 \Rightarrow_{#3} #2}
\newcommand*{\unitp}{\iota}
\newcommand*{\strength}{\tau}
\newcommand*{\monad}{\mathcal T}
\newcommand*{\@circmonadSymb}{{\mathcal T}_{\catM}}
\newcommand*{\circmonadNoSt@r}[3]{\@circmonadSymb(#1, #2, #3)}
\newcommand*{\circmonadSt@r}{\@circmonadSymb}
\newcommand*{\circmonad}{\@ifstar\circmonadSt@r\circmonadNoSt@r}
\newcommand*{\symm}{s}
\newcommand*{\ev}{\mathbf{ev}}
\newcommand*{\applysem}{\mathbf{apply}}
\newcommand*{\functionSet}[2]{#1 \Rightarrow_{\scriptstyle \Set} #2}
\newcommand*{\boxsem}{\mathbf{box}}
\newcommand*{\dup}{\mathbf{dup}}
\newcommand*{\perm}{\mathbf{perm}}

\newcommand*{\Id}{\mathrm{Id}}

\newcommand*{\effOne}{e}
\newcommand*{\effTwo}{d}
\newcommand*{\nulleff}{\varepsilon}
\newcommand*{\abstraction}{\alpha}
\newcommand*{\refine}[3][]{#2 \rhd_{#1} #3}
\newcommand*{\catEff}{\tilde{\catE}}
\newcommand*{\circmonadEff}[1]{\@circmonadSymb^{#1}}
\newcommand*{\ple}{\lesssim}

\newcommand*{\MatSet}[3]{M_{#2, #3}(#1)}

\protected\def\unicodeotimes{\ensuremath{\otimes}}
\DeclareUnicodeCharacter{2297}{\unicodeotimes}
\makeatother

\begin{document}

\title[On Circuit Description Languages, Indexed Monads, and Resource Analysis]{On Circuit Description Languages, Indexed Monads,\texorpdfstring{\\}{} and Resource Analysis}

\author{Ken Sakayori}
\orcid{https://orcid.org/0000-0003-3238-9279}
\affiliation{%
  \institution{The University of Tokyo}
  \country{Japan}}

\author{Andrea Colledan}
\orcid{0000-0002-0049-0391}
\affiliation{%
  \institution{University of Bologna}
  \country{Italy}
}
\affiliation{%
  \institution{Centre Inria d'Université Côte d'Azur}
  \country{France}
}

\author{Ugo Dal Lago}
\orcid{0000-0001-9200-070X}
\affiliation{%
  \institution{University of Bologna}
  \country{Italy}
}
\affiliation{%
  \institution{Centre Inria d'Université Côte d'Azur}
  \country{France}
}

\begin{abstract}
In this paper, a monad-based denotational model is introduced and shown 
adequate for the \PQ\ family of calculi, themselves being idealized 
versions of the \Quipper\ programming language. The use of a monadic approach 
allows us to separate the \emph{value} to which a term reduces from the 
\emph{circuit} that the term itself produces as a side effect. In turn, this
enables the denotational interpretation and validation of rich type systems in 
which the size of the produced circuit can be controlled. Notably, the proposed 
semantic framework, through the novel concept of circuit algebra, suggests 
forms of effect typing guaranteeing quantitative properties about the 
resulting circuit, even in presence of optimizations.
\end{abstract}

\ccsdesc[500]{Theory of computation~Denotational semantics}
\ccsdesc[500]{Theory of computation~Program analysis}
\ccsdesc[300]{Software and its engineering~Domain specific languages}
\ccsdesc[300]{Hardware~Quantum computation}

\maketitle

\section{Introduction}
\label{sec:introduction}

Quantum computing promises to revolutionize various sub-fields within computer 
science by solving complex problems exponentially faster than classical 
computing~\cite{Shor94}. This technology leverages the concept of quantum bits 
(or qubits), a unit of information whose dynamics is governed by the rules of 
quantum mechanics, thus enabling superposition and entanglement, keys to the 
aforementioned speedup. To harness this potential, several programming 
languages have been developed specifically for quantum computing, such as 
\Qsharp~\cite{SvoreGTAGHR18}, \Qiskit~\cite{JavadiAbhariTKWLGMNBCJG24}, and \Cirq~[\citeyear{Circ25}]. In turn, fields like program verification have
adapted well-known techniques like abstract interpretation~\cite{Perdrix08,YuPalsberg21}, type systems~\cite{AmyTS19,ColledanDalLago25}, and
Hoare logic~\cite{YingYW17,LiuZWYLLYZ19,ZhouYY19} to these languages.

We can identify at least two ways to design a quantum programming language. On the one hand, we could simply allow programs written in traditional programming languages to access not only \emph{classical} data but also \emph{quantum} data. The latter cannot be used the same way as the former, and is rather supported by specific \emph{initialization}, \emph{modification}, and \emph{reading} operations. As an example, the reading of a qubit value, often called a \emph{measurement}, can alter its value and has, in general, a probabilistic outcome, thus being substantially different from the corresponding classical operation. In programming languages of this kind, the quantum data are assumed to be stored in an external device accessible interactively through the aforementioned operations. This model, often indicated with the acronym QRAM~\cite{Knill22}, is adopted by a multitude of proposals in the literature (see e.g.~\cite{Selinger04,SandersZ00,BettelliCS03,SelingerValiron05}).

In theory, QRAM languages are the natural adaptation of classical programming languages to the quantum world. In practice, however, quantum hardware architectures can hardly be programmed \emph{interactively}: not only is the number of qubits available very small, but the \emph{time} within which computation must be completed should itself be minimized, given that the useful lifespan of a qubit is short. Consequently, quantum architectures typically take as input a whole quantum circuit, i.e. a precise description of \emph{all} the necessary qubits and the operations to be performed on them. This circuit must therefore be available in its entirety, preferably already subjected to an aggressive optimization process. In such a context, so-called \emph{circuit description languages} (CDLs for short) are to be preferred, and most mainstream languages in the field, including \Qiskit\ and \Cirq, are of this nature. CDLs are high-level languages used to describe and generate circuits, a quintessential example being the quantum circuit. Circuits are typically seen like any other ordinary data structure, with specific operations on them available through, e.g., methods or subroutines. Directly manipulating circuits from within a classical program offers the advantage of having more direct control over their shape and size. This is crucial given the state of quantum hardware architectures today, which provide a \emph{limited number} of \emph{error-prone} qubits, and for which \emph{not all operations} can be implemented at the same cost.

A peculiar circuit description language is \Quipper{}~\cite{GreenLRSV13}. In \Quipper, circuits are not just like any other data structure. Rather, they are seen as the by-product of certain effect-producing computations that modify some underlying quantum circuit when executed.
This, combined with the presence of higher-order functions and operations meant to turn any term (of an appropriate type) into a circuit, makes \Quipper{} a very powerful and flexible idiom. Its metatheory has been the subject of quite some investigations by the programming language community in the last years, with contributions ranging from advanced type systems~\cite{FuKRS20,FuKS22} to fancy features like dynamic lifting~\cite{FuKRS22,FuKRS23,ColledanDalLago23} to denotational semantics~\cite{RiosSelinger17,LindenhoviusMZ18,FuKS22,FuKRS23,FuKRS22}.
This last aspect of \Quipper, in particular, has been studied by providing semantic models for some of the languages of the so-called \PQ\ family, which includes various calculi, such as \PQM~\cite{RiosSelinger17}, \PQD~\cite{FuKS22}, \PQDyn~\cite{FuKRS23}, etc. In these cases, such semantics are built around concepts such as that of a presheaf and turn out to take the shape of a LNL model~\cite{Benton94}. A by-product of the use of presheaves is that the interpretation of the term and the underlying circuit are somehow \emph{merged} into a single mathematical object. As a result, it is difficult to read interesting features of the underlying circuit from the interpretation of a term or closure: what the circuit does and what the program does to produce the circuit are inextricably coupled.

This coupling, in turn, prevents those models from adequately accounting for variants of the \PQ\ family that are specifically designed to control the shape of the produced circuits, and more specifically to derive upper bounds on the size of the latter~\cite{ColledanDalLago24a,ColledanDalLago25}.
The correctness of these systems has been proved by purely operational means, and a denotational semantics for them is still missing. This ultimately makes such systems somewhat rigid and complicates their definition.

The aim of this paper is precisely to give a denotational semantics to languages in the \PQ\ family in which the interpretation of terms is \emph{kept separate} from that of the produced circuit. This is achieved by seeing circuit building as an indexed monad~\cite{Atkey09}.\footnote{Indexed monads are also known as and were originally called parameterized monads. We avoid this name since other ``parameters'', such as parameters of circuits or grades of a monad, appear in this work.} Remarkably, this new point of view allows us to give semantics to languages such as Colledan and Dal Lago's \PQR, and even allows us to justify some of their peculiarities. The introduced semantic framework suggests a natural way to unify so-called \emph{local} and \emph{global} circuit metrics, at the same time allowing the definition of metrics that go substantially beyond those proposed by Colledan and Dal Lago, in particular accounting for simple forms of circuit optimization.

The contributions of this paper can be thus summarized as follows:
\begin{itemize}
\item 
    First, we give a simple type system for \PQ. The introduced system is a slight variation on the theme of \PQM~\cite{RiosSelinger17}, whereas the input to the circuit produced by each effectful functional term needs to be exposed in its type and thus becomes an integral part of the arrow type, turning it into a \emph{closure type}. This change, as we will explain in the next section, seems inevitable since, without it, it would not be possible to know even the nature, i.e. the type, of the circuit produced by the term in question. Noticeably, closure types are present in \citeauthor{ColledanDalLago25}'s~[\citeyear{ColledanDalLago25}] most recent contribution.
    We call this calculus with closure type \PQC{}.
\item
    We then show that Atkey's indexed monad is an appropriate framework for giving denotational semantics to \PQC.
    By treating circuits as (pre)monoidal morphisms and considering the category-action indexed monad---a many-sorted generalization of the writer monad---we maintain a clear separation between the value a term evaluates to and the circuit produced alongside the evaluation.
    \PQC{} is interpreted in the parameterized Freyd category induced by this indexed monad, making explicit how ``parameters'', computations, and circuits interact as these three notions are interpreted in different categories.
    The semantics is proved both sound and computationally adequate.
\item
    We then move to richer type systems, where simple types are enriched by a form of effect typing. The model based on indexed monads remains adequate, and suggests an abstract notion of circuit algebra, through which it is possible to capture various circuit metrics (including all those considered by Dal Lago and Colledan), but also new forms of metrics induced by assertion-based circuit optimization schemes~\cite{HanerHT20}.
  \item We briefly discuss how dependent types might be incorporated into both the syntax and semantics of the variant of \PQC{}  (without effect typing) introduced earlier.
    On the semantic side, this is achieved by applying the families construction to the denotational model of \PQC{} in the spirit of fibered adjunction models~\cite{AhmanGP16}.
\end{itemize}

The rest of this paper is structured as follows. After the next section, which serves to frame the problem without going into the technical details, we move on to Section~\ref{sec:simple-type}, in which \PQC{}, i.e.\ a slight variation of the \PQM\ calculus, is introduced and endowed with an operational semantics. In Section~\ref{sec:categorica-semantics-st}, then, a monadic denotational semantics for \PQC{} is introduced and proved adequate. In Section~\ref{sec:effect-system}, an extension of these results to a calculus with effect typing, along the lines of \PQR, is presented.
Section~\ref{sec:dependent-type} discusses the possibility of extending our framework to incorporate dependent types.
Section~\ref{sec:related-work} discusses related work, and Section~\ref{sec:conclusion} concludes the paper.

\section{A Monadic Semantics for CDLs: Why and How?}
In this section, we will describe in a little more detail the problem of giving a monadic denotational semantics to CDLs, focusing on \emph{how} this can be done, but also on \emph{why} such an effort might be worth it.

Typically, a program written in a CDL is just a program in a mainstream programming language (e.g. \Python, \Haskell, or dialects thereof) whose purpose is that of facilitating the construction of (quantum) circuits which, once built, can then be sent to quantum hardware for their evaluation, or merely simulated through high-performance classical hardware. As already mentioned, the CDL we are mainly concerned with is \Quipper, which is embedded in \Haskell.

A program written in any CDL, and particularly in \Quipper, does not describe a \emph{single} circuit but a \emph{family} of circuits, depending on some parameters, e.g.~a number \( n \) representing the number of input qubits, or, in the case of Shor's algorithm, the size of the natural number to be factored. The ability to describe families of circuits enables \Quipper\ to succinctly and elegantly describe quantum algorithms, thus having a pragmatic impact and attracting the attention of the programming language community~\cite{RiosSelinger17,LindenhoviusMZ18,FuKS22,FuKRS23}, which proposed idealized languages capturing the essence of \Quipper. Such formal calculi come equipped with an operational semantics, type system, and often with a denotational semantics.

Most forms of denotational semantics for the \PQ\ family are based on presheaves, and enjoy a \emph{constructive property}, which states that the interpretation of a judgment in a certain form is indeed a parameterized family of circuits. For example, in the categorical semantics of \PQM{}~\cite{RiosSelinger17}, the judgment \( n : \textsf{nat}, x : \qubitt \vdash \termOne : \qubitt \) is interpreted as a function \(\sem{\termOne} \colon \mathbb{N} \to \catM(\qubitt, \qubitt) \), where \( \catM(\qubitt, \qubitt) \) is the set of circuits whose input and output interfaces both consist of a single qubit. In the above, the type \textsf{nat} can be replaced by any ``classical data type'' and \( \qubitt \) can be replaced by any ``wire type''.
However, some judgments cannot be interpreted as a family of circuits in the same way, including terms with free variables with a function type. As an example, a term of type 
\begin{equation}\label{eq:exampletype}
	n : \textsf{nat}, x : (\qubitt\multimap\qubitt) \vdash \termTwo :  \qubitt
\end{equation}
which evaluates to a qubit type value while generating a circuit, would be interpreted as \( \sem{\termTwo} \colon \mathbb N \to \mathrm{Nat}(\sem{\qubitt \to \qubitt}, \sem{\qubitt})\) where \( \alpha \in \mathrm{Nat}(\sem{\qubitt \to \qubitt}, \sem{\qubitt})\) is a natural transformation (i.e.~a morphism in the presheaf category).
Each component \( \alpha_\mtypeOne \) has a type \( \sem{\qubitt \to \qubitt}(\mtypeOne) \to \catM(\mtypeOne, \qubitt) \) because \( \sem{\qubitt} \) is defined via the Yoneda embedding \( \yo \colon \catM \to [\catM^{\op}, \Set] \) and \( \sem{\qubitt}(T) = \yo(\qubitt)(T) = \catM(\mtypeOne, \qubitt) \).
In other words, and more informally, \( \termTwo \) is interpreted as a family of polymorphic functions
\(\{f_i\}_{i\in\mathbb{N}}\), each of them having type
\[
\forall (T \colon \mathit{WireType}).\,\mathit{Clos}(\qubitt,\qubitt)[T]\rightarrow
\catM(T, \qubitt)
\]
The type \(\mathit{Clos}(\qubitt,\qubitt)[T]\) represents a \emph{closure type} disclosing the type \(T\) of the data that the closure captures. Given a ``closure'' as input, \(f_i\) returns a circuit whose input and output interfaces are \(T\) and \(\qubitt\), respectively.
In a sense, then, not even the \emph{input type} of the generated circuit can be read from \( \sem{\termTwo}(n)\) because we need to know the \emph{actual data} that will be passed to \( x \) to determine the interface of the circuit. This implies that modular reasoning about the produced circuits cannot be easily performed \emph{within} the model, in which it would be hard to interpret type systems specifically built for intensional analysis~\cite{ColledanDalLago24a}.

This paper introduces a denotational semantics for a CDL in which \emph{every} judgment is interpreted as a family of circuits whose types are uniquely defined. More specifically, any judgment \( \contextOne \vdash \termOne : \typeOne \) is interpreted as a function
\begin{equation}
  \sem{\termOne} \colon \sem{\flat \contextOne} \to \sem {\flat \typeOne} \times \catM(\sem{\sharp \contextOne}, \sem{\sharp \typeOne}),\label{eq:intro:interpretation}
\end{equation}
where \( \sharp \) and \( \flat \) are operations extracting the ``circuit part'' and ``parametric part'' of any type, respectively.
The mathematical object $\sem{\termOne}$, in other words, is a family of pairs \( \{(v_i, \circuitOne_i)\}_{i \in \sem{\flat \contextOne}} \) indexed by \( \sem{\flat \contextOne}\) where \( v_i \) is a ``value'' in \( \sem {\flat \typeOne} \) and \( \circuitOne_i \) is a circuit in  \( \catM(\sem{\sharp \contextOne}, \sem{\sharp \typeOne}) \); a family of circuits is a special case where \( v_i \) is a unit value.
The fact that \emph{every} judgment is interpreted as a family of circuits is important since it allows us to \emph{compositionally} reason about the family of circuits generated by a program. For example, upper bounds to the width of the circuits generated by any program can be computed by looking at the interpretation of the subprograms.

The just sketched construction evidently has the form of a monad~\cite{Moggi91}, structurally very similar to a writer monad. There is, however, one important difference: the type of circuit produced during the execution of a term is \emph{not} fixed a priori. As a consequence, the classical notion of monad, being somehow monomorphic, cannot be applied directly. Instead, indexed monads~\cite{Atkey09} can be fruitfully employed to model the circuit generated by the underlying program.
In fact, the interpretation~\eqref{eq:intro:interpretation} can be rewritten as
\begin{equation}
   \sem{\termOne} \colon \sem{\flat \contextOne} \to \monad(\sem{\sharp \contextOne}, \sem{\sharp \typeOne}, \sem {\flat \typeOne}),
   \label{eq:intro:monadicinterpretation}
 \end{equation}
 where \( \monad(\mObjOne, \mObjTwo, \setObjOne)  \) is an indexed monad defined as \( \setObjOne \times \catM(\mObjOne, \mObjTwo)\).
 Using such a monadic approach allows us to structure the interpretation of languages in the \PQ\ family in a new way, fundamentally different from that considered in the literature on the subject: every term is interpreted as a mathematical object in which the value produced and the underlying circuit are kept separate.
 Technically, we interpret terms in a parameterized Freyd category obtained by applying an indexed version of the well-known Kleisli construction to the indexed monad above.

If we look at Equation (\ref{eq:intro:monadicinterpretation}) in more detail, we soon realize that a monadic interpretation like the one we are discussing requires knowing $\sem{\sharp \contextOne}$ whenever the effectful term $M$ is a value $\lambda x.M$ (where $x$ is any of the variables in $\Gamma$). In other words, the ``circuit portion'' of $\Gamma$ must become part of the (functional) type of $\lambda x.M$. This last observation justifies the small discrepancies between \PQC{} (the language we present in Section~\ref{sec:simple-type} below) and \PQM\ and provides a denotational reading to some of the type-theoretical tricks in~\cite{ColledanDalLago24a}. As an example, the
typing judgment \eqref{eq:exampletype} becomes
\(n : \textsf{nat}, x : (\arrowst{\qubitt}{\qubitt}{\mtypeOne} ) \vdash \termTwo :  \qubitt\),
where $T$ is the type of circuit variables the argument function captures.

The advantage of moving to a monadic view like the one just described is that we have now \emph{exposed} the space \( \catM\) of circuits in the interpretation. As we will see in Section~\ref{sec:effect-system}, in fact, this naturally suggests a way to control the size of the generated circuits through a form of effect typing: any well-behaved functor from \( \catM \) to a category \( \catE \) which captures the relevant characteristics of the underlying circuit, e.g. its size, induces a form of effect typing which is sound by construction.
This is what denotational semantics is good for: not only is the programming language in question interpreted compositionally, but the interpretation naturally suggests what in the language might be subject to modification or adaptation while, at the same time, indicating what the underlying axiomatics should be, and factoring out most proofs.

\section{Simple Types}
\label{sec:simple-type}
In this section, we introduce the syntax and operational semantics of \PQC{}, our dialect of \PQM{}~\cite{RiosSelinger17}. We will point out the differences with the language
\PQM{}, still trying to keep the presentation as self-contained as possible.

\subsection{Type and Syntax}
\label{sec:type-and-syntax}
We first introduce the simple type system of \PQC. The grammar for types is
as follows:
\begin{align*}
  \text{Types} \quad	\typeOne,\typeTwo
  &\Coloneqq \ptypeOne \mid \mtypeOne \mid \tensor{\typeOne}{\typeTwo} \mid \arrowst \typeOne \typeTwo \mtypeOne \\
  \text{Parameter Types}	\quad \ptypeOne,\ptypeTwo
  &\Coloneqq \unitt \mid \natt \mid \bang{\typeOne} \mid \tensor{\ptypeOne}{\ptypeTwo} \mid \circt{}{\mtypeOne}{\mtypeTwo} \\
  \text{Bundle Types} \quad  \mtypeOne,\mtypeTwo
  &\Coloneqq \munitt \mid \wtypeOne \mid \tensor{\mtypeOne}{\mtypeTwo}
\end{align*}
There are three kinds of types. In addition to generic types, which are 
intended for (not necessarily duplicable) terms, there are 
\emph{parameter} types whose inhabitants are freely duplicable and which 
include circuits and values of type $!A$.
There are the unit and natural number type as base types, but the specific choice of base types is not that important.
We also need \emph{bundle}
types, which are used to give a type to circuit wires. Observe that the
tensor product operator $\otimes$ is available in the three kinds of types, 
while the construction of functions is not available in parameter and bundle types.

Typing is almost the same as in \PQM{} as defined in~\cite{RiosSelinger17}. 
There is a significant difference in the definition of the arrow type, however.
We annotate the arrow type with a bundle type \( \mtypeOne \), which describes 
the types of the free variables captured by the function, effectively turning
it into a form of \emph{closure type}.
Note that we only care about the bundle type variables that are captured by the 
function and drop the information of variables with parameter types.
A similar kind of annotation was used by~\citet{ColledanDalLago24a}, although in their
work the label was rather a natural number abstracting the type \( \mtypeOne \). We have already
argued about the need for this change to the type system, and will come back to 
that in the next section. We might write \(\typeOne \lineararrow \typeTwo  \) 
for \( \arrowst \typeOne \typeTwo \munitt \) for the sake of
simplifying the notation.

Another difference compared to \PQM{} (or \PQR{}, as defined 
in~\cite{ColledanDalLago24a}) is that we do not have a type for lists.
We removed lists as they cannot be directly interpreted in our semantic model.
However, as we shall see in Section~\ref{sec:dependent-type}, we can extend our language with a vector type (i.e.~a list with specified length).

The syntax of \PQC{} terms is defined as follows. We use a fine-grained
call-by-value style syntax~\cite{LevyPT03} as done by~\citet{ColledanDalLago24a}.\footnote{It is easy to extend the language with effect-free constants such as arithmetic operations or or meta-operations on circuits as in~\cite{RiosSelinger17}, but we we omit these, since they are immaterial to our main semantic results.}
\begin{align*}
  \text{Terms} \quad \termOne,\termTwo
  &\Coloneqq \app{\valOne}{\valTwo} \mid 
    \dest{\varOne}{\varTwo}{\valOne}{\termOne} \mid
    \ifz \valOne \termOne \termTwo
  \\
  &\,\, \mid \force{\valOne} \mid \boxt{\mtypeOne}{\valOne} 
   \mid \apply{\valOne}{\valTwo} \\
  &\,\,\mid \return{\valOne} \mid \letin{\varOne}{\termOne}{\termTwo} \displaybreak[1] \\
  \text{Values} \quad \valOne,\valTwo
  &\Coloneqq \unitv \mid n \mid \varOne \mid \labOne \mid
  \abs{\varOne}{\typeOne}{\termOne} \mid \lift{\termOne} \mid \boxedCirc{\struct\labOne}{\circuitOne}{\struct{\labTwo}} \mid
  \tuple{\valOne}{\valTwo}\\
  \text{Wire Bundles} \quad \struct\labOne,\struct\labTwo &\Coloneqq \unitv \mid \labOne \mid \tuple{\struct{\labOne}}{\struct{\labTwo}}.
\end{align*}
The informal behavior for terms is in line with that of \PQM, which adds to
the usual constructs of a call-by-value linear lambda-calculus 
(abstractions, applications, let bindings, pairs, etc.) specific 
operators for circuit manipulation:
\begin{itemize}
	\item $\labOne$ is a \emph{label}, that is, a pointer to a wire in the underlying circuit.
	\item $\boxedCirc{\struct\labOne}{\circuitOne}{\struct\labTwo}$ is a \emph{boxed circuit} and represents the circuit $\circuitOne$ as a value in the language. Wire bundles $\struct\labOne$ and $\struct\labTwo$ represent the input and output interfaces of $\circuitOne$, respectively.
	\item $\apply{\valOne}{\valTwo}$ appends a boxed circuit $\valOne$ to the wires identified by $\valTwo$ among the outputs of the underlying circuit.
	\item $\boxt{\mtypeOne}{\valOne}$ evaluates a circuit building function $\valOne$ in isolation and returns the result as a boxed circuit.
\end{itemize}

At this point, we should make it clear what we mean by a ``circuit''.
We do not fix what a circuit is (except when considering concrete examples), as is often the case for \PQ{} calculi.
Usually, when giving a semantics to \PQ{}, circuits are seen as morphisms in a \emph{monoidal category} \( \catM \) and the semantics is parametric to the choice of \( \catM \).
In this work, we assume that the category \( \catM \) of circuits is a \emph{premonoidal} category.
Roughly, a premonoidal category is a monoidal category without the interchange law
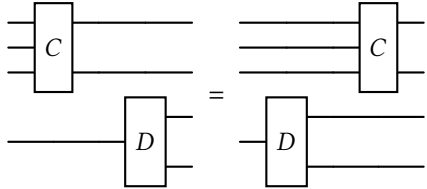
\[
  \begin{quantikz}[column sep=1em,row sep=0.2em]
    & \gate[3]{\circuitOne} & & & \\
    &  &\setwiretype{n} & & \\
    &  & & &\\
    \setwiretype{n} & & & \gate[3]{\circuitTwo} & \setwiretype{q} \\
    &  & & &\setwiretype{n} \\
    \setwiretype{n} & & & & \setwiretype{q}
  \end{quantikz}
  =
    \begin{quantikz}[column sep=1em,row sep=0.2em]
    & & & \gate[3]{\circuitOne}  & \\
    & & & &\setwiretype{n}  \\
    & & & & \\
    \setwiretype{n}& \gate[3]{\circuitTwo} & \setwiretype{q} & & \\
    & & \setwiretype{n} & & \\
    \setwiretype{n} & & \setwiretype{q} & &
  \end{quantikz}
\]
which is too strong in a cost sensitive scenario.
(For example, we may say that the two circuits above are different because the width of the circuit on the left-hand-side is \( 4 \) whereas the width of the circuit on the right-hand-side is \( 5 \).)

\begin{definition}[Premonoidal Category~\cite{PowerRobinson97}]
  A \emph{binoidal category} is a category \( \catA \) equipped with, for each \( a \in \obj(\catA) \), endofunctors \( a \rtensor - \) and \( - \ltensor a \) from \( \catA \) to \( \catA \) such that \( a \ltensor  b = a \rtensor b  \defeq a \otimes b\) for every pair of objects \( (a, b) \) of \( \catA \).
  A morphism \( f \colon a \to b \) in \( \catA \) is \emph{central} if it interchanges with any morphism \( g \colon a' \to b'\): \( (f \ltensor a'); (b \rtensor g) = (a \rtensor g); (f \ltensor b') \) and \( (a' \rtensor f); (g \ltensor b) = (g \ltensor a); (b' \rtensor f) \).\footnotemark
  \footnotetext{We use diagrammatic order for compositions in this paper.}
  In case two composites agree, we write \( f \otimes g\) and \( g \otimes f \), respectively.
  A \emph{premonoidal category} is a binoidal category equipped with an object \( I \), the ``monoidal unit'',  central natural (separately at each given component) isomorphisms for associativity and right and left units satisfying the standard pentagon and triangle equations.
  A premonoidal category is \emph{strict} if all the coherence morphisms are identities, and is \emph{symmetric} if it has a central isomorphism \( a \otimes b \cong b \otimes a\) that is natural (at each component) and satisfies the usual axioms of a symmetry.
\end{definition}

Typically, the category of circuits \( \catM \) is defined by giving a syntactic description of circuits.
In other words, it is the free category generated by some collection of base types and gates.
While we do not fix the category  \( \catM \), we assume that \( \catM \) has some distinguished objects \( \qubitts \) and \( \bitts \) that are used to interpret qubits and classical bits, respectively.
We also assume that \( \catM \) is small and a \emph{strict} symmetric premonoidal category.\footnotemark
\footnotetext{Every premonoidal category is equivalent to a strict one since the coherence theorem holds~\cite{PowerRobinson97}.}

A typing judgment for terms is of the form \( \cjudgst \contextOne \termOne \typeOne \),
and intuitively means that \( \termOne \) is well-typed under the typing context \( \contextOne \); similarly, we have a typing judgment for values of the form \( \vjudgst \contextOne \valOne \typeOne \).
A \emph{typing context} \(\contextOne\) is a finite sequence of bindings each 
of which is either of the form \( x : \typeOne \) or \( \labOne : \wtypeOne \).
The reason for using finite sequences rather than finite sets is to simplify the definition of the semantics; if two contexts \( \contextOne_1 \) and \( \contextOne_2\) are equal up to permutation, then we write \( \contextOne_1 \permequiv \contextOne_2 \).
We use metavariables \( \lvOne, \lvTwo, \ldots \) to denote variables or labels.
A typing context is a \emph{label context} if it is of the form \( \labOne_1 : \wtypeOne_1, \ldots, \labOne_n : \wtypeOne_n  \).
Label contexts are denoted by \( \lcOne, \lcTwo \).
A \emph{parameter context}, written \( \pcontextOne \), is a typing context that only contains variables with  parameter types.

Typing rules are in Figure~\ref{fig:st-typing-rules}, and most of them are 
self-explanatory.
In the typing rules, when we write \( \pcontextOne, \contextOne \), we stipulate that \( \contextOne \) does not contain any variable with a parameter type.
Rules \textit{circ}, \textit{box} and \textit{apply} are specific to a CDL, and thus warrant discussion.
A boxed circuit \( \boxedCirc{\struct\labOne}{\circuitOne}{\struct\labTwo} \) is well-typed if the labels \( \struct \labOne\) and \( \struct \labTwo \) acting as language-level interfaces to \( \circuitOne \) have types that match with the (co)domain of \( \circuitOne \).
The notation \( \circuitOne \colon \lcOne \to \lcTwo \) means that \( \circuitOne \) is a morphism from \( \sem{\lcOne} \) to \( \sem{\lcTwo} \) in \( \catM \);
\( \sem{\lcOne} \) is the obvious
interpretation, which we formally define in Section~\ref{sec:categorica-semantics-st}.
The \emph{box} rule says that if \( \valOne \) is a circuit building function that, once applied to an input of type \( \mtypeOne \), builds a circuit
of output type \( \mtypeTwo \), then \( \valOne \) can be turned into a circuit whose interface has types \( \mtypeOne \) and \( \mtypeTwo \).
Note that the \textit{box} rule requires the typing context to be \( \pcontextOne \) so as to ensure that the function is not capturing any variable with a bundle type.
The rule \textit{apply}, on the other hand, can be read as a special version of the typing rule for function application.

Once again, the main novelty with respect to \PQM{} has to do with the arrow 
type. The operation \( \sharp \) used in the \textit{abs} rule extracts a
bundle type from any type \( A \).\footnotemark
\footnotetext{A similarly looking operation was called \emph{wire count} 
in~\cite{ColledanDalLago24a}.
We do not use this name since we extract the whole type and not just a natural 
number counting ``how many wires are used''.}
Formally, the operation \( \sharp \) is inductively defined as follows:
\begin{gather*}
  \sharp(\ptypeOne) \defeq \munitt  \qquad \sharp(\munitt) \defeq \munitt \qquad \sharp(\wtypeOne) \defeq \wtypeOne \qquad \sharp(\arrowst \typeOne \typeTwo \mtypeOne) \defeq \mtypeOne \\
  \sharp(\tensor \ptypeOne \ptypeTwo) \defeq \tensor {\sharp(\ptypeOne)} {\sharp(\ptypeTwo)} \qquad \sharp(\tensor \mtypeOne \mtypeTwo) \defeq \tensor {\sharp(\mtypeOne)} {\sharp(\mtypeTwo)}.
\end{gather*}
We let the operator \( \sharp \) act on
typing contexts by \( \sharp(\lvOne_1 : \typeOne_1, \ldots \lvOne_n : \typeOne 
_n ) \defeq \sharp(\typeOne_1) \otimes \cdots \otimes \sharp(\typeOne_n) \).
It should now be clear that the rule \textit{abs} does nothing more than 
inserting the \emph{bundle} type of the variables free in the abstraction we 
are typing into its arrow type.
\begin{figure}[t]
	\centering
	\fbox{
    \begin{minipage}{0.97\linewidth}
    \begin{mathpar}
			\inference[\textit{unit}]
			{ }
			{\vjudgst{\pcontextOne}{\unitv}{\unitt}}
			\and
      \inference[\textit{nat}]
			{ }
			{\vjudgst{\pcontextOne}{n}{\natt}}
      \and
			\inference[\textit{lab}]
			{ }
			{\vjudgst{\pcontextOne, \labOne:\wtypeOne}{\labOne}{\wtypeOne}}
			\and
			\inference[\textit{var}]
			{ }
			{\vjudgst{\pcontextOne,\varOne:\typeOne}{\varOne}{\typeOne}}
			\and
			\inference[\textit{abs}]
			{\cjudgst{\contextOne,\varOne:\typeOne}{\termOne}{\typeTwo}}
			{\vjudgst{\contextOne}{\abs{\varOne}{\typeOne}{\termOne}}{\arrowst{\typeOne}{\typeTwo}{\rcount{\contextOne}}}}
			\and
			\inference[\textit{app}]
			{\vjudgst{\pcontextOne,\contextOne_1}{\valOne}{\arrowst{\typeOne}{\typeTwo}{\mtypeOne}}
				&
				\vjudgst{\pcontextOne,\contextOne_2}{\valTwo}{\typeOne}}
			{\cjudgst{\pcontextOne,\contextOne_1,\contextOne_2}{\app{\valOne}{\valTwo}}{\typeTwo}}
			\and
			\inference[\textit{lift}]
			{\cjudgst{\pcontextOne}{\termOne}{\typeOne}}
			{\vjudgst{\pcontextOne}{\lift{\termOne}}{\bang{\typeOne}}}
			\and
			\inference[\textit{force}]
			{\vjudgst{\pcontextOne}{\valOne}{\bang{\typeOne}}}
			{\cjudgst{\pcontextOne}{\force{\valOne}}{\typeOne}}
			\and
			\inference[\textit{circ}]
			{\circjudgment{\circuitOne}{\lcOne}{\lcTwo}
        &
        \lcOne \permequiv \lcOne'
        &
        \lcTwo \permequiv \lcTwo'
				\\\
				\vjudgst{\lcOne'}{\struct\labOne}{\mtypeOne}
				&
				\vjudgst{\lcTwo'}{\struct\labTwo}{\mtypeTwo}
				} {\vjudgst{\pcontextOne}{\boxedCirc{\struct\labOne}{\circuitOne}{\struct\labTwo}}{\circt{}{\mtypeOne}{\mtypeTwo}}}
			\and
			\inference[\textit{box}]
			{\vjudgst{\pcontextOne}{\valOne}{{\arrowst{\mtypeOne}{\mtypeTwo}{I}}}}
			{\cjudgst{\pcontextOne}{\boxt{\mtypeOne}{\valOne}}{\circt{}{\mtypeOne}{\mtypeTwo}}}
			\and
			\inference[\textit{apply}]
			{\vjudgst{\pcontextOne,\contextOne_1}{\valOne}{\circt{}{\mtypeOne}{\mtypeTwo}}
				&
				\vjudgst{\pcontextOne,\contextOne_2}{\valTwo}{\mtypeOne}}
			{\cjudgst{\pcontextOne,\contextOne_1,\contextOne_2}{\apply{\valOne}{\valTwo}}{\mtypeTwo}}
			\and
			\inference[\textit{dest}]
			{\vjudgst{\pcontextOne,\contextOne_1}{\valOne}{\tensor{\typeOne}{\typeTwo}}
				\\\
				\cjudgst{\pcontextOne,\contextOne_2,\varOne:\typeOne,\varTwo:\typeTwo}{\termOne}{\typeThree}}
			{\cjudgst{\pcontextOne,\contextOne_2,\contextOne_1}{\dest{\varOne}{\varTwo}{\valOne}{\termOne}}{\typeThree}}
      \and
      \inference[\textit{ifz}]
			{\vjudgst{\pcontextOne}{\valOne}{\natt}
				\\\
				\cjudgst{\pcontextOne,\contextOne}{\termOne}{\typeOne}
        \and
        \cjudgst{\pcontextOne,\contextOne}{\termTwo}{\typeOne}
      }
			{\cjudgst{\pcontextOne,\contextOne}{\ifz{\valOne}{\termOne}{\termTwo}}{\typeOne}}

			\and
			\inference[\textit{pair}]
			{\vjudgst{\pcontextOne,\contextOne_1}{\valOne}{\typeOne}
				&
				\vjudgst{\pcontextOne,\contextOne_2}{\valTwo}{\typeTwo}}
			{\vjudgst{\pcontextOne,\contextOne_1,\contextOne_2}{\tuple{\valOne}{\valTwo}}{\tensor{\typeOne}{\typeTwo}}}
			\and
			\inference[\textit{return}]
			{\vjudgst{\contextOne}{\valOne}{\typeOne}}
			{\cjudgst{\contextOne}{\return{\valOne}}{\typeOne}}
			\and
			\inference[\textit{let}]
			{\cjudgst{\pcontextOne,\contextOne_1}{\termOne}{\typeOne}
				&
				\cjudgst{\pcontextOne,\contextOne_2,\varOne:\typeOne}{\termTwo}{\typeTwo}}
			{\cjudgst{\pcontextOne,\contextOne_2,\contextOne_1}{\letin{\varOne}{\termOne}{\termTwo}}{\typeTwo}} \\
			\inference[\textit{ex}]
			{\cjudgst{\contextOne_1, \lvOne : \typeOne, \lvTwo : \typeTwo, \contextOne_2}{\termOne}{\typeThree}}
			{\cjudgst{\contextOne_1, \lvTwo : \typeTwo, \lvOne : \typeOne, \contextOne_2}{\termOne}{\typeThree}}
    \end{mathpar}
    \end{minipage}}
	\caption{Typing Rules for \PQC{}.}
	\label{fig:st-typing-rules}
\end{figure}

The \textit{box} rule is slightly different from the one commonly seen in the 
\PQ{} family, in which \( \boxt{} \) is a coercion from \( \bang(\mtypeOne
\lineararrow \mtypeTwo) \) to \( \circt{} \mtypeOne \mtypeTwo \).
In our rule, instead, we drop the \( \bang \) operator.
Intuitively, \( \bang \) is needed to ensure that the function with type \( 
\mtypeOne \lineararrow \mtypeTwo \) does not capture any variable with bundle 
type. But this information is already explicit in our type system by the 
subscript \( I \), and there is no reason to additionally require the of-course 
modality. We will later also give a semantic explanation against this design 
choice (see Proposition~\ref{prop:box-as-iso} and the remark after it).

\subsection{Operational Semantics}
\label{sec:operational-semantics}
The operational semantics is defined as a big-step evaluation relation on \emph{configurations}.
A configuration is a pair \( \config \circuitOne \termOne \), where \( 
\circuitOne \) is a circuit being generated and \( \termOne \) is the term 
being evaluated.
The definition of the big-step evaluation relation \( \eval \) is in 
Figure~\ref{fig:operational-semantics}. (The rule for evaluating the else branch of \( \ifzoperator \) is omitted.)

The \textit{box} rule relies on the \emph{freshlabels} function, which is used to produce a fresh label context $\lcOne$ and a wire bundle $\struct\labOne$ such that $\vjudgst{\lcOne}{\struct\labOne}{\mtypeOne}$.
On the other hand, the \emph{apply} rule relies on the \emph{append} function, which attaches the circuit $\circuitTwo$ to the wires identified by $\struct\labThree$ among the outputs of $\circuitOne$.
This operation often requires a renaming of the labels in $\circuitTwo$, so that its input interface $\struct\labOne$ matches $\struct\labThree$. More formally, we say that two boxed circuits $\boxedCirc{\struct\labOne}{\circuitTwo}{\struct\labTwo}$ and $\boxedCirc{\struct{\labOne'}}{\circuitTwo'}{\struct{\labTwo'}}$ are \emph{equivalent}, and we write $\boxedCirc{\struct\labOne}{\circuitTwo}{\struct\labTwo} \cong \boxedCirc{\struct{\labOne'}}{\circuitTwo'}{\struct{\labTwo'}}$, if they only differ by a renaming of labels. What $\operatorname{append}$ does, then, is find $\boxedCirc{\struct{\labThree}}{\circuitTwo'}{\struct{\labFour}} \cong \boxedCirc{\struct\labOne}{\circuitTwo}{\struct\labTwo}$ and return $\struct{\labFour}$, along with a circuit $\circuitThree$ defined as follows:
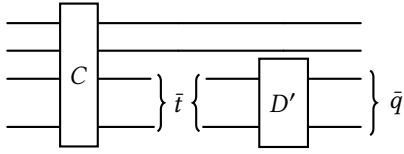
\[
\begin{quantikz}[row sep = 3pt, column sep = 20pt]
  &\gate[4]{\circuitOne}&&&\\
  &&&&\\
  &&\midstick[2]{\ $\struct{\labThree}$\ }&\gate[2]{\circuitTwo'}&\rstick[2]{\ $\struct{\labFour}$}\\
  &&&&
\end{quantikz}
\]
Overall, this semantics is the same as the one given in~\cite{ColledanDalLago24a}, except for the \textit{box} rule, which is modified to align with the modifications made in the 
typing rules.

\begin{figure}[t]
	\centering
	\fbox{
    \begin{minipage}{0.97\linewidth}
    \begin{mathpar}
			\inference[\textit{app}]
			{\config{\circuitOne}{\termOne\sub{\valOne}{\varOne}}
				\eval \config{\circuitTwo}{\valTwo}}
			{\config{\circuitOne}{\app{(\abs{\varOne}{\typeOne}{\termOne})}{\valOne}}
				\eval \config{\circuitTwo}{\valTwo}}
			\and
			\inference[\textit{dest}]
			{\config{\circuitOne}{\termOne\sub{\valOne}{\varOne}\sub{\valTwo}{\varTwo}}
				\eval \config{\circuitTwo}{\valThree}}
			{\config{\circuitOne}{\dest{\varOne}{\varTwo}{\tuple{\valOne}{\valTwo}}{\termOne}}
				\eval \config{\circuitTwo}{\valThree}}
			\and
      \inference[\textit{if-zero}]
			{\config{\circuitOne}{\termOne}
				\eval \config{\circuitTwo}{\valOne}}
			{\config{\circuitOne}{\ifz{0}{\termOne}{\termTwo}}
				\eval \config{\circuitTwo}{\valOne}}
      \and
			\inference[\textit{force}]
			{\config{\circuitOne}{\termOne} \eval \config{\circuitTwo}{\valOne}}
			{\config{\circuitOne}{\force{(\lift{\termOne})}} \eval \config{\circuitTwo}{\valOne}}
			\and
			\inference[\textit{apply}]
			{\config{\circuitThree}{\struct{\labFour}} = \append{\circuitOne}{\struct{\labThree}}{\boxedCirc{\struct\labOne}{\circuitTwo}{\struct{\labTwo}}}}
			{\config{\circuitOne}{\apply{\boxedCirc{\struct\labOne}{\circuitTwo}{\struct{\labTwo}}}{\struct{\labThree}}} \eval \config{\circuitThree}{\struct{\labFour}}}
			\and\!\!
			\inference[\textit{box}]
			{(\lcOne,\struct{\labOne})=\freshlabels{\mtypeOne}
				&
				\config{\cidentity{\lcOne}}{\app{\valOne}{\struct{\labOne}}} \eval \config{\circuitTwo}{\struct{\labTwo}}}
			{\config{\circuitOne}{\boxt{\mtypeOne}{\valOne}} \eval \config{\circuitOne}{\boxedCirc{\struct{\labOne}}{\circuitTwo}{\struct{\labTwo}}}}
			\and
			\inference[\textit{return}]
			{ }
			{\config{\circuitOne}{\return{\valOne}} \eval \config{\circuitOne}{\valOne}}
			\and
			\inference[\textit{let}]
			{\config{\circuitOne}{\termOne} \eval \config{\circuitThree}{\valOne}
				&
				\config{\circuitThree}{\termTwo\sub{\valOne}{\varOne}} \eval \config{\circuitTwo}{\valTwo}}
			{\config{\circuitOne}{\letin{\varOne}{\termOne}{\termTwo}} \eval \config{\circuitTwo}{\valTwo}}
    \end{mathpar}
    \end{minipage}}
	\caption{\PQC{} big-step operational semantics.}
	\label{fig:operational-semantics}
\end{figure}

\subsection{Type Preservation}

We write \( \configjudgment \lcOne \circuitOne \termOne \typeOne {\lcOne'} \) and say that the configuration \( \config \circuitOne \termOne \) is \emph{well-typed under \( \lcOne \) and \( \lcOne' \)} if \( \circuitOne : \lcOne \to \lcTwo, \lcOne' \) and \( \cjudgst \lcTwo \termOne \typeOne \) for some label context \( \lcTwo \) disjoint from \( \lcOne' \).
Similarly, we write \( \configjudgment {\lcOne} {\circuitOne} {\valOne} \typeOne {\lcOne'} \) if \( \configjudgment {\lcOne} \circuitOne {{\return \valOne}} \typeOne {\lcOne'} \).
We have the following type preservation theorem.

\begin{theorem}[Type Preservation]
  \label{thm:type-preservation}
  If \( \configjudgment \lcOne \circuitOne \termOne \typeOne {\lcOne'} \) and \( \config \circuitOne \termOne \eval \config{\circuitTwo}{\valOne} \), then \( \configjudgment {\lcOne} {\circuitTwo} {\valOne} \typeOne {\lcOne'} \).
\end{theorem}
\begin{proof}
  By induction of the derivation of \( \config \circuitOne \termOne \eval \config{\circuitTwo}{\valOne} \).
\end{proof}

\section{A Monadic Semantics for \PQC{}}
\label{sec:categorica-semantics-st}
Here we define a new monadic denotational semantics for \PQC{},
introduced in Section~\ref{sec:simple-type}. We first explain the 
categorical structure we use such as the indexed monad for
circuits. We then define the interpretation, and prove its soundness and 
adequacy.

\subsection{The Circuit Monad}
We now review the notion of indexed monads~\cite{Atkey09} (aka parameterized monads), which plays a key role in our model.
Intuitively, an indexed monad is a ``multi-sorted generalization'' of a monad.
Its formal definition is given as follows.
\begin{definition}[Indexed Monad~\cite{Atkey09}]
  Let \( \catC \) be any cartesian category and \( \catS \) be a category.
  A (strong) \( \catS \)-indexed monad on \( \catC \) is a quadruple \( (\monad, \unit, \mult, \strength) \) where
  \begin{itemize}
	\item \( \monad \colon \catS^\op \times \catS \times \catC \to \catC \) is a functor
	\item the \emph{unit} \( \unit \) is a family of morphisms \( \unit_{S, X} \colon X \to \monad (S, S, X) \) natural in \( X \) and dinatural in \( S \)
	\item the \emph{multiplication} \( \mult \) is a family of morphisms \( \mult_{S_1, S_2, S_3, X} \colon \monad (S_1, S_2, \mathcal \monad(S_2, S_3, X)) \to \monad(S_1, S_3, X) \) natural in \( S_1\), \( S_3 \) and \( X \) and dinatural in \( S_2 \).
	\item the \emph{strength} \( \strength \) is a family of morphisms \( \strength_{X, S_1, S_2, Y} \colon  X \times \monad (S_1, S_2, Y) \to \monad (S_1, S_2, X \times Y)\) natural in \( X \), \( S_1 \), \( S_2 \) and \( Y \).
  \end{itemize}
  The unit and multiplication must obey the evident monad laws and the axiom for strength (cf.\ Definition~\ref{def:cat-graded-monad}).
\end{definition}

The indexed monad we are interested in is the circuit monad \( \circmonad* \colon \catM^\op \times \catM \times \Set \to \Set \).
The circuit monad is defined as follows\footnotemark
\footnotetext{This is just an indexed monad for category actions~\cite{Atkey09}
and we do not claim any novelty in this definition.}:
\begin{align*}
  \circmonad \mObjOne \mObjTwo \setObjOne &\defeq \setObjOne \times \catM(\mObjOne, \mObjTwo) \\
  \unit_{\mObjOne, \setObjOne}(x) &\defeq (x, \id_\mObjOne) \\
  \mult_{\mObjOne_1, \mObjOne_2, \mObjOne_3, \setObjOne}((x, f), g) &\defeq (x, f; g) \\
  \strength_{\setObjOne, \mObjOne, \mObjTwo, \setObjTwo}(x, (y, f)) &\defeq ((x,y), f).
\end{align*}
So, the unit augments a value with the identity circuit and multiplication is just a sequential composition of circuits.
Since we assumed that the category of circuits is premonoidal, this premonoidal structure lifts to \( \circmonad* \) (in the sense of \cite[Def.~5]{Atkey09}).
That is, there is a natural transformation \( (\mObjOne \rtensor -)^{\dagger}_{\mObjTwo_1, \mObjTwo_2, \setObjOne} \colon \circmonad {\mObjTwo_1} {\mObjTwo_2} {\setObjOne} \to  \circmonad {\mObjOne \otimes \mObjTwo_1} {\mObjOne \otimes \mObjTwo_2} {\setObjOne} \) satisfying certain desired properties.
In elementary terms, this is merely the map that associates \( (x, \circuitOne) \) to \( (x, \mObjOne \rtensor \circuitOne) \).

As one might expect, we will interpret terms in the Kleisli category of \( \circmonad* \), which we now briefly explain.
Given an \( \catS \)-indexed monad  \( \monad \) on \( \catC \), its \emph{Kleisli category} \( \catC_\monad \) is a category whose objects are pairs of \( \catC \) and \( \catS \) objects and homsets \( \catC_\monad((\setObjOne, \mObjOne), (\setObjTwo, \mObjTwo)) \defeq \catC(\setObjOne, \monad (\mObjOne, \mObjTwo, \setObjTwo)) \).
The identity morphisms and composition of morphisms are defined using units and multiplication as the obvious generalization of those in the Kleisli category of an ordinary monad.
It is known that the Kleisli category induces a \emph{parameterized Freyd category}~\cite{Atkey09} \( J \colon \catC \times \catS \to \catC_{\monad}\) as is the case for the ordinary monad and Freyd category; here \( J \) is an identity on objects functor that strictly preserves the premonoidal structure of \( \catC \).
Since Freyd categories are known to have a better match with the fine-grained call-by-value syntax, our semantical model will be based on the parametrized Freyd category induced by \( \circmonad* \).
The circuit monad \( \circmonad* \) also has Kleisli exponentials: there is a functor \( \Klarr \setObjOne {-} \mObjOne \colon \Set_{\circmonad*} \to \Set \) for every objects \( \setObjOne \), \( \mObjOne \),  and there is a natural isomorphism \( \Lambda_{\setObjTwo, (\setObjOne, \mObjOne), (\setObjThree, \mObjTwo)} \colon \Set_{\circmonad*}((\setObjTwo \times \setObjOne, \mObjOne), (\setObjThree, \mObjTwo)) \cong \Set(\setObjTwo, \Klarr \setObjOne {(\setObjThree, \mObjTwo)} \mObjOne ) \).
This means that the parameterized Freyd category induced by \( \circmonad* \) is a closed parameterized Freyd category.
We note that the object \( \Klarr \setObjOne {(\setObjThree, \mObjTwo)} \mObjOne \) is just \( \functionSet \setObjOne {\setObjThree \times \catM(\mObjTwo, \mObjOne) }\), the set of functions from \( \setObjOne \) to \( \setObjThree \times \catM(\mObjTwo, \mObjOne) \).
The counit of this adjunction is written as \( \ev \colon ((\Klarr \setObjOne {(\setObjTwo, \mObjTwo)} \mObjOne) \times \setObjOne, \mObjOne) \to  (\setObjTwo, \mObjTwo) \).
The Kleisli category \( \Set_{\circmonad*} \) is a premonoidal category, where \( (\setObjOne, \mObjOne)\otimes  (\setObjTwo, \mObjTwo) = (\setObjOne \times \setObjTwo, \mObjOne \otimes \mObjTwo ) \) and the binoidal functors are defined in a way analogous to how the premonoidal structure was lifted to \( \circmonad* \).
\footnotemark
\footnotetext{This should not be confused with the premonoidal structure of a parameterized Freyd category with respect to the cartesian category, which exist even if \( \catS \) is not premonoidal.
See also Appendix~\ref{appx:monad-Freyd} for a review on parameterized Freyd categories.
}
\begin{figure}
  \[
  \begin{tikzcd}[row sep=large, column sep=large,wire types={n,n}]
    \Set \times \disc(\obj(\catM)) = \catV \arrow[d,hook]  & \arrow[l,"\unitp_\mObjOne"',hook'] \Set   \\
    \Set \times\catM  \arrow[r, "J"] & \Set_{\circmonad*} \arrow[u,"\Klarr \setObjOne {-} \mObjOne"']
  \end{tikzcd}
  \]
  \caption{Overview of the model for \PQM{}.}
  \label{fig:model-overview-st}
\end{figure}

\subsection{Interpreting Types, Programs, and Configurations}
We now give the denotational semantics of \PQC{} using the Freyd category induced by \( \circmonad* \).
In Figure~\ref{fig:model-overview-st}, we summarize the overall structure of the interpretation.
Terms will be interpreted in the Kleisli category \( \Set_{\circmonad*} \), and values will be interpreted in \( \Set \times \disc(\obj(\catM)) \), where \( \disc(\obj(\catM)) \) is the discrete category whose objects are those of \( \catM \).

The interpretation of types is given in Figure~\ref{fig:interpretation-of-types-st}.
Types are interpreted as objects in \( \Set \times \catM \), i.e.~as pairs of a set and an object of \( \catM \).
We write \( \sharp \colon \Set \times \catM \to \catM \) and \( \flat \colon \Set \times \catM \to \Set \) for the obvious forgetful functors.
The object \( \Klarr \setObjOne \setObjTwo \mObjOne \) used in the interpretation of the arrow type is the parameterized Kleisli arrow of the closed parameterized Freyd category.
It is used to model the type of a ``code'' under the idea that a closure is a pair of a code and an environment.
The subsctipt \( T \) represents the type of the additional bundled typed arguments corresponding to the free variables.
Note that \( !A \) is interpreted as \((\Klarr 1 {\sem A} I, I) \) representing closures that do not capture any free variables having a bundle type.
The objects \( \qubitts \) and \( \bitts \) are the interpretation of qubits and bits, respectively, that we assumed to exist in \( \catM \).
\begin{figure}
  \begin{align*}
  &\text{\fbox{\( \sem{\typeOne} \in \Set \times \catM \)}} \\[-0.3em]
	&\;\;\sem \ptypeOne \defeq (\semP \ptypeOne, I) \quad
  \sem \mtypeOne \defeq (1, \semM \mtypeOne) \quad
  \sem {\arrowst \typeOne \typeTwo \mtypeOne} \defeq (\Klarr {\flat \sem \typeOne} {\sem \typeTwo} {\semM \mtypeOne \otimes \semM{\sharp A}}, \semM \mtypeOne) \\
  &\;\;\sem{\tensor \typeOne \typeTwo} \defeq (\setObjOne \times \setObjTwo, \mObjOne \otimes \mObjTwo) \text{ where } \sem \typeOne = (\setObjOne, \mObjOne) \text{ and } \sem \typeTwo = (\setObjTwo, \mObjTwo) \\
  &\text{\fbox{\( \semP{\ptypeOne} \in \Set\)}} \\[-0.3em]
  &\;\;\semP \unitt \defeq 1 \quad
    \semP \natt \defeq \mathbb N \quad \\
    &\;\;\semP{\tensor \ptypeOne \ptypeTwo} \defeq \semP \ptypeOne \times \semP \ptypeTwo \quad
    \semP{!A} \defeq \Klarr 1 {\sem A} {I} \quad
  \semP{\circt {} \mtypeOne \mtypeTwo} \defeq \catM(\semM \mtypeOne, \semM \mtypeTwo) \\
  &\text{\fbox{\( \semM{\mtypeOne} \in \catM \)}} \\[-0.3em]
    &\;\;\semM{\qubitt}\defeq \qubitts \quad
    \semM{\bitt} \defeq \bitts \quad
    \semM{\munitt} \defeq I \quad
    \semM{\tensor \mtypeOne \mtypeTwo} \defeq \tensor {\semM \mtypeOne} {\semM \mtypeTwo}
  \end{align*}
  \caption{Interpretation of Simple Types.}
  \label{fig:interpretation-of-types-st}
\end{figure}

\begin{lemma}
  Given a type \( \typeOne \), we have \( \semM{\sharp A} = \sharp \sem \typeOne \).
\end{lemma}

In our interpretation, circuit types and function types for bundle types are isomorphic, and this supports our design choice for the typing rule \textit{box}.
\begin{proposition}
  \label{prop:box-as-iso}
  We have an isomorphism \( \boxsem:  \flat \sem{\arrowst \mtypeOne \mtypeTwo {}}  \cong \semP{\circt {} \mtypeOne \mtypeTwo}\) in \( \Set \).
  Therefore, \( \sem{\arrowst \mtypeOne \mtypeTwo {}} \cong \sem{\circt {} \mtypeOne \mtypeTwo} \)  in \( \Set_{\circmonad*} \).
\end{proposition}
\begin{proof}
  By unrolling the definition, we have the following obvious isomorphisms for sets
  \begin{align*}
    \flat \sem{\arrowst \mtypeOne \mtypeTwo {}}
    & = \Klarr 1 {(1, \semM{\mtypeTwo})} {\semM \mtypeOne} \\
    &\cong \functionSet 1 { 1 \times \catM(\semM \mtypeOne, \semM \mtypeTwo)} \\
    &\cong \catM(\semM \mtypeOne, \semM \mtypeTwo) \\
    &= \semP{\circt {} \mtypeOne \mtypeTwo}.
  \end{align*}
  The isomorphism between \( \sem{\arrowst \mtypeOne \mtypeTwo {}} \cong \sem{\circt {} \mtypeOne \mtypeTwo} \)  is given by \( J(\boxsem, \id_I )\).
\end{proof}

\begin{remark}
  As briefly mentioned, in \PQ{} calculi, the \( \boxt{} \) operator usually coerces a function of type \( !(\mtypeOne \multimap \mtypeTwo )\) to a circuit type.
  This is because, in presheaf models of \PQ{}, there is an isomorphism \( \sem{!(\arrowst \mtypeOne \mtypeTwo {})} \cong \sem{\circt {} \mtypeOne \mtypeTwo} \).
  In our model, adding a bang to the type \( \arrowst \mtypeOne \mtypeTwo {} \) means to additionally thunk a function that has no free variable having a bundle type.
  Invoking this thunk may cause some effects, i.e.~produce a circuit while returning a function that corresponds to a boxed circuit.
  Hence, we do not have an isomorphism between \( \sem{!(\arrowst \mtypeOne \mtypeTwo {})} \) and \( \sem{\circt {} \mtypeOne \mtypeTwo} \).
  While our calculus draws inspiration from those by~\citet{ColledanDalLago24a,ColledanDalLago25}, they also coerce \( !(\mtypeOne \multimap \mtypeTwo )\)
  to \( \circt{} \mtypeOne \mtypeTwo \).
  In their works, the issue of the effect is circumvented using the effect system; either by requiring that \( !(\mtypeOne \multimap \mtypeTwo ) \) has zero effect or by adding an effect annotation to the circuit type.
\end{remark}

Now we define the interpretation of typing judgments.
The interpretation is defined in Figure~\ref{fig:interpretation-val-judgment-st}~and~\ref{fig:interpretation-judgment-st}.
A valid typing judgment for values \( \vjudgst \contextOne \valOne \typeOne\) is interpreted as a morphism \( \sem{\vjudgst \contextOne \valOne \typeOne} \colon  \sem{\contextOne} \to \sem \typeOne \) in \( \catV ( = \Set \times \disc(\catM))\) capturing the fact that values only produces trivial circuits.
In contrast, a computational judgment \( \cjudgst \contextOne \termOne \typeOne \) is interpreted as \( \sem{\cjudgst \contextOne \termOne \typeOne} \colon  \sem{\contextOne} \to \sem \typeOne \) in \( \Set_{\circmonad*} \).
We sometimes denote these morphisms by \( \sem{\valOne} \) and \( \sem{\termOne} \).
The morphisms, \( \Delta \), \( \pi_x \) and \( ! \) in Figure~\ref{fig:interpretation-val-judgment-st} are the diagonal map, projection, and the unique map to the terminal object \( 1 \), respectively, which exist in \( \Set \); the morphism \( \dup_\setObjOne \) in Figure~\ref{fig:interpretation-judgment-st} is defined as \( J(\Delta_\setObjOne, \id_I) \) and it is used for duplicating a values with a parameter type.
Some obvious coherence isomorphisms of the cartesian product in \( \Set \) are omitted for simplicity.

\begin{figure}
  \begin{flushleft}
  \fbox{\( \sem{\vjudgst \contextOne \valOne \typeOne} \)}
  \end{flushleft}
  \begin{align*}
	&\sem{\vjudgst{\pcontextOne}{\unitv}{\unitt}} \defeq (!_{\sem \pcontextOne}, \id_{I}) \\
	&\sem{\vjudgst{\pcontextOne, \labOne:\wtypeOne}{\labOne}{\wtypeOne}} \defeq (!_{\semP \pcontextOne}, \id_{\semM w}) \\
	&\sem{\vjudgst{\pcontextOne,\varOne:\typeOne}{\varOne}{\typeOne}} \defeq (\pi_x, \id_I) \\ &\sem{\vjudgst{\contextOne}{\lambda x. \termOne}{\arrowst{\typeOne}{\typeTwo}{\rcount{\contextOne}}}}
	  \defeq (\Lambda(\sem{\termOne}), \id_{\semM{\rcount{\contextOne}}}) \\
	&\sem{\vjudgst{\pcontextOne}{\lift{\termOne}}{\bang{\typeOne}}} \defeq  (\Lambda(\sem{\cjudgst{\pcontextOne}{\termOne}{\typeOne}}), \id_I) \\
  &\sem{\vjudgst{\pcontextOne}{\boxedCirc{\struct\labOne}{\circuitOne}{\struct\labTwo}}{\circt{}{\mtypeOne}{\mtypeTwo}}} \defeq (!_{\sem \pcontextOne}; \widehat {\semM{\boxedCirc{\struct \labOne}\circuitOne{\struct \labTwo}}}, \id_I)  \\
  &\text{where }
    \begin{aligned}[t]
      \semM{\boxedCirc{\struct \labOne} \circuitOne {\struct \labTwo}} &\defeq \semM{\mtypeOne} \xrightarrow{\cong} \semM{\lcOne} \xrightarrow{\circuitOne} \semM{\lcTwo} \xrightarrow{\cong} \semM{\mtypeTwo} \\
      \hat \circuitOne &\defeq 1 \xrightarrow{\Lambda(J(\id_1, \circuitOne))} \Klarr 1 {(1, \semM \mtypeTwo)} {\semM \mtypeOne} \xrightarrow[\cong]{\mathmakebox[4em]{\boxsem}} \catM(\semM \mtypeOne, \semM \mtypeTwo)
    \end{aligned}\\ &\sem{\vjudgst{\pcontextOne,\contextOne_1,\contextOne_2}{\tuple{\valOne}{\valTwo}}{\tensor{\typeOne}{\typeTwo}}}
    \defeq
    \begin{aligned}[t]
        \sem \pcontextOne  \otimes \sem {\contextOne_1} \otimes \sem {\contextOne_2}
        &\xrightarrow{(\Delta_{\sem \pcontextOne}, \id_I) \otimes \id} \sem \pcontextOne \otimes \sem \pcontextOne  \otimes \sem {\contextOne_1} \otimes \sem {\contextOne_2} \\
        &\xrightarrow{\mathmakebox[3em]{\cong}} (\sem \pcontextOne  \otimes \sem {\contextOne_1}) \otimes (\sem \pcontextOne \otimes \sem {\contextOne_2}) \\
        &\xrightarrow{\sem \valOne \otimes \sem \valTwo} \sem{\typeOne \otimes \typeTwo}
    \end{aligned}
  \end{align*}
  \caption{Interpretation of Values.}
  \label{fig:interpretation-val-judgment-st}
\end{figure}

\begin{figure*}
  \begin{flushleft}
    \fbox{\( \sem{\cjudgst \contextOne \termOne \typeOne} \)}
  \end{flushleft}
  \begin{align*}
	\sem{\cjudgst{\pcontextOne,\contextOne_1,\contextOne_2}{\app{\valOne}{\valTwo}}{\typeTwo}}
      &\defeq \begin{aligned}[t]
        \sem \pcontextOne  \otimes \sem {\contextOne_1} \otimes \sem {\contextOne_2}
        &\xrightarrow{\dup_{\sem \pcontextOne}\otimes \sem {\contextOne_2} \otimes \sem {\contextOne_1}} \sem \pcontextOne \otimes \sem \pcontextOne  \otimes \sem {\contextOne_1} \otimes \sem {\contextOne_2}  \\
        &\xrightarrow{\mathmakebox[4em]{\cong}} (\sem \pcontextOne  \otimes \sem {\contextOne_1}) \otimes (\sem \pcontextOne \otimes \sem {\contextOne_2}) \\
        &\xrightarrow{J(\sem \valOne) \otimes J(\sem \valTwo)} \sem{\arrowst \typeOne  \typeTwo {\mtypeOne} } \otimes \sem \typeOne \xrightarrow{\ev} \sem \typeTwo
      \end{aligned} \\
	\sem{\cjudgst{\pcontextOne}{\force{\valOne}}{\typeOne}}
      &\defeq J(\sem{\vjudgst{\pcontextOne}{\valOne}{\bang{\typeOne}}}); \ev\\
    \sem{\cjudgstLessSpace{\pcontextOne,\contextOne_1,\contextOne_2}{\apply{\valOne}{\valTwo}}{\mtypeTwo}}
      &\defeq \begin{aligned}[t]
        \sem \pcontextOne  \otimes \sem {\contextOne_1} \otimes \sem {\contextOne_2}
        &\xrightarrow{\dup_{\sem \pcontextOne} \otimes \sem {\contextOne_2} \otimes \sem {\contextOne_1}} \sem \pcontextOne \otimes \sem \pcontextOne  \otimes \sem {\contextOne_1} \otimes \sem {\contextOne_2}  \\
        &\xrightarrow{\mathmakebox[4em]{\cong}} (\sem \pcontextOne  \otimes \sem {\contextOne_1}) \otimes (\sem \pcontextOne \otimes \sem {\contextOne_2} ) \\
        &\xrightarrow{J(\sem \valOne) \otimes J(\sem \valTwo)} \sem{\circt{} \mtypeOne \mtypeTwo} \otimes \sem \mtypeTwo \\
        &\xrightarrow{\applysem} \sem \mtypeTwo
      \end{aligned} \\
    \sem{\cjudgst{\pcontextOne}{\boxt{\mtypeOne}{\valOne}}{\circt{}{\mtypeOne}{\mtypeTwo}}}
      &\defeq \sem \pcontextOne \xrightarrow{J(\sem \valOne)} \sem{\arrowst \mtypeOne \mtypeTwo {}}  \xrightarrow[\cong]{J(\boxsem, \id_I)} \sem{\circt{} \mtypeOne \mtypeTwo} \\ %
    \sem{\cjudgst{\contextOne}{\return{\valOne}}{\typeOne}}
    &\defeq J(\sem{\vjudgst{\contextOne}{\valOne}{\typeOne}}) \\
    \sem{\cjudgstLessSpace{\pcontextOne,\contextOne_2,\contextOne_1}{\letinLessSpace{\varOne}{\termOne}{\termTwo}}{\typeTwo}}
    &\defeq \begin{aligned}[t]
        \sem \pcontextOne  \otimes \sem {\contextOne_2} \otimes \sem {\contextOne_1}
        &\xrightarrow{(\dup_{\sem \pcontextOne} \otimes \sem{\contextOne_1} \otimes \sem{\contextOne_2});\cong} (\sem \pcontextOne \otimes \sem {\contextOne_2})  \otimes (\sem \pcontextOne \otimes \sem {\contextOne_1}) \\
        &\xrightarrow{(\sem{\pcontextOne} \otimes \sem{\contextOne_2})\rtensor  \sem \termOne} \sem \pcontextOne \otimes \sem {\contextOne_2} \otimes \sem \typeOne \xrightarrow{\sem \termTwo} \sem \typeTwo
    \end{aligned} \\
    \sem{\cjudgst{\contextOne_1, \lvOne : \typeOne, \lvTwo : \typeTwo, \contextOne_2}{\termOne}{\typeThree}}
			&= \sem{\contextOne_1} \otimes \sem{\typeOne} \otimes \sem{\typeTwo} \otimes \sem{\contextOne_2} \xrightarrow{\cong} \sem{\contextOne_1} \otimes \sem{\typeTwo} \otimes \sem \typeOne \otimes \sem{\contextOne_2} \xrightarrow{\sem \termOne} \sem \typeThree
  \end{align*}
  \caption{Interpretation of computational judgments of the simple type system (excerpt).}
  \label{fig:interpretation-judgment-st}
\end{figure*}

The interpretation of values is standard, except for $\lambda$-abstractions and 
boxed circuits.
A $\lambda$-abstraction is interpreted as a closure.
The first element of \( \sem{\lambda x. \termOne} \) is the semantic 
counterpart of the function that takes variables of types \( 
\sharp(\contextOne) \) as additional parameters, and the second element
\( \id_{\sem {\sharp \contextOne}} \) is the ``environment''.
The interpretation of \( \lift M \) is just the interpretation of a thunk \( \lambda(). M \).
In the interpretation of boxed circuits, we use isomorphisms \( \semM{\mtypeOne} \cong \semM{\lcOne} \) and \( \semM{\lcTwo} \cong \semM{\mtypeTwo} \) that exist thanks to the premises of the typing rule \textit{circ} such as \( \lcOne \permequiv \lcOne' \) and \( \vjudgst {\lcOne'} {\struct \labOne} \mtypeOne \).
The important part of the interpretation of a boxed circuit \( \circuitOne \) is the map \( \hat \circuitOne \), which is just the global element \( \hat \circuitOne(*) \defeq \circuitOne \).
The definition using  \( \Lambda \) and  \( \boxsem \) emphasizes the idea that boxed circuits can be seen as special functions.

As for the interpretation of terms, the interpretation of \( \apply \valOne \valTwo \) and \( \boxt \mtypeOne \valOne \) are the most interesting cases.
The morphism \( \applysem_{\mObjOne, \mObjTwo} \colon (\catM(\mObjOne, \mObjTwo), I) \otimes (1, \mObjOne) \to (1, \mObjTwo) \) is defined by
\begin{align*}
  (\catM(\mObjOne, \mObjTwo), I) \otimes (1, \mObjOne)
    \xrightarrow{J(\boxsem^{-1}, \id_I) \otimes \id_{(1, \mObjOne)}}
  (\Klarr 1 {(1, \mObjTwo)} \mObjOne, I) \otimes (1, \mObjOne)
  \xrightarrow{\ev} (1, \mObjTwo).
\end{align*}
The box operator is interpreted by the post-composition of the isomorphism between function types and circuit types given in Proposition~\ref{prop:box-as-iso}.
The interpretation of the \( \letoperator \) operator is also worth explaining.
The premonoidal product \( (\sem{\pcontextOne} \otimes \sem{\contextOne_2})\rtensor  \sem \termOne \) adds wires of type \( \sharp \sem{\contextOne_2} \) (and \(\sharp \sem{\pcontextOne} = \munitt \) which can be ignored) to the circuit produced by \( \termOne \) so that \( \termTwo \) can use these wires.

\begin{remark}
  The interpretation is somewhat unorthodox in that some syntactic constructs do not have a corresponding semantic operator.
  This is because we are doing two things at once.
  The interpretation can be factorized into (1) a syntactic translation from \PQC{} to the internal language of parametrized Freyd categories (called the command calculus)~\cite{Atkey09} and (2) interpreting the translated term.
  The syntactic translation resembles a variant of closure conversion, but we are not sure if a category theoretic explanation can be given to this translation.
\end{remark}

The categorical semantics is correct with respect to the big-step operational semantics.
To further elaborate on this result we first extend the interpretation to configurations.
Intuitively, \( \sem {\config \circuitOne \termOne} \) is the morphism obtained
by post-composing \( \sem{\termOne} \), together with some parallel wires, to
\( \circuitOne \).

\begin{definition}[Interpretation of Configurations]
  Suppose that \( \configjudgment \lcOne \circuitOne \valOne \typeOne {\lcOne'} \) and suppose that \( \lcTwo \) and \( \lcTwo' \) are the label contexts that satisfy \( C \colon \lcOne \to  \lcTwo', \lcOne' \); \( \lcTwo \permequiv \lcTwo' \); and \( \vjudgst {\lcTwo} \valOne \typeOne \).
  Then we define\( \sem {\config \circuitOne \valOne} \) as \( (\id_1, \circuitOne; \perm); (\sem{\vjudgst {\lcTwo} \valOne \typeOne} \otimes \id_{\sem {\lcOne'}}) \), which is a morphism in \( \Set \times \catM \).
  Here, \( \perm \) is the isomorphism \(\semM{\lcTwo'} \otimes \semM{\lcOne'} \xrightarrow{\cong} \semM{\lcTwo} \otimes \semM{\lcOne'} \).
  Similarly, for \( \configjudgment \lcOne \circuitOne \termOne \typeOne {\lcOne'} \), we define \( \sem {\config \circuitOne \termOne} \) as \( J(\id_1, \circuitOne); (\sem{\cjudgst {\lcTwo} \termOne \typeOne} \ltensor \sem{\lcOne'})\), which is a morphism in \( \Set_{\circmonad*} \).
  Here, \( \lcTwo \) is the label context that types \( \termOne \) as in the case of \( \valOne \).
\end{definition}

\subsection{Main Results}

We are now ready to state the soundness and computational adequacy properties.
\begin{restatable}[Soundness]{theorem}{soundnessSt}
  \label{thm:soundness}
  Suppose that \( \configjudgment \lcOne \circuitOne \termOne \typeOne {\lcOne'} \) and \( \config \circuitOne \termOne \eval \config {\circuitOne'} \valOne \).
  Then \( \sem{\config \circuitOne \termOne} = J(\sem{\config {\circuitOne'} \valOne}) \).
\end{restatable}
\begin{proof}
  By induction on the derivation of the big-step evaluation relation.
    See Appendix~\ref{appx:st-proofs} for the details.
\end{proof}
\begin{restatable}[Computational Adequacy]{theorem}{computationalAdequacy}
  Suppose that \( \configjudgment \emptycontext \circuitOne \termOne \unitt \emptycontext \) and \( \sem{\config \circuitOne \termOne} = J(\sem{\config \circuitTwo \valOne})\).
  Then \( \config \circuitOne \termOne \eval \config \circuitTwo \valOne \) (possibly up to renamings of labels).
\end{restatable}
\begin{proof}
    By an argument utilizing logical relations similar to those defined in~\cite{ColledanDalLago24a}, which are given in Appendix~\ref{appx:st-proofs}.
\end{proof}

\section{Effect System}
\label{sec:effect-system}
We extend the type system defined in Section~\ref{sec:simple-type} with effect annotations that estimate the properties (e.g.~size) of the circuit generated by a program.
To put it another way, we introduce the type system underlying (a non-dependent version of) \PQR{}.
Then we give the interpretation of programs in \PQR{} by using the category graded monad~\cite{OrchardWE20}.
This exemplifies that the type-and-effect system of \PQR{}, although rooted in an operational perspective, also has a natural denotational reading.

\subsection{Effects for Circuits}
When designing an effect system, the key question to ask is ``What kind of structure should we assume on effects?''.
A common choice is to use a preordered monoid~\cite{Katsumata14}, where the monoid multiplication is used to compute the effect of sequential execution and the preorder is used for subtyping.
We make the same choice, but use \emph{categories} instead of \emph{monoids} because circuits are \emph{many-sorted} in the sense that circuits have various input and output interfaces.
Moreover, since we have postulated that circuits form a premonoidal category, it is natural to require that the algebraic structure representing the effect---dubbed \emph{circuit algebra}---also be premonoidal.

\begin{definition}[Circuit Algebra]
  A \emph{circuit algebra} \( \catE \) is a strict symmetric premonoidal category that is preorder enriched.
  The preorder enrichment means that:
  \begin{itemize}
    \item each homset \( \catE(\eObjOne, \eObjTwo) \) is a preordered set;
    \item composition of morphisms preservers the order: if \( \effOne_1 \ple  \effTwo_1 \) and \( \effOne_2 \ple \effTwo_2 \), then \( \effOne_1; \effOne_2 \ple \effTwo_1 ; \effTwo_2 \);
    \item \( \eObjOne \ltensor - \) (resp.~\( - \rtensor \eObjOne \)) preserves the order for any \( \eObjOne\): if \( \effOne \ple \effTwo \), then \( \eObjOne \ltensor \effOne \ple \eObjOne \ltensor \effTwo \).
  \end{itemize}
  We call morphisms of \( \catE \), ranged over by \( \effOne, \effTwo, \ldots \), \emph{effect annotations}.
  The identity over the monoidal unit \( \eObjUnit \) of \( \catE \) is denoted by \( \nulleff \), and we call it the \emph{null effect}.
  Objects of \( \catE \) will often be written in lowercase script letters so that they are distinguishable from objects of \( \catM \).
\end{definition}

Effect annotations are meant to abstract the actual effect.
We propose to consider this abstraction as a functor.
\begin{definition}[Abstraction]
  An abstraction \( \abstraction \) from the category of circuits \( \catM \) to a circuit algebra \( \catE \) is a strict symmetric premonoidal functor \( \abstraction \colon \catM \to \catE \).
  That is, \( \abstraction \) is a functor satisfying \( \abstraction(\mObjOne \otimes \mObjTwo) =  \abstraction(\mObjOne) \otimes \abstraction(\mObjTwo) \) (equality on the nose), \( \abstraction(\mObjOne \rtensor f) = \abstraction(\mObjOne) \rtensor \abstraction(f) \), \( \abstraction(g\ltensor \mObjTwo) = \abstraction(g) \ltensor \abstraction(\mObjTwo) \) and preserves the symmetry.
\end{definition}
In case \( \catM \) is the syntactic category of circuits, defining an abstraction is no different from giving a functorial semantics to circuits.
We believe that allowing arbitrary interpretations of circuits as effect annotations is not only conceptually clean but also helps us conceive of a wide variety of examples---though in practice, we should seek efficiently implementable effects.

\begin{remark}
  We defined abstractions as \emph{strict} premonoidal functors because it is known that a non-strict premonoidal functor is tricky to define~\cite{StatonLevy13,RomanSobocinski25}.
  A way to circumvent this issue is to use \emph{effectful categories}~\cite{RomanSobocinski25}, which are premonoidal categories endowed with a chosen family of central morphisms, as the definition of circuits.
  However, using premonoidal categories and strict premonoidal functors are enough to deal with examples of circuit algebras we show below, which contains examples for resource estimations that have been considered in the literature.
\end{remark}

\subsubsection*{Examples of Circuit Algebras}
Here we give some examples of circuit algebras that capture some notions of circuit metrics.

Since the concept of circuit metrics is intrinsically intensional, namely the way gates are placed \emph{is} important, we consider the syntactic category of circuits for the category \( \catM \) in the following examples.
A \emph{signature} is a tuple \( \sig = (\sig_0, \sig_1) \) where \( \sig_0 \) is the set of \emph{object variables}, \( \sig_1\) is the set of \emph{generators}, which are typed constants of the form \( f : \sigma \to \tau \) with \( \sigma, \tau \in \sig_0^*\).
We write \( \catM_{\sig} \) for the free strict symmetric premonoidal category (with trivial center) generated  by the signature \( \sig \), which can be defined by designing an appropriate term calculus (as in \cite{JoyalStreet91}) or by considering string diagrams.
To facilitate understanding, we shall informally deal with string diagrams by depicting the graph and considering them as ``monoidal string diagrams without interchange law'' (see e.g.~\cite{RomanSobocinski25} for a more mathematically formal definition).
Note that a functor \( \abstraction \colon \catM_{\sig} \to \catE \) is determined if we define how object variables and generators are mapped to objects and morphisms of \( \catE \), respectively.

As a prototypical example, we shall consider string diagrams over the signature \( \sigQC \) for quantum circuits.
The set \( \sigQC_0 \) is defined as \( \{ \qubitt, \bitt \} \) and the generators are finite sequences over the set of gates \( \sig_G \).
We assume that the set \( \sig_G \), which is also a set of generators, contains usual quantum gates such as the Hadamard gate \( \Hadamard \colon \qubitt \to \qubitt \), CNOT gate \( \CNOT \colon (\qubitt, \qubitt) \to (\qubitt, \qubitt) \) and measurement \( \mathit{meas} \colon \qubitt \to \bitt \).
A sequence \( \tilde \gateOne = \gateOne_1 \cdots, \gateOne_n \in \sigQC_1 \) has the type \( \sigma_1 \cdots \sigma_n \to \tau_1 \cdots \tau_n \) provided that \( \gateOne_i \colon \sigma_i \to \tau_i \), and intuitively corresponds to applying the gates \( g_i \) in parallel.
In a sense, this means that it is users responsibility to explicitly state which gates to be placed in parallel, and users cannot expect the gates to automatically slide and be parallelized.

\begin{example}[Gate count and naive depth]
\label{ex:gate-count-and-depth}
  The simplest notion of circuit metric we consider is \emph{gate count}.
  The number of gates can be captured by the monoid \( (\mathbb N, +) \), which can be seen as a single object circuit algebra (where we denote the only object as \( \star \)) whose morphisms are natural numbers \( n \colon \star \to \star \).
  Sequential composition is defined as addition and the identity morphism is given by \( 0 \).
  The functor \( \star \ltensor {-} \) is simply given as the identity functor: \( \star \ltensor \star \defeq \star \) and \( \star \ltensor n \defeq n \).
  Obviously, this category is order enriched by considering the standard ordering for the natural numbers.
  The abstraction functor \( \abstraction_{G} \colon \catM_{\sigQC} \to \mathbb N \) is defined by mapping \( \tilde \gateOne \in \sigQC_1 \) to the number of gates in \( \tilde \gateOne \).

  A very rough estimation of the circuit \emph{depth} can be given by the same circuit algebra.
  We define \( \abstraction_{D} \colon \catM_{\sigQC} \to \mathbb N \) by \( \abstraction_{D} (\tilde \gateOne) = 1\) for every \( \gateOne \in \sigQC_1 \).
  For example, let us consider the following circuits.
  \par\medskip
  \noindent\begin{minipage}{.45\linewidth}
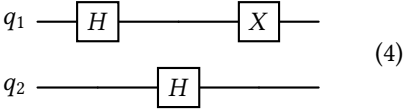
   \begin{equation}
  \begin{quantikz}[column sep=1.5em,row sep=1em ]
			\lstick{$q_1$} & \gate{H} &  & \gate{X} & \rstick{} \\
			\lstick{$q_2$} &  & \gate{H}&& \rstick{}
		\end{quantikz}
\label{eq:circ-depth}
\end{equation}
\end{minipage}
\begin{minipage}{.5\linewidth}
  
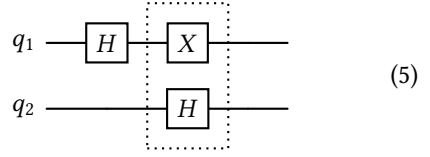
\begin{equation}
  \begin{quantikz}[column sep=1.5em,row sep=1em]
			\lstick{$q_1$} & \gate{H} &  \gate{X} \gategroup[2,steps=1, style={dotted}]{} & & \rstick{} \\
			\lstick{$q_2$} &  & \gate{H}&& \rstick{}
		\end{quantikz}
    \label{eq:circ-depth-opt}
  \end{equation}
  \end{minipage}\\[1em]
  Here, in \eqref{eq:circ-depth-opt}, the X and Hadamard gate are placed parallelly.
  The number of gates is estimated as \( 3 \) both in \eqref{eq:circ-depth} and \eqref{eq:circ-depth-opt}.
  On the other hand, the depth is estimated as \( 3 \) in \eqref{eq:circ-depth} since the three gates are sequentially composed, but \( 2 \) in \eqref{eq:circ-depth-opt}.
  While this way of counting depth has its own benefit of being easy to compute, it is not satisfactory because typically the depth of \eqref{eq:circ-depth} is also defined as \( 2 \); below we shall see a better way to count depth.
\end{example}

\begin{example}[Width~\cite{ColledanDalLago25}]
\label{ex:width}
  We explain a circuit algebra \( \category W  \) that is used to estimate (upper bounds on) the width of a circuit and an abstraction \( \abstraction_{\category W} \colon \catM_{\sigQC} \to \category W \).
  Recall that the \emph{width} of a circuit is just a natural number that is defined as the maximum number of wires active at any point in the circuit.
  Therefore, \emph{morphisms} in \( \category W \) should be natural numbers.
  For example, a quantum circuit \( \circuitOne \) depicted as
  
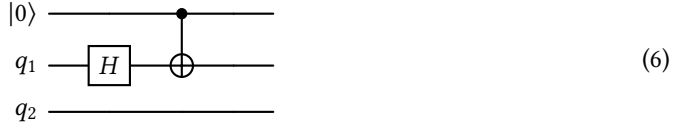
\begin{equation}
  \begin{quantikz}[column sep=1.5em,row sep=1em]
			\lstick{$\ket 0$} && \ctrl{1} \setwiretype{q} && \rstick{} \\
			\lstick{$q_1$} & \gate{H} & \targ & && \rstick{} \\
      \lstick{$q_2$} &&  && \rstick{}
		\end{quantikz}
  \label{eq:circ-width}
  \end{equation}
  has with width \( 3 \), i.e.~\( \abstraction_{\category W}(\circuitOne) = 3 \).
  Note that a wire of type \( \qubitt \) is counted as a circuit of width \( 1 \).
  This leads us to define \( \abstraction(\id_{\qubitt}) \colon \abstraction(\qubitt) \to \abstraction(\qubitt) \) as the natural number \( 1 \). %
  Now how should we define \( \abstraction(\qubitt) \)?
  Since this object should contain enough information to define the width of wires of type \( \qubitt \), a natural choice is to define this as the natural number \( 1 \).
  Hence, we also define \emph{objects} of \( \category W \) as natural numbers.
  The sequential composition of morphisms \( k_1 \xrightarrow{m}  k_2 \) and \( k_2 \xrightarrow{n} k_3 \) in \( \category W \) is defined as \( k_1 \xrightarrow{\max(m, n)} k_3 \) reflecting the definition of the width of a circuit.
  The functor \( k \rtensor {-} \) (resp. \( k \rtensor {-} \)) represents parallelly adding \( k \) wires.
  Hence, we define \( k \rtensor m  \defeq k + m \).
  To summarize, \( \category W \) is given by the following data: %
  \begin{itemize}
    \item \( \obj(\category W) \defeq \mathbb N \),
    \item \( \category W(k_1, k_2) \defeq (\mathbb N, \le_{\mathbb N}) \), and the composition is defined by max, with the identity morphism over an object \( k \) being the morphism \( k \) itself.
    \item functors \( k \rtensor - \) (resp. \( - \ltensor k \)) such that the action on objects and morphisms are both defined as \( k + (-)\). \qedhere
  \end{itemize}
\end{example}

\begin{example}[Depth]
  We give a better circuit algebra for depth, which may be seen as a way to count the naive depth of a circuit after optimizing it by sliding gates.
  The idea is to track the depth using \emph{matrices over max-plus tropical semiring} \( (\mathbb N \cup \{ - \infty \}, \max, + )\), where the \( (i, j) \) component of the matrix describes the cost for traversing from the \( i \)-th input to the \( j \)-th output.
  We define a circuit algebra \( \category {D} \) as follows:
  \begin{itemize}
    \item \( \obj(\category{D}) \defeq \mathbb N \),
    \item \( \category{D}(k_1, k_2) \defeq  \MatSet {\mathbb N}{k_1}{k_2} \times \MatSet{\mathbb N}{1}{k_1} \times  \MatSet{\mathbb N}{k_2}{1} \), where \( \MatSet{\mathbb N}{k_1}{k_2} \) is the set of \( k_1 \times k_2 \) matrices over the tropical semiring.
      The composition \( (\mathbf{A_1}, \mathbf{v_1}, \mathbf{w_1}) \circ  (\mathbf{A_2}, \mathbf{v_2}, \mathbf{w_2}) \) is given as  \( (\mathbf{A_1} \mathbf{A_2}, \max(\mathbf{v_1} \mathbf{A_2}, \mathbf{v_2}),  \max(\mathbf{w_1}, \mathbf{A_1} \mathbf{v_2})\); here \( \max \) acts on vectors component-wise.
      The identity morphism over \( k \) is given as \( (\mathbf{I}_k, \mathbf{0}_{1 \times k}, \mathbf{0}_{k \times 1})\), where \( \mathbf{I}_k \) is the \( k \times k \) identity matrices and \( \mathbf{0}_{n \times m} \) is the \( n \times m \) zero matrix over the tropical semiring.
      These should not be confused with the standard identity and zero matrices, say over \( \mathbb Q\). For example, \( \mathbf{0}_{1 \times k} \) is \( (- \infty, \ldots, -\infty ) \).
      The ordering on morphisms is given by ordering over matrices where \( \mathbf{A} \le \mathbf{B} \) if \( a_{i, j} \le  b_{i, j} \) for every \( (i, j) \).

    \item functors \( k \rtensor - \) (resp. \( - \ltensor k \)) are defined by the ``direct sum'' of matrices.
      Concretely, for \( (\mathbf{A}, \mathbf{v}, \mathbf{w}) \in \category{D}(k_1, k_2) \), we define \( (\mathbf{A}, \mathbf{v}, \mathbf{w}) \rtensor k \) as
      \( ( \left(\begin{smallmatrix}\mathbf{A} & \mathbf{0}_{k_1 \times k} \\ \mathbf{0}_{k \times k_2} & \mathbf{I}_k \end{smallmatrix} \right), (\mathbf{v}, \mathbf{0}_{1 \times k}), \left( \begin{smallmatrix} \mathbf{w} \\ \mathbf{0}_{k \times 1}\end{smallmatrix} \right) )\).
      Here again, the zero and identity matrices are those over the tropical semiring.
  \end{itemize}
  As mentioned, components of \( (\mathbf{A}, \mathbf{v}, \mathbf{w})\) represents the cost of the paths in a circuit as illustrated in Figure~\ref{fig:depth}.
  The role of the vectors \( \mathbf{v} \) and \( \mathbf{w} \) is to track the depth of wires that have ``dead ends'', for instance, created by qubit creation or annihilation.
  The \( i \)-th component of \( \mathbf{v} \) describes the maximum cost for traversing the circuit from the \( i \)-th input until it reaches an end.
  Conversely, the \( i \)-th component of the vector \( \mathbf{w} \) tracks the maximum depth of a path starting from an ``end''  and ending at the \( i \)-th output.

  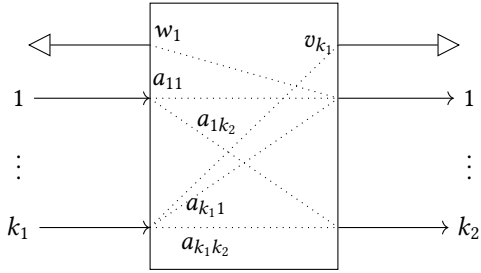
\begin{figure}[h]
  \begin{tikzpicture}[node distance=1cm and 2cm]

  \node (source) at (0, 1.2) {};
  \node (i1) at (0,0.5) {\( 1 \)};
  \node (k1) at (0,-1.2) {\(k_1\)};
  \node (vdots1) at (0, -0.3) {\( \vdots \)};

  \node[draw, minimum width=2.5cm, minimum height=3.5cm] (box) at (3,0) {};
  \node (sink) at (6, 1.2) {};
  \node (o1) at (6,0.5) {\( 1 \)};
  \node (k2) at (6,-1.2) {\( k_2 \)};
  \node (vdots2) at (6, -0.3) {\( \vdots \)};

  \draw[-{Triangle[open, length=3mm,width=3mm]}] (box.west |- source) -- (source);
  \draw[->] (i1) -- (box.west |- i1);
  \draw[->] (k1) -- (box.west |- k1);

  \draw[-{Triangle[open, length=3mm,width=3mm]}] (box.east |- sink) -- (sink);
  \draw[->] (box.east |- o1) -- (o1);
  \draw[->] (box.east |- k2) -- (k2);

  \draw[dotted] (box.west |- source) -- (box.east |- o1) node[pos=0.1, above] {\(w_{1}\)};
  \draw[dotted] (box.west |- i1) -- (box.east |- o1) node[pos=0.1, above] {\(a_{11}\)};
  \draw[dotted] (box.west |- i1) -- (box.east |- k2) node[pos=0.2, right] {\(a_{1 k_2}\)};
  \draw[dotted] (box.west |- k1) -- (box.east |- o1) node[pos=0.3, below] {\(a_{k_1 1}\)};
  \draw[dotted] (box.west |- k1) -- (box.east |- k2) node[pos=0.3, below] {\(a_{k_1 k_2}\)};
  \draw[dotted] (box.west |- k1) -- (box.east |- sink) node[pos=0.9, above] {\(v_{k_1}\)};
\end{tikzpicture}
\caption{Components of \( (\mathbf{A}, \mathbf{v}, \mathbf{w})\).}
\label{fig:depth}
\end{figure}

  The depth of the circuit \eqref{eq:circ-depth} of Example~\ref{ex:gate-count-and-depth}, is given as \( (\left(\begin{smallmatrix} 1 & -\infty \\ -\infty&  2 \end{smallmatrix}\right) , (-\infty, -\infty),  \left(  \begin{smallmatrix} -\infty \\ -\infty \end{smallmatrix} \right)) \).
  Now we can correctly conclude that the depth is \( 2 \), not \( 3 \), as the maximum number in the tuple is \( 2 \).
  As an example containing a ``dead end'', we show the depth of the circuit \eqref{eq:circ-width} of Example~\ref{ex:width}.
  It is given as \( (\left(\begin{smallmatrix} 2 & 2 & -\infty \\ -\infty & -\infty &  0 \end{smallmatrix}\right) , (-\infty, -\infty),  \left(  \begin{smallmatrix} 1\\ 1 \\ -\infty \end{smallmatrix} \right)) \).
\end{example}
It is natural to ask whether any matrix circuit algebra can capture optimized circuit width taking into account the so-called qubit recycling.
Given well-known results on the NP-hardness of qubit recycling~\cite{Jiang24}, this seems difficult.
On the other hand, it is possible to capture the so-called qubit dependency graph, which serves as input to a (heuristic) solver that outputs a qubit recycling strategy~\cite{Jiang24}.
The graph can be calculated as an adjacent matrix (over the boolean semiring) by employing an approach similar to that used for capturing circuit depth.

\begin{example}[Assertion-Based Optimization]
  \label{ex:optimization}
  We consider the size of quantum circuits modulo an optimization that removes operations that act trivially on states that are known to satisfy certain conditions.
  For example, if we know that the input of a \CNOT{} gate is in a state \( \ket \psi = \alpha_{00}\ket{0 0} + \alpha_{01} \ket{01} \), then we know that this \CNOT{} operation is \emph{trivial} in the sense that it behaves as the identity operation because \( \ket{0b} \xmapsto{\CNOT} \ket{0 b} \) for \( b \in \{ 0, 1\}\).
  This is actually the case for the circuit \eqref{eq:circ-width} of Example~\ref{ex:width}, meaning that we may remove the \CNOT{} gate since it is redundant.
  In \eqref{eq:circ-width}, the fact that the first qubit is zero is evident because it was created right before the CNOT gate, but we may also \emph{assert} such properties against inputs of the circuit and remove gates based on these assertions.
  This is the core idea of an automated optimization methodology proposed by \citet{HanerHT20}.\footnote{This optimization method is implemented as a transpiler pass in Qiskit~\cite{HoareOptimizer}.}

  Our aim here is not to give an effect annotation that works as an optimizer that transforms a given circuit but to capture the size of the circuit after the optimization as an effect annotation.
  For the size, we consider gate counts for simplicity, but the other circuit metrics can be used as well.
  Since the size of the circuit depends on the precondition, we consider a function of the type \( \powerset X  \to \powerset Y \times \mathbb N \), essentially a forward predicate transformer combined with a ``cost monad '' (i.e. a writer monad) \( (-) \times \mathbb N \).
  The reason for returning not just the size but also the postcondition in \( \powerset Y \), is simply to make the effect annotations compose.
  The choice of the set \( X \) and \( Y \) is the key to obtaining a tractable notion of effect annotation.
  Here, we follow \citet{HanerHT20} and take \( X  \defeq \{0, 1\}^m \) and \( Y \defeq \{1, 0\}^n \), where \( m \) and \( n \) are the number of input and output qubits, respectively.
  A bitstring \( b \in \{ 0, 1 \}^n \) represents the \( b \)-th (written in the binary format) computational base state, and \( L \subseteq \{ 0, 1\}^n \) can be considered as the set of possible outcomes of the quantum state.
  Formally, we say that a (pure) state \( \ket {\psi} \in \mathbb{C}^{2^n} \) satisfies the predicate \( L \) if \( \ket \psi = \sum_{b \in \{0, 1\}^n} \alpha_b \ket b \) and \( |\alpha_b| > 0 \) implies \( b \in L \).\footnote{It is easy to extend this satisfaction relation to mixed states.}
  For \CNOT{}, we define a function \( \effOne : \{ 0, 1 \}^2 \to \{ 0, 1\}^2 \times \mathbb N \), which, for instance, associates \( \{00, 01\} \) to \( (\{ 00, 01 \}, 0) \) and \( \{ 00, 10 \} \) to \( (\{ 00, 11 \}, 1) \).
  The reason why we assign the cost \( 0 \) to \( \{00, 01\}\) is because \CNOT{} acts as an identity for \( \ket \psi \) satisfying \( \{00, 01\} \), meaning that this \CNOT{} gate can be removed by optimization.
  To capture the linearity of the operation, we require that the effect \( \powerset(\{ 0, 1\}^m) \to \powerset(\{ 0, 1\}^n ) \times \mathbb N \) to be join preserving.\footnotemark
  \footnotetext{It is well-known that there is a bijection between join preserving functions from \( \powerset X \) to \( \powerset Y \) and Kleisli morphisms \( f \colon X \to \powerset Y \) . Similarly, we may think that we are working with morphisms \( X \to \powerset Y \times \mathbb N\). }

  The circuit algebra we consider, denoted as \( \Asrt \) is defined as follows: %
  \begin{itemize}
    \item Objects are natural numbers
    \item A morphism \(\effOne \in \Asrt(m, n)\) is a function from \( \powerset (\{0, 1\}^m ) \) to \( \powerset(\{0, 1\}^n) \times \mathbb N \) that is join preserving.
          That is, \( \effOne(L_1 \cup L_2) = (L'_1 \cup L'_2 ,\max(c_1, c_2))\) where \( \effOne(L_i) = (L'_i , c_i)\).
          Sequential composition is defined as that of the writer monad, and the identity morphism is simply \( L \mapsto (L, 0) \).
          The ordering between morphisms is defined as \( \effOne \ple \effTwo \) if, for every \( L \), \( \effOne(L) \ple_{\powerset(\{0, 1\}^k) \times \mathbb N} \effTwo(L) \), where \( \ple_{\powerset(\{0, 1\}^k) \times \mathbb N} \) is the product order of \( (\powerset (\{0, 1\}^k), \subseteq) \) and \( (\mathbb N, \le_{\mathbb N})\).
  \item the functor \( k \rtensor - \) that acts on \( \effOne : m \to n \) as \( (k \rtensor \effOne)(L) \defeq \bigcup_{(b_1, b_2) \in L} \effOne(\{ b_1\}) \times \{ b_2 \} \). On objects, it just acts as \( k + (-)\) as the previous examples.
  \end{itemize}
  If we consider a subsignature of \( \sigQC\) that only has \( \qubitt \) as object variable and unitary gates as generators and the free premonoidal category generated by it, then we can define the abstraction function to \( \Asrt \) by giving the interpretation to these unitary operators as we did for \CNOT{}.
  We can also handle \( \bitt \) and non unitary operation such as \( \mathrm{meas}: \qubitt \to \bitt \) by considering the set of predicates \( \powerset(\{0,1\}^m \times \{ 0, 1\}^n) \) for a state with \( m \) qubits and \( n \) classical bits.
\end{example}
The circuit algebra examples we have discussed encompass all those previously considered by~\citet{ColledanDalLago25}. In contrast, Example~\ref{ex:optimization} represents a novel contribution: to the best of the authors' knowledge, no existing resource analysis techniques for quantum programs in the literature account for optimizations, i.e., no such techniques is capable of deriving bounds that reflect the improvement in size induced by optimizations.
It is also worth noting that the notions of width and depth used by Colledan and Dal Lago differ in nature: the former is global, assigning a single numerical value to each circuit, while the latter is local, assigning a value to each individual qubit. These two types of metrics are captured in different ways in~\textit{op.\ cit}.
Circuit algebras provide a unified formalism that can accommodate both global and local metrics within the same framework.
As we will see, the proof of soundness of the resulting type system will be done just once.

\subsection{Type-and-Effect System}
Now we add effect annotations to the types.
The type system is parameterized by an abstraction to a circuit algebra \( \abstraction \colon \catM \to \catE \).
As usual, arrow types have annotations that estimate the scope of effect caused by invoking the function.
We also annotate \( \bang \) and \( \circt{} \mtypeOne \mtypeTwo \) because thunks and boxed circuits can be thought of as special functions.
The grammar of types and parameter types now become as follows:
\begin{align*}
  \text{Types} \quad	\typeOne,\typeTwo
  &\Coloneqq \ptypeOne \mid \mtypeOne \mid \tensor{\typeOne}{\typeTwo} \mid \arroweff \typeOne \typeTwo \mtypeOne {\effOne \colon \eObjOne \to \eObjTwo} \\
  \text{Parameter types}	\quad \ptypeOne,\ptypeTwo
  &\Coloneqq \unitt \mid  \natt \mid \tensor{\ptypeOne}{\ptypeTwo} \mid \bangeff{\typeOne}{\effOne \colon \eObjOne \to \eObjTwo} \mid \circt{\effOne \colon \eObjTwo \to \eObjThree}{\mtypeOne}{\mtypeTwo}
\end{align*}
The definition of bundle types remains the same.

The typing judgment for computations now takes the form \( \cjudgeff \contextOne \termOne \typeOne {\effOne \colon \eObjOne \to \eObjTwo} \).
The annotation \( \effOne  \colon \eObjOne \to \eObjTwo \) gives the information about the circuit generated by \( \termOne \); the type of the effect annotation \( \eObjOne \to \eObjTwo \) are sometimes omitted for readability.
The shape of typing judgment for values is the same as before.

\begin{figure}[t]
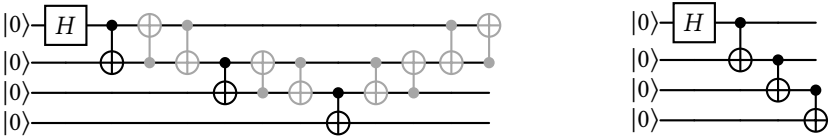

	\centering
	\fbox{
    \begin{minipage}{0.97\linewidth}
    \begin{mathpar}
			\inference[\textit{abs}]
			{\cjudgeff{\contextOne,\varOne:\typeOne}{\termOne}{\typeTwo} {\effOne \colon \eObjOne \to \eObjTwo}}
			{\vjudgst{\contextOne}{\abs{\varOne}{\typeOne}{\termOne}}{\arroweff{\typeOne}{\typeTwo}{\rcount{\contextOne}} {\effOne \colon
				\eObjOne \to \eObjTwo}}}
			\and
			\inference[\textit{app}]
			{\vjudgst{\pcontextOne,\contextOne_1}{\valOne}{\arroweff{\typeOne}{\typeTwo}{\mtypeOne}{\effOne \colon \eObjThree \to \eObjTwo}}
				&
				\vjudgst{\pcontextOne,\contextOne_2}{\valTwo}{\typeOne}}
			{\cjudgeff{\pcontextOne,\contextOne_1,\contextOne_2}{\app{\valOne}{\valTwo}}{\typeTwo}{\effOne \colon \eObjThree \to \eObjTwo}}
			\and
			\inference[\textit{lift}]
			{\cjudgeff{\pcontextOne}{\termOne}{\typeOne} {\effOne \colon \eObjOne \to \eObjTwo}}
			{\vjudgst{\pcontextOne}{\lift{\termOne}}{\bangeff{\typeOne} \effOne}}
			\and
			\inference[\textit{force}]
			{\vjudgst{\pcontextOne}{\valOne}{\bangeff{\typeOne} \effOne}}
			{\cjudgeff{\pcontextOne}{\force{\valOne}}{\typeOne} {\effOne \colon \eObjOne \to \eObjTwo}}
			\and
			\inference[\textit{circ}]
			{\circjudgment{\circuitOne}{\lcOne}{\lcTwo}
			  &
        \lcOne \permequiv \lcOne'
        &
        \lcTwo \permequiv \lcTwo'
			  \\
				\vjudgst{\lcOne'}{\struct\labOne}{\mtypeOne}
				\;\;
				\vjudgst{\lcTwo'}{\struct\labTwo}{\mtypeTwo}
        \;\;
			  \abstraction(\semM{\boxedCirc{\struct\labOne}{\circuitOne}{\struct \labTwo}}) = \effOne
				} {\vjudgst{\pcontextOne}{\boxedCirc{\struct\labOne}{\circuitOne}{\struct\labTwo}}{\circt{\effOne}{\mtypeOne}{\mtypeTwo}}}
        \and
        			\inference[\textit{box}]
			{\vjudgst{\pcontextOne}{\valOne}{{\arroweff{\mtypeOne}{\mtypeTwo}{I}{\effOne\colon \eObjOne \to \eObjTwo}}}}
			{\cjudgeff{\pcontextOne}{\boxt{\mtypeOne}{\valOne}}{\circt{\effOne\colon \eObjOne \to \eObjTwo}{\mtypeOne}{\mtypeTwo}} \nulleff}

      \and
      			\inference[\textit{apply}]
			{\vjudgst{\pcontextOne,\contextOne_1}{\valOne}{\circt{\effOne \colon \eObjOne \to \eObjTwo}{\mtypeOne}{\mtypeTwo}}
				\\
				\vjudgst{\pcontextOne,\contextOne_2}{\valTwo}{\mtypeOne}}
			{\cjudgeff{\pcontextOne,\contextOne_1,\contextOne_2}{\apply{\valOne}{\valTwo}}{\mtypeTwo} {\effOne \colon \eObjOne \to \eObjTwo}}
			\\
			\and
			\inference[\textit{dest}]
			{\vjudgst{\pcontextOne,\contextOne_1}{\valOne}{\tensor{\typeOne}{\typeTwo}}
				\\\
				\cjudgeff{\pcontextOne,\contextOne_2,\varOne:\typeOne,\varTwo:\typeTwo}{\termOne}{\typeThree}{\effOne}}
			{\cjudgeff{\pcontextOne,\contextOne_2,\contextOne_1}{\dest{\varOne}{\varTwo}{\valOne}{\termOne}}{\typeThree}{\effOne}}
      \and
			\inference[\textit{ifz}]
			{\vjudgst{\pcontextOne}{\valOne}{\natt}
				\\\
				\cjudgeff{\pcontextOne,\contextOne}{\termOne}{\typeOne}{\effOne}
        \and
        \cjudgeff{\pcontextOne,\contextOne}{\termTwo}{\typeOne}{\effOne}}
			{\cjudgeff{\pcontextOne,\contextOne}{\ifz{\valOne}{\termOne}{\termTwo}}{\typeOne}{\effOne}}
			\and

			\and
			\inference[\textit{return}]
			{\vjudgst{\contextOne}{\valOne}{\typeOne} \quad \abstraction(\sharp \sem \typeOne) = \eObjOne }
			{\cjudgeff{\contextOne}{\return{\valOne}}{\typeOne} {\id_{\eObjOne} \colon \eObjOne \to \eObjOne}}
			\and
			\inference[\textit{let}]
			{
			  \cjudgeff{\pcontextOne,\contextOne_1}{\termOne}{\typeOne} {\effOne_1 \colon \eObjOne_1 \to \eObjOne_1'}
				\\
				\cjudgeff{\pcontextOne,\contextOne_2,\varOne:\typeOne}{\termTwo}{\typeTwo}  {\effOne_2 \colon \eObjOne_2 \to \eObjOne_2'} \\
				\abstraction(\sem{\sharp \contextOne_i})= {\eObjTwo_i}  & \effOne = (\id_{\eObjTwo_2} \rtensor \effOne_1); \effOne_2
			 }
			{\cjudgeff{\pcontextOne,\contextOne_2,\contextOne_1}{\letin{\varOne}{\termOne}{\termTwo}}{\typeTwo} \effOne} \and
  	\inference[\textit{sub}]
			{
			  \cjudgeff{\contextOne}{\termOne}{\typeOne} {\effOne_1 \colon \eObjOne \to \eObjTwo} &
        \effOne_1 \ple \effOne_2
      }
      {\cjudgeff{\contextOne}{\termOne}{\typeOne} {\effOne_2 \colon \eObjOne \to \eObjTwo}}
    \end{mathpar}
    \end{minipage}}
      \caption{Typing Rules for the Effect System of \PQR{} (excerpt).}
	\label{fig:eff-typing-rules}
\end{figure}

Typing rules are given in Figure~\ref{fig:eff-typing-rules}.
Rules for lambda abstraction, application, thunking and forcing are the standard rules.
The rule for boxed circuit adds the effect \( \effOne \) calculated by abstracting the circuit \( \circuitOne \) as the annotation.
In practice, the boxed circuits are the primitive constant circuits for which the abstraction is predefined.
The rules of greatest interest may be the rules for \( \returnoperator \) and \( \letoperator \).
Usually, the rule for \( \returnoperator \) in a type-and-effect system adds the null effect since values are effectless.
In our calculus, we cannot give such a uniform treatment to values.
Instead, we add the annotation \( \id_{\eObjOne} \), which represents the effect of the identity circuit of type \( \sharp \sem{ \typeOne } \).
This is different from the null effect; for example, if we are interested in width, the width of identity circuits cannot be treated as zero.
For the same reason that we cannot ignore identity circuits, the rule for \( \letoperator \) is annotated by \( (\id \rtensor \effOne_1); \effOne_2 \) rather than \( \effOne_1; \effOne_2 \).
We also have a new subsumption rule, which allows us to relax the effect annotation.\footnotemark
\footnotetext{It is also possible to define a subtyping relation based on the \( \effOne \ple \effOne' \) in a standard way.}
We note that the domain (resp.~codomain) of the effect is an abstraction of the bundle type of the type environment (resp.~return type).
\begin{lemma}
  Suppose that \( \cjudgeff \contextOne \termOne \typeOne {\effOne \colon \eObjOne \to \eObjTwo} \).
  Then we have \( \abstraction({\sharp \sem {\contextOne}}) =  \eObjOne \) and \( \abstraction({\sharp \sem {\typeOne}}) = \eObjTwo \).
  \qed
\end{lemma}

We also adjust the typing for configurations.
We write \( \configjudgmentEff \lcOne \circuitOne \termOne \typeOne {\lcOne'} {\effOne : \eObjOne \to \eObjTwo} \) if \( \circuitOne \colon  \lcOne  \to \lcTwo, \lcOne' \) and \( \cjudgeff \lcTwo \termOne \typeOne {\effOne : \eObjOne \to \eObjTwo} \) for some label context \( \lcTwo \) disjoint from \( \lcOne \).
Type preservation (Theorem~\ref{thm:type-preservation}) holds even after this modification.

\begin{example}[A program with information about assertion-based optimization]
  We provide a simple example of how \PQR{} can be used to verify the resource usage of a circuit at the language level.
  We use a program that generates an \emph{inefficient} circuit for a linear nearest-neighbor (LNN) architecture, taken from Figure~3 of~\cite{HanerHT20}.
  In a LNN architecture gates can be applied only to adjacent quibits, and because of this restriction, a programmer (or a compiler) might write an inefficient code following a certain idiom.
  The program\footnotemark in Figure~\ref{fig:prog-for-LNN} is an example of such a program, which generates the left circuit given below.
  \footnotetext{For readability, we did not write the program using the syntax of Section~\ref{sec:simple-type}, but a variant that uses a ML-like syntax.}
  This circuit can be optimized into the one in the right by optimizing the trivial CNOT gates, which are gray in the left circuit.
\begin{figure}[t]
\begin{lstlisting}
-- init4: Circ(I, qubit ⊗ qubit ⊗ qubit ⊗ qubit) [e1]
let (q1, q2, q3, q4) = apply(init4, *) in
-- hadamard : Circ(qubit, qubit) [e2]
let q1 = apply(hadamard,q1) in
-- cnot12 : Circ(qubit ⊗ qubit, qubit ⊗ qubit) [e3]
let (q1, q2) = apply(cnot12, (q1, q2)) in
-- cnot21cnot12 : Circ(qubit ⊗ qubit, qubit ⊗ qubit) [e4]
-- The sequential composition of cnot21 and cnot12
let (q1, q2) = apply(cnot21cnot12, (q1, q2)) in
let (q2, q3) = apply(cnot12, (q2, q3)) in
let (q2, q3) = apply(cnot21cnot21) in
...
return (q1, q2, q3, q4)
\end{lstlisting}
\caption{A program that generates a chain of redundant CNOTs for a LNN architecture.}
\label{fig:prog-for-LNN}
\end{figure}

\begin{quantikz}[column sep=0.5em,row sep=0.2em]
  \ket 0 & \gate{H} &  \ctrl{1} & \targ[style={draw=gray!70}]{}   & \ctrl[style={fill=gray!70, draw=gray!70}]{1} &           &           &          &          &          &           & \ctrl[style={fill=gray!70, draw=gray!70}]{1} & \targ[style={draw=gray!70}]{}   \\
  \ket 0 &          &  \targ{}  & \ctrl[style={fill=gray!70, draw=gray!70}]{-1} & \targ[style={draw=gray!70}]{}  &  \ctrl[]{1} & \targ[style={draw=gray!70}]{}   & \ctrl[style={fill=gray!70, draw=gray!70}]{1} &          & \ctrl[style={fill=gray!70, draw=gray!70}]{1} & \targ[style={draw=gray!70}]{}   & \targ[style={draw=gray!70}]{}  & \ctrl[style={fill=gray!70, draw=gray!70}]{-1} \\
  \ket 0 &          &           &           &          &  \targ{}  & \ctrl[style={fill=gray!70, draw=gray!70}]{-1} & \targ[style={draw=gray!70}]{}  & \ctrl[]{1} & \targ[style={draw=gray!70}]{}  & \ctrl[style={fill=gray!70, draw=gray!70}]{-1} &          &           \\
  \ket 0 &          &           &           &          &           &           &          & \targ{}  &          &           &          &
\end{quantikz}
\qquad\qquad
\begin{quantikz}[column sep=0.5em,row sep=0.2em]
  \ket 0 & \gate{H} &  \ctrl{1} &          &          \\
  \ket 0 &          &  \targ{}  & \ctrl{1} &          \\
  \ket 0 &          &           &  \targ{} & \ctrl{1} \\
  \ket 0 &          &           &          & \targ{}  \\
\end{quantikz}
\par
These circuits are circuits that entangle four qubits as is clear from the right circuit; the left circuit is also a somewhat natural implementation of such a circuit that a compiler may emit.
The three consecutive \( \CNOT \) gates acting on the top two qubits of the left circuit, as well as those acting on the second and third qubits, are implementations of swap gates using \CNOT{} gates.
These swaps are often inserted to naively implement an operation that acts at a distance.
The left circuit above is just an implementation of the following circuit (followed by a simple optimization that removes two consecutive applications of the same \CNOT{} gate) that entangles the qubits by repeatedly applying \CNOT{} to the first qubit and the other qubits.
\[
\begin{quantikz}[column sep=0.5em,row sep=0.2em]
  \ket 0 & \gate{H} &  \permute{2,1} &               &           &               &          & \permute{2,1} & \ctrl{1}\\
  \ket 0 &          &                & \permute{2,1} &           & \permute{2,1} & \ctrl{1} &               & \targ{} \\
  \ket 0 &          &                &               &  \ctrl{1} &               & \targ{}  &               & \\
  \ket 0 &          &                &               &  \targ{}  &               &          &               & \\
\end{quantikz}
= \begin{quantikz}[column sep=0.5em,row sep=0.2em]
  \ket 0 & \gate{H} &  \ctrl{3} & \ctrl{2}         & \ctrl{1}         \\
  \ket 0 &          &           &                  & \targ{}        \\
  \ket 0 &          &           &  \targ{}         &  \\
  \ket 0 &          &  \targ{}  &                  &  \\
\end{quantikz}
\]
We show how, by using the circuit algebra given in Exampe~\ref{ex:optimization}, the effect system can capture that the number of gates of the produced circuit \emph{after being optimized}.
In Figure~\ref{fig:prog-for-LNN}, the types of circuits that are being applied are given as comments, where \( \mathit{Circ}(T, S)^\effOne \) is written as \textsf{Circ(T, S) [e]}.
Here we explain how each \( \effOne_i \) is defined.
The effect annotation \( \effOne_1 \) for the four qubit initialization is defined by \( \effOne_1(S) = ( \{ 0000 \}, 0) \). (We are not counting the initializations as gates.)
For the Hadamard gate, the effect annotation \( \effOne_2 \) is defined by \( \effOne_2(S)= (\{1, 0\}, 1) \); the first element is $\{ 1, 0\}$ because the result after applying the Hadamard gate is a superposition of base states, and the second element $1$ is the count.
The effect annotation \( \effOne_3 \) for the CNOT gate is the one that we explained in Example~\ref{ex:optimization}, which, in particular, satisfies \( e_3(\{00, 10\}) = ( \{00, 11\}, 1 ) \).
Note that this means we have \( q_1 = q_2 \) after line 6.
We can define the effect annotation \( \effOne_4\) for \texttt{cnot21cnot12} to satisfy \( e_4(\{ 00, 11\}) = (\{00, 11\}, 0)\) because \( (\CNOT_{12}\circ \CNOT_{21})(\ket{00}) = \ket{00}\) and \( (\CNOT_{12}\circ \CNOT_{21})(\ket{11}) = \ket{11}\).
That is, \texttt{cnot21cnot12} acts as identity and can be removed if \( q_1 = q_2 \).
Hence, the effect annotation \( \effOne \) (restricted to the first two qubits) for the program from line 1 to line 9, satisfies \( \effOne(\{\varepsilon \}) = (\{00, 11\}, 3)\) meaning that \( q_1 = q_2 \) and we only need three gates, as opposed to five, after the optimization.
\end{example}

\subsection{Categorical Semantics}

\subsubsection{Circuit Monad with Effect Annotation}
We refine the circuit monad from Section~\ref{sec:simple-type} by annotating it with effects.
An established approach to giving a semantics of a type-and-effect system is to use \emph{graded monads}~\cite{Katsumata14,Mellies12}.
We follow this approach and refine the circuit monad as a \emph{category-graded monad}~\cite{OrchardWE20}, which can be thought of as a many-sorted generalization of graded monads.
\begin{definition}[Cat-graded Monads~\cite{OrchardWE20}]
  \label{def:cat-graded-monad}
  A (preorder enriched) category-graded monad (or an \( \category A \)-graded monad) on \( \Set \) consists of a family of endofunctors \( \monad^f \colon \Set \to \Set \) indexed by morphisms \( f \) in \( \category A \)
  and families of natural transformations
  \begin{itemize}
	\item \( \unit_a \colon \Id_{\Set} \to  \monad^{\id_a}\) for \( a \in \obj(\category A) \),
	\item \( \mult_{f, g} \colon \monad^{f} \monad^{g} \to \monad^{f;g} \) for \( f \colon a \to b\) and \( g \colon b \to c \),
  \item  \( \monad^{f \ple f'} \colon \monad^{f} \to \monad^{f'} \) for \( f, f' \colon a \to b \) such that \( f \ple f' \),
  \end{itemize}
  satisfying the following unital and associativity laws.
  \[
  \begin{tikzcd}[row sep=large, column sep=large,wire types={n,n}]
	\monad^{f} \arrow[d,"\monad^f \unit_b"']\arrow[r,"\unit_a \monad^f"]\arrow[rd,equal]& \monad^{\id_a} \monad^f \arrow[d,"\mult_{\id_a, f}"]\\
	\monad^{f} \monad^{\id_b} \arrow[r,"\mult_{f, \id_b}"'] & \monad^f
\end{tikzcd}
\qquad
  \begin{tikzcd}[row sep=large, column sep=large,wire types={n,n}]
	\monad^f \monad^g \monad^h \arrow[d,"\mult_{f, g}\monad_h"']\arrow[r,"\monad^f \mult_{g, h}"] & \monad^f \monad^{g;h} \arrow[d,"\mult_{f, g;h}"]\\
	\monad^{f;g} \monad^h \arrow[r,"\mult_{f;g , h}"'] & \monad^{f;g;h}
  \end{tikzcd}
\]
Moreover, we have the following commutativity concerning the preordering
\[
  \begin{tikzcd}[row sep=large, column sep=large,wire types={n,n}]
	\monad^{f} \arrow[d,"\monad^{f \ple f} "'] \arrow[rd,"\id_{\monad^f}"] \\
	\monad^{f}  \arrow[r,equal] & \monad^f
\end{tikzcd}
\qquad
\begin{tikzcd}[row sep=large, column sep=large,wire types={n,n}]
	\monad^{f} \arrow[d,"\monad^{f \ple g} "'] \arrow[rd,"\monad^{f \ple h}"] \\
	\monad^g  \arrow[r,"\monad^{g \ple h}"'] & \monad^h
\end{tikzcd}
\qquad
  \begin{tikzcd}[row sep=large, column sep=large,, wire types={n,n}]
	\monad^{f} \monad^g \arrow[r,"\monad^{f \ple f'} \monad^{g \ple g'}"] \arrow[d,"\mult_{f,g}"] & \monad^{f'} \monad^{g'} \arrow[d,"\mult_{f',g'}"] \\
	\monad^{f;g}  \arrow[r,"\monad^{f;g \ple f';g'}"] & \monad^{f'; g'}
\end{tikzcd}
  \]
\end{definition}

We are interested in category-graded monads of a specific kind, namely those constructed from abstractions to circuit algebras.
Given \( \abstraction \colon \catM \to \catE \), we define an endofunctor on \( \Set \), parameterized by \( \effOne \), by \( \circmonadEff  {\effOne} (\setObjOne) \defeq \setObjOne \times \catM^{\ple \effOne}(\mObjOne, \mObjTwo) \) where the set \( \catM^{\ple \effOne}(\mObjOne, \mObjTwo) \subseteq \catM(\mObjOne, \mObjTwo)\) is defined as \( \{ \circuitOne \in \catM(\mObjOne, \mObjTwo) \mid \abstraction(\circuitOne) \ple \effOne \}\).
This construction is reminiscent of graded monads arising from effect observations ~\cite{Katsumata14}.\footnotemark
\footnotetext{Our recipe is different from Katsumata's in that (1) we deal with indexed monads rather than ordinary monads, and (2) it is tailored for a specific monad, i.e.~category-action monads.}
There is a caveat to this definition: this is not exactly an \( \catE \)-graded monad, as the grading is defined with respect to a slight modification of E.
The annotation \( \effOne \colon \eObjOne \to \eObjTwo \) does not tell us the type of circuits, but only the type of circuits \emph{after the abstraction}.
To remedy the problem, for the category of grades, we use \( \catEff \) whose objects are those of \( \catM \) and whose homset \( \catEff(\mObjOne, \mObjTwo)  \) is defined as \( \catE(\abstraction(\mObjOne), \abstraction(\mObjTwo)) \).

We spell out the definition of the \( \catEff \)-graded monad constructed from abstraction to circuit algebra.
As mentioned, the endofunctor \( \circmonadEff  {\effOne \colon \mObjOne \to \mObjTwo}\) acts on
an object \( \setObjOne \) as \( \setObjOne \times \catM^{\ple \effOne}(\mObjOne, \mObjTwo) \).
The unit and multiplication are defined exactly the same way as in the (ordinary) circuit monad.
That is, we have
\begin{align*}
  \unit_{\mObjOne, \setObjOne}(x) \defeq (x, \id_{\mObjOne}),
  \qquad \mult_{\effOne_1, \effOne_2, \setObjOne}((x, \circuitOne),\circuitTwo) \defeq (x, \circuitOne; \circuitTwo).
\end{align*}
Note that the multiplication is well-defined because if \( \abstraction(\circuitOne) \ple \effOne_1 \) and \( \abstraction(\circuitTwo) \ple \effOne_2 \), we have \( \abstraction(\circuitOne; \circuitTwo) = \abstraction(\circuitOne); \abstraction(\circuitTwo)  \ple \effOne_1; \effOne_2 \).
Each component of the natural transformation \( {\circmonad*}_{\setObjOne}^{\effOne_1 \ple \effOne_2} \) is the inclusion from \( \setObjOne \times \catM^{\ple \effOne_1}(\mObjOne, \mObjTwo) \) to
\( \setObjOne \times \catM^{\ple \effOne_2}(\mObjOne, \mObjTwo) \).
Moreover, the category-graded monad constructed this way, has premonoidal lifting as in the case for the ordinary circuit monad.
For example, there is a natural transformation \( (\mtypeOne \rtensor - )^{\dagger}_{\effOne\colon \mtypeOne \to \mtypeTwo, X} \colon  \circmonadEff  {\effOne} X \to \circmonadEff {\mtypeOne \rtensor \effOne} X \) that respects unit and the multiplication of the monad, which is defined as \( (x, \circuitOne) \mapsto (x, \mtypeOne \rtensor \circuitOne)\)

\subsubsection{Interpretation}
The interpretation of types can be obtained by replacing \( \catM(\mObjOne, \mObjTwo) \) with \( \catM^{\ple \effOne}(\mObjOne, \mObjTwo) \).
For example, we define
\begin{align*}
  \semP{\circt{\effOne \colon \eObjOne \to \eObjTwo}{\mtypeOne}{\mtypeTwo}}
  &\defeq \catM^{\ple \effOne}(\semM{\mtypeOne}, \semM{\mtypeTwo}) \\
  \sem{\arroweff \typeOne \typeTwo \mtypeOne {\effOne \colon \eObjTwo \to \eObjThree}}
  &\defeq (\functionSet {\flat \sem \typeOne}  {\flat \sem \typeTwo \times \catM^{\ple \effOne} (\sharp \sem \typeOne \otimes \semM \mtypeOne, \sharp \sem \typeTwo)}, \semM{\mtypeOne} ).
\end{align*}
The first element of \( \sem{\arroweff \typeOne \typeTwo \mtypeOne {\effOne \colon \eObjTwo \to \eObjThree}}
 \) can also be written as \( \functionSet {\flat \sem \typeOne}  {\circmonadEff {\effOne \colon \sharp \sem \typeOne \otimes \semM \mtypeOne \to \sharp \sem \typeTwo} (\flat \sem \typeTwo)} \).
The interpretation of type \( \bangeff \typeOne \effOne \) is similar to that of the function type.
Note that we have an isomorphism \( \flat \sem{\arroweff {\mtypeOne} {\mtypeTwo}{} {\effOne}} \cong \semP {\circt {\effOne} {\mtypeOne} {\mtypeTwo}}\) as before.\footnotemark
\footnotetext{By abuse of notation, we also denote this isomorphism as \( \boxsem \).}
The rest of the types are interpreted as in the case of the simple type system.

We now discuss how the judgments are interpreted.
Computational judgments of the shape \( \cjudgeff \contextOne \termOne \typeOne \effOne \) are interpreted as morphisms \( \sem \termOne \colon \flat \sem{\contextOne} \to \circmonadEff {\effOne \colon \sharp \sem{\contextOne} \to \sharp \sem{\typeOne}} \flat \sem \typeOne \) in \( \Set \).
Note that this is almost identical to the interpretation in the simple type system when the interpreted term is regarded as a morphism in \( \Set \).
Although the interpretation of the simply typed terms were given as morphisms in the Kleisli category, we define the interpretation of terms typed in the type-and-effect system in \( \Set \).
This is because the notion of Kleisli category for category-graded monads is tricky to define (see Remark~\ref{rem:Freyd-vs-monad} for a further discussion).
The value judgments remain to be interpreted as morphisms in \( \Set \times \disc(\obj(\catM)) \).

Let us look at the interpretation of \( \returnoperator \) and \( \letoperator \) since they highlight the use of the monad.
As usual, \( \returnoperator \) is interpreted as the interpretation of a value post-composed with the unit.
\begin{align*}
			\left\llbracket\infer{\vjudgst{\contextOne}{\valOne}{\typeOne} \quad \abstraction(\sharp \sem \typeOne) = \eObjOne }
			{\cjudgeff{\contextOne}{\return{\valOne}}{\typeOne} {\id_{\eObjOne} \colon \eObjOne \to \eObjOne}} \right\rrbracket
      \defeq \flat \sem{\vjudgst{\contextOne}{\valOne}{\typeOne} }; \unit_{\sharp \sem \typeOne}
\end{align*}
The interpretation of \( \letoperator \) is essentially the same as how composition is defined in the Kleisli category of a (standard) monad:
\begin{align*}
  &\left\llbracket
    \begin{matrix} %
      \infer{\cjudgeff{\pcontextOne,\contextOne_1}{\termOne}{\typeOne} {\effOne_1 \colon \eObjOne_1 \to \eObjOne_1'}
      \\
      \cjudgeff{\pcontextOne,\contextOne_2,\varOne:\typeOne}{\termTwo}{\typeTwo}  {\effOne_2 \colon \eObjOne_2 \to \eObjOne_2'} \\
        \refine {\sem{\sharp \contextOne_i}} {\eObjTwo_i}  \quad \effOne = (\id_{\eObjTwo_2} \rtensor \effOne_1); \effOne_2
      }
  {\cjudgeff{\pcontextOne,\contextOne_2,\contextOne_1}{\letin{\varOne}{\termOne}{\termTwo}}{\typeTwo} \effOne}
  \end{matrix}
    \right\rrbracket
\end{align*}
is given as
\begin{align*}
  \sem \pcontextOne \times \flat \sem {\contextOne_2} \times \flat \sem{\contextOne_1}
  &\xrightarrow{\Delta \times \id; \cong} \sem \pcontextOne \times \flat \sem {\contextOne_2} \times  \sem \pcontextOne \times \flat \sem{\contextOne_1} \\
  &\xrightarrow{\id \times \unit_{\eObjTwo_2} \times \sem \termOne} \sem \pcontextOne \times \circmonadEff {\id_{\eObjTwo_2}} \flat \sem{\contextOne_2} \times \circmonadEff {\effOne_1 } \flat \sem{\typeOne} \\
  &\xrightarrow{\id \times \rtensorTwo; \strength}
    \circmonadEff {\id \rtensor \effOne_1}(\sem \pcontextOne \times \flat \sem{\contextOne_2} \times \flat \sem \typeOne ) \\
  &\xrightarrow{\circmonadEff {\id \rtensor \effOne_1} \sem \termTwo}
     \circmonadEff {\id \rtensor \effOne_1}(\circmonadEff{\effOne_2}(\flat \sem \typeTwo) ) \\
  &\xrightarrow {\mult_{\id \rtensor \effOne_1, \effOne_2}} \circmonadEff {(\id \rtensor \effOne_1); \effOne_2}(\flat \sem \typeTwo).
\end{align*}
The morphism \( \tau \) is the strength (which can be defined in a straightforward way), and \( \rtensorTwo_{\mtypeOne_1, \mtypeTwo_1, \mtypeOne_2, \mtypeTwo_2}  \) is the morphism\footnotemark
\footnotetext{Aside from this elementary definition, \( \rtensorTwo \) can be defined using liftings of \( \id \rtensor - \) and \( - \ltensor \id \), and the strength and multiplication of the cat-graded monad.}
\begin{gather*}
  \rtensorTwo \colon \circmonadEff {\effOne_1 \colon \mtypeOne_1 \to \mtypeTwo_1} X \times \circmonadEff {\effOne_2 \colon \mtypeOne_2 \to \mtypeTwo_2 } Y \to  \circmonadEff {(\id \rtensor \effOne_1); (\effOne_2 \ltensor \id)  } (X \times Y) \\
  ((x, \circuitOne),  (y, \circuitTwo)) \mapsto ((x,y),  (\mtypeOne_2 \rtensor \circuitOne);(\circuitTwo \ltensor \mtypeTwo_1)).
\end{gather*}
Intuitively, \( \rtensorTwo \) pairs \( x, y \) and, at the same time,  ``parallelly composes'' \( \circuitOne \) and \( \circuitTwo \) in the order \( \circuitOne \) first followed by \( \circuitTwo \).
In the interpretation of \( \letoperator \), it is simply used to parallelly augment wires of the type \( \sharp \sem {\contextOne_2}\) to the circuit generated by \( \termOne \).

The interpretation for the other constructs is also essentially unchanged from the interpretation in Figure~\ref{fig:interpretation-judgment-st}.
A minor difference is that instead of using the functor \( J \) to map value morphisms into the category of computations, we use \( \unit \) as we did for the interpretation of \( \return \valOne \).
For example, we have
\begin{align*}
  \sem{\cjudgeff{\pcontextOne}{\boxt{\mtypeOne}{\valOne}}{\circt{\effOne\colon \eObjOne \to \eObjTwo}{\mtypeOne}{\mtypeTwo}} \nulleff}
  \defeq \flat \sem{\vjudgst{\pcontextOne}{\valOne}{{\arroweff{\mtypeOne}{\mtypeTwo}{I}{\effOne\colon \eObjOne \to \eObjTwo}}}}; \boxsem; \unit_{\munitt, \circt \effOne {\mtypeOne}{\mtypeTwo}}.
\end{align*}
The subsumption rule is interpreted by postcomposing the component of the natural transformation \( \circmonadEff {\effOne_1 \ple \effOne_2 }  \) at \( \sem \typeOne \).
(See Appendix~\ref{appx:interpretation-eff} for the interpretaion of the other constructs.)

\begin{remark}
  \label{rem:Freyd-vs-monad}
  While we used (parameterized) Freyd category in Section~\ref{sec:simple-type}, in this section, we had to define the interpretation explicitly using the structures of monads.
  Freyd categories have a better match with the fine-grained call-by-value calculus that \PQR{} is based on.
  The difficulty lies in defining a suitable notion of ``locally graded category''~\cite{Wood78,Levy19} for category graded monads.
  A locally graded category is a category-like structure whose homsets are indexed by grades (i.e.~elements of a preordered monoid), and this is central to the definition of graded Freyd categories~\cite{GaboardiKOS21}.
  Naively, we may define a category like structure in which homsets are indexed by morphisms, but determining the laws such a structure should satisfy seems non trivial and is left for future work.
\end{remark}

\subsection{Correctness}
Analogous to the case of simple types, the semantics is sound and adequate.
Interpretation of configurations are defined as in Section~\ref{sec:simple-type} by considering \( J \) as a map that associates \( (f,\circuitOne) \in (\Set \times \catM)((\setObjOne, \mObjOne),(\setObjTwo, \mObjTwo))\) to a morphism from \( X \) to \( \setObjTwo \times \catM^{\ple \abstraction(\circuitOne)}(\mObjOne, \mObjTwo) \) in \( \Set \).

\begin{restatable}[Soundness]{theorem}{soundnessEff}
  \label{thm:soundness-eff}
  Suppose that \( \configjudgmentEff \lcOne \circuitOne \termOne \typeOne {\lcOne'}  {\effOne \colon \eObjOne \to \eObjTwo} \) and \( \config \circuitOne \termOne \eval \config {\circuitOne'} \valOne \).
  Then \( \sem{\config \circuitOne \termOne} = J(\sem{\config {\circuitOne'} \valOne}) \).
  \qed
\end{restatable}
\begin{restatable}[Computational adequacy]{theorem}{computationalAdequacyEff}
  Suppose that \( \configjudgmentEff \emptycontext \circuitOne \termOne \unitt \emptycontext {\effOne : \eObjOne \to \eObjUnit} \) and \( \sem{\config \circuitOne \termOne} = J(\sem{\config \circuitTwo \valOne})\).
  Then \( \config \circuitOne \termOne \eval \config \circuitTwo \valOne \) (possibly up to renaming of labels).
  Moreover, there must exist a circuit \( \circuitThree \) such that \( \abstraction(\circuitThree) \ple \effOne \) and \( \circuitOne; \circuitThree = \circuitTwo \).
  \qed
\end{restatable}

\section{Discussion: Dependent Types}
\label{sec:dependent-type}

One feature that some of the languages of the \PQ{} family have is a form of dependent types~\cite{FuKRS20,FuKS22,ColledanDalLago24a,ColledanDalLago25}.
Dependent types are useful in the context of quantum circuit programming because one can express, at the level of types, the number of qubits a function takes as input.
For example, a function implementing the so-called Quantum Fourier Transform has the type \( \mathit{qft} \colon \prod (n : \natt),  \vect \qubitt n \multimap \vect \qubitt n \).
We briefly discuss the possibility of adding dependent types to the languages and models from the previous sections.

The denotational model given in Section~\ref{sec:simple-type} can be straightforwardly extended to support some dependent types.
We can apply the families construction as shown in Figure~\ref{fig:model-overview-dependent}.
Since ordinary adjunctions lift to fibered adjunctions over \( \Set \) by a pointwise definition, this construction gives an instance of a parameterized version of the fibered adjunction models~\cite{AhmanGP16}.
Fibered adjunction models are models for a language with dependent types and computational effects.
On the syntax side, this means allowing dependent types with ``value restriction'': types can only depend on values with parameter type, which correspond to morphisms in \( \Set \).
Although such a type system is less expressive compared to the type system of \PQD{}~\cite{FuKS22} that
allows types to depend on (the shape of) terms with quantum data types, it is expressive enough to express the type for \( \mathit{qft} \) we described above.

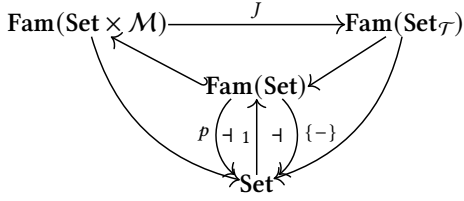
\begin{figure}
  \[
  \begin{tikzcd}[wire types = {n,n,n,n}]
    \Fam{\Set \times\catM} \arrow[rddd, bend right = 30] \arrow[rr, "J"]
    &
    & \Fam{\Set_{\mathcal T}} \arrow[lddd, bend left = 30] \arrow[ld, end anchor=east]
    \\ %
    & \Fam{\Set}
      \arrow[dd,"p"'{name=p}, bend right=50, start anchor={[xshift=-1ex]}, end anchor={[xshift=-1ex]}, pos=.45]
      \arrow[dd,"\{-\}"{name=comp}, bend left=50, start anchor={[xshift=1ex]}, end anchor={[xshift=1ex]}, pos=.45]
      \arrow[lu,"",hook', start anchor=west]
    &
    \\
    & &  \\
    & \Set  \arrow[uu,"1"{name=term}]&
    \arrow[phantom, from=p, to=term, "\dashv", style = {font = \scriptsize},pos=.6]
    \arrow[phantom, from=term, to=comp, "\dashv", style = {font = \scriptsize}]
  \end{tikzcd}
  \]
  \caption{Model for the dependently typed \PQR{}. The functors \( 1 \) and \( \{-\}\) are the terminal object functor and the comprehension functor, respectively. }
  \label{fig:model-overview-dependent}
\end{figure}

What is more challenging is to model effect annotations that are dependent on terms.
\citet{ColledanDalLago24a,ColledanDalLago25} considered a type system in which the function \( \mathit{qft} \) has the type \( \mathit{qft} \colon \prod (n : \natt),  \vect \qubitt n \stackrel{n \colon n \to n}{\multimap} \vect \qubitt n \), where the \( n \) over the arrow is the effect annotation expressing the estimated width of the circuit \( \mathit{qft} \) generates.\footnotemark
\footnotetext{The syntax we are using here is not exactly the syntax used in~\cite{ColledanDalLago24a,ColledanDalLago25}}
Their syntax forces annotations to be arithmetic terms depending on arithmetic variables so that type inference based on SMT-solving can semi-automatically infer the effect information.
If we are only interested in annotations that are arithmetic expressions, then it seems possible to write out the interpretation just by indexing the interpretation we described in Section~\ref{sec:effect-system} with (tuples of) natural numbers.
However, it is not clear (a) to what extent the effect annotations can be generalized and (b) whether there is a categorical explanation (e.g.~an explanation in the form of some generalization of fibered monads) that captures the nature of the interpretation.
Further investigation is left for future work.

\section{Related Work}
\label{sec:related-work}

Giving a denotational semantics to quantum programming languages has generally been more challenging than the corresponding problem for classical languages. Even in the case of imperative QRAM languages, this task is not so simple, given that quantum data can be interpreted by semantic domains which are different from the usual ones, see~\cite{Ying16} for an overview.

When, in addition to a quantum store, the underlying language is also endowed with higher-order functions, the task becomes even more complicated.
A few years after the quantum \( \lambda \)-calculus~\cite{SelingerValiron05} was introduced, a fully abstract model for the \emph{linear fragment} of the quantum \( \lambda \)-calculus was given by using the category \( \CPM \) of completely positive maps~\cite{SelingerValiron08}.
However, designing a model of (variants of) the full quantum \( \lambda \)-calculus remained a challenge.
One difficulty is the tension between the finite and the infinite.
The category \( \CPM \) or its subcategory \( \Q \) of trace-preserving completely positive maps are inherently finite since their definition relies on finite Hilbert spaces, whereas quantum \( \lambda \)-calculus comes with infinite features such as the \( ! \) modality or term level recursion.
To overcome this difficulty, various approaches have been studied.
These include (but are not limited to) studies based on presheaves~\cite{MalherbeSS13}, (\( \Sigma \)-monoid) enriched presheaves~\cite{TsukadaAsada24}, quantitative semantics of linear logic~\cite{PaganiSV14}, operator algbra~\cite{ChoWesterbaan16}, geometry of interaction~\cite{YoshimizuHFL14,HasuoHoshino17} and game semantics~\cite{ClairambaultdVW19,ClairambaultdeVisme20}.
The first four approaches can be considered as taking a certain ``completion'' of \( \CPM \) or \( \Q \) to support higher-order and infinite types.
The latter two approaches are more operational, since they are based on interactive semantics.
We note that some of these~\cite{PaganiSV14,ClairambaultdeVisme20,TsukadaAsada24} are fully abstract models of the full quantum \( \lambda \)-calculus.

In CDLs, the types of problems encountered in giving a denotational semantics are different.
On the one hand, a semantics of CDL needs to cope with the distinction between circuit \emph{generation} time and circuit \emph{execution} time that does not exist in languages such as the quantum \( \lambda \)-calculus. Moreover, circuits can, like in \Quipper, support specific operators which do not exist for other data structures and which allow programs to be
interpreted as circuits e.g.,  \Quipper's $\mathit{box}$ operator.
On the other hand, giving denotational semantics to CDLs is somehow easier because the characteristics of quantum circuits, say, compared to Boolean circuits, are abstracted away. This simplicity, however, no longer holds if the CDL allows an interplay between the host and the circuit level language, known as \emph{dynamic lifting}.

As already mentioned, since the introduction of the \PQ\ family with \PQS\ and \PQM, a presheaf-based denotational semantics has been known to be adequate~\cite{RiosSelinger17,LindenhoviusMZ18}.
Later, extensions of the language with dependent types~\cite{FuKS22} and dynamic lifting~\cite{FuKRS23,FuKRS22} have been considered.
The model for dependent types uses the families construction as discussed in Section~\ref{sec:dependent-type}.
At a superficial level, the model for dynamic lifting also shares an idea with the semantics we introduced in Section~\ref{sec:effect-system}.
Both models hold a morphism representing a quantum circuit and a morphism representing its interpretation, either as a quantum operation or effect annotation.
However, the foundations of these models remain fundamentally unaltered in that they use presheaves, and further comparison with our models is left for future work.
It should be noted that the members of the \PQ\ family intended for the analysis of circuit size, like \PQR{}~\cite{ColledanDalLago24a}, lack a denotational account, although being solidly grounded from an operational point of view. We believe this to be a result of the intrinsic difficulty of reflecting intensional properties of the underlying circuit in the aforementioned presheaf-based semantics.

Another commonly used type of CDLs, aside from the \PQ\ family, consists of those with a clear stratification between classical and quantum layers.
\QWIRE{}~\cite{PaykinRZ17} and \EWIRE{}~\cite{RennelaStaton19} are languages that have a dedicated language for circuits with a linear type system that can be embedded into a non-linear host language.
On the semantic side, this embedding has been nicely captured using enriched category theory~\cite{RennelaStaton19}.
\VQPL{}~\cite{JiaKLMZ22} is another quantum programming language with two subcalculi, one for classical programs and the other for quantum programs, where the quantum programs are more high-level than mere circuits.
They support rich features such as classical recursive types, inductive quantum types and dynamic lifting.
Furthermore, \VQPL{} has an adequate denotational model that unites domain-theoretic models of classical programming and von Neumann algebras for quantum interpretation.
Overall, these languages cannot have data structures with \emph{both} quantum data and classical data, unlike languages in the \PQ{} family. While our semantics also has a clear separation between \( \Set \) and the category of circuits, our languages allow \emph{mixing} classical types and quantum types. Our key observation here is that the introduction of the closure types allows us to decompose these mixed types into classical and quantum parts.

Since the pioneering works of \citet{Moggi89,Moggi91} it can certainly be said that the concept of a monad is one of the most powerful mathematical constructions in giving meaning to computational effects.
Various generalizations on plain monads, such as indexed (aka parameterized) monads~\cite{Atkey09}, graded monads~\cite{Katsumata14,Mellies12} or category-graded monads~\cite{OrchardWE20} have been proposed in the literature.
The circuit monads in Section~\ref{sec:categorica-semantics-st} and Section~\ref{sec:effect-system} can be seen as an instance of an indexed monad and category-graded monads, respectively.
Defining effect annotations as abstractions of computational effects is a general idea that could be applied to other things besides circuits.
Identifying the essences of the construction in Section~\ref{sec:effect-system} and deriving a recipe to construct category-graded monad might be  of independent interest.
Additionally, as discussed in Remark~\ref{rem:Freyd-vs-monad}, a ``Freyd category'' for category-graded monads appears to be absent from the literature and is worth investigating further.
The literature, by the way, offers some examples of circuit monads~\cite{Valiron16,Elliott13}, none of which has been applied to languages in the \PQ{} family.

\section{Conclusion}
\label{sec:conclusion}

In this work, we introduce a monadic denotational semantic model for circuit description languages in which the role of the monad is played by a circuit construction. This way, the structure of the produced circuit can be observed as an effect, making the model potentially capable of reflecting intensional features of circuits produced by terms of any type. This allows us to give semantics to members of the \PQ\ family for which a denotational account is not, to the authors' knowledge, known.
Remarkably, by considering an abstract notion of circuit algebra, different notions of circuit metrics are treated uniformly, including simple forms of circuit optimization; we believe that this abstraction may help reveal further concrete circuit metrics.

Among the future developments of this work, we must certainly mention the study of denotational semantic models for \PQRA, a recently introduced member of the \PQ\ family able to support the analysis of a wide range of circuit metrics via types.
Some of its features, such as effect and dependent typing, have already been treated separately (see Sections~\ref{sec:effect-system} and \ref{sec:dependent-type}), but the study of their combination in the same semantic framework is left to future work.
Although our work primarily focuses on denotational semantics, it is natural to consider implementing the type system described in  Section~\ref{sec:effect-system}, e.g., by adapting or extending the QuRA tool~\cite{ColledanDalLago25}.
While supporting the circuit algebra examples discussed in this paper should not pose major technical challenges, aspects such as efficient automatic type inference merit further investigation.
We plan to explore these directions in future work.
We would also like to extend the language and the model with features of \Quipper{} that we did not cover in this paper such as term and type level recursion and dynamic lifting.

\begin{acks}
  The research leading to these results has received funding from the
  MUR grant PRIN 2022 PNRR No. P2022HXNSC - ``Resource Awareness in Programming'' and
  European Union - NextGenerationEU through the Italian Ministry of University and Research under PNRR - M4C2 - I1.4 Project CN00000013 ``National Centre for HPC, Big Data and Quantum Computing''.
\end{acks}

\bibliographystyle{ACM-Reference-Format}
\bibliography{IEEEabrv,main}


\begin{thebibliography}{59}


\ifx \showCODEN    \undefined \def \showCODEN     #1{\unskip}     \fi
\ifx \showISBNx    \undefined \def \showISBNx     #1{\unskip}     \fi
\ifx \showISBNxiii \undefined \def \showISBNxiii  #1{\unskip}     \fi
\ifx \showISSN     \undefined \def \showISSN      #1{\unskip}     \fi
\ifx \showLCCN     \undefined \def \showLCCN      #1{\unskip}     \fi
\ifx \shownote     \undefined \def \shownote      #1{#1}          \fi
\ifx \showarticletitle \undefined \def \showarticletitle #1{#1}   \fi
\ifx \showURL      \undefined \def \showURL       {\relax}        \fi
\providecommand\bibfield[2]{#2}
\providecommand\bibinfo[2]{#2}
\providecommand\natexlab[1]{#1}
\providecommand\showeprint[2][]{arXiv:#2}

\bibitem[Ahman et~al\mbox{.}(2016)]%
        {AhmanGP16}
\bibfield{author}{\bibinfo{person}{Danel Ahman}, \bibinfo{person}{Neil Ghani},
  {and} \bibinfo{person}{Gordon~D. Plotkin}.} \bibinfo{year}{2016}\natexlab{}.
\newblock \showarticletitle{Dependent Types and Fibred Computational Effects}.
  In \bibinfo{booktitle}{\emph{Foundations of Software Science and Computation
  Structures - 19th International Conference, {FOSSACS} 2016, Held as Part of
  the European Joint Conferences on Theory and Practice of Software, {ETAPS}
  2016, Eindhoven, The Netherlands, April 2-8, 2016, Proceedings}}
  \emph{(\bibinfo{series}{Lecture Notes in Computer Science},
  Vol.~\bibinfo{volume}{9634})}, \bibfield{editor}{\bibinfo{person}{Bart
  Jacobs} {and} \bibinfo{person}{Christof L{\"{o}}ding}} (Eds.).
  \bibinfo{publisher}{Springer}, \bibinfo{pages}{36--54}.
\newblock
\href{https://doi.org/10.1007/978-3-662-49630-5\_3}{doi:\nolinkurl{10.1007/978-3-662-49630-5\_3}}


\bibitem[Amy(2019)]%
        {AmyTS19}
\bibfield{author}{\bibinfo{person}{Matthew Amy}.}
  \bibinfo{year}{2019}\natexlab{}.
\newblock \showarticletitle{Sized Types for Low-Level Quantum Metaprogramming}.
  In \bibinfo{booktitle}{\emph{Reversible Computation - 11th International
  Conference, {RC} 2019, Lausanne, Switzerland, June 24-25, 2019, Proceedings}}
  \emph{(\bibinfo{series}{Lecture Notes in Computer Science},
  Vol.~\bibinfo{volume}{11497})},
  \bibfield{editor}{\bibinfo{person}{Michael~Kirkedal Thomsen} {and}
  \bibinfo{person}{Mathias Soeken}} (Eds.). \bibinfo{publisher}{Springer},
  \bibinfo{pages}{87--107}.
\newblock
\href{https://doi.org/10.1007/978-3-030-21500-2\_6}{doi:\nolinkurl{10.1007/978-3-030-21500-2\_6}}


\bibitem[Atkey(2009)]%
        {Atkey09}
\bibfield{author}{\bibinfo{person}{Robert Atkey}.}
  \bibinfo{year}{2009}\natexlab{}.
\newblock \showarticletitle{Parameterised notions of computation}.
\newblock \bibinfo{journal}{\emph{J. Funct. Program.}} \bibinfo{volume}{19},
  \bibinfo{number}{3-4} (\bibinfo{year}{2009}), \bibinfo{pages}{335--376}.
\newblock
\href{https://doi.org/10.1017/S095679680900728X}{doi:\nolinkurl{10.1017/S095679680900728X}}


\bibitem[Benton(1994)]%
        {Benton94}
\bibfield{author}{\bibinfo{person}{P.~N. Benton}.}
  \bibinfo{year}{1994}\natexlab{}.
\newblock \showarticletitle{A Mixed Linear and Non-Linear Logic: Proofs, Terms
  and Models (Extended Abstract)}. In \bibinfo{booktitle}{\emph{Computer
  Science Logic, 8th International Workshop, {CSL} '94, Kazimierz, Poland,
  September 25-30, 1994, Selected Papers}} \emph{(\bibinfo{series}{Lecture
  Notes in Computer Science}, Vol.~\bibinfo{volume}{933})},
  \bibfield{editor}{\bibinfo{person}{Leszek Pacholski} {and}
  \bibinfo{person}{Jerzy Tiuryn}} (Eds.). \bibinfo{publisher}{Springer},
  \bibinfo{pages}{121--135}.
\newblock
\href{https://doi.org/10.1007/BFB0022251}{doi:\nolinkurl{10.1007/BFB0022251}}


\bibitem[Bettelli et~al\mbox{.}(2003)]%
        {BettelliCS03}
\bibfield{author}{\bibinfo{person}{S. Bettelli}, \bibinfo{person}{T. Calarco},
  {and} \bibinfo{person}{L. Serafini}.} \bibinfo{year}{2003}\natexlab{}.
\newblock \showarticletitle{Toward an architecture for quantum programming}.
\newblock \bibinfo{journal}{\emph{The European Physical Journal D - Atomic,
  Molecular and Optical Physics}} \bibinfo{volume}{25}, \bibinfo{number}{2}
  (\bibinfo{date}{Aug.} \bibinfo{year}{2003}), \bibinfo{pages}{181–200}.
\newblock
\showISSN{1434-6079}
\href{https://doi.org/10.1140/epjd/e2003-00242-2}{doi:\nolinkurl{10.1140/epjd/e2003-00242-2}}


\bibitem[Cho and Westerbaan(2016)]%
        {ChoWesterbaan16}
\bibfield{author}{\bibinfo{person}{Kenta Cho} {and} \bibinfo{person}{Abraham
  Westerbaan}.} \bibinfo{year}{2016}\natexlab{}.
\newblock \showarticletitle{Von Neumann Algebras form a Model for the Quantum
  Lambda Calculus}.
\newblock \bibinfo{journal}{\emph{CoRR}}  \bibinfo{volume}{abs/1603.02133}
  (\bibinfo{year}{2016}).
\newblock
\showeprint[arXiv]{1603.02133}
\urldef\tempurl%
\url{http://arxiv.org/abs/1603.02133}
\showURL{%
\tempurl}


\bibitem[Clairambault and de~Visme(2020)]%
        {ClairambaultdeVisme20}
\bibfield{author}{\bibinfo{person}{Pierre Clairambault} {and}
  \bibinfo{person}{Marc de Visme}.} \bibinfo{year}{2020}\natexlab{}.
\newblock \showarticletitle{Full abstraction for the quantum lambda-calculus}.
\newblock \bibinfo{journal}{\emph{Proc. {ACM} Program. Lang.}}
  \bibinfo{volume}{4}, \bibinfo{number}{{POPL}} (\bibinfo{year}{2020}),
  \bibinfo{pages}{63:1--63:28}.
\newblock
\href{https://doi.org/10.1145/3371131}{doi:\nolinkurl{10.1145/3371131}}


\bibitem[Clairambault et~al\mbox{.}(2019)]%
        {ClairambaultdVW19}
\bibfield{author}{\bibinfo{person}{Pierre Clairambault}, \bibinfo{person}{Marc
  de Visme}, {and} \bibinfo{person}{Glynn Winskel}.}
  \bibinfo{year}{2019}\natexlab{}.
\newblock \showarticletitle{Game semantics for quantum programming}.
\newblock \bibinfo{journal}{\emph{Proc. {ACM} Program. Lang.}}
  \bibinfo{volume}{3}, \bibinfo{number}{{POPL}} (\bibinfo{year}{2019}),
  \bibinfo{pages}{32:1--32:29}.
\newblock
\href{https://doi.org/10.1145/3290345}{doi:\nolinkurl{10.1145/3290345}}


\bibitem[Colledan and Dal~Lago(2023)]%
        {ColledanDalLago23}
\bibfield{author}{\bibinfo{person}{Andrea Colledan} {and} \bibinfo{person}{Ugo
  Dal~Lago}.} \bibinfo{year}{2023}\natexlab{}.
\newblock \showarticletitle{{On Dynamic Lifting and Effect Typing in Circuit
  Description Languages}}. In \bibinfo{booktitle}{\emph{28th International
  Conference on Types for Proofs and Programs (TYPES 2022)}}
  \emph{(\bibinfo{series}{Leibniz International Proceedings in Informatics
  (LIPIcs)}, Vol.~\bibinfo{volume}{269})},
  \bibfield{editor}{\bibinfo{person}{Delia Kesner} {and}
  \bibinfo{person}{Pierre-Marie P\'{e}drot}} (Eds.).
  \bibinfo{publisher}{Schloss Dagstuhl -- Leibniz-Zentrum f{\"u}r Informatik},
  \bibinfo{address}{Dagstuhl, Germany}, \bibinfo{pages}{3:1--3:21}.
\newblock
\showISBNx{978-3-95977-285-3}
\showISSN{1868-8969}
\href{https://doi.org/10.4230/LIPIcs.TYPES.2022.3}{doi:\nolinkurl{10.4230/LIPIcs.TYPES.2022.3}}


\bibitem[Colledan and Dal~Lago(2024)]%
        {ColledanDalLago24a}
\bibfield{author}{\bibinfo{person}{Andrea Colledan} {and} \bibinfo{person}{Ugo
  Dal~Lago}.} \bibinfo{year}{2024}\natexlab{}.
\newblock \showarticletitle{Circuit Width Estimation via Effect Typing and
  Linear Dependency}. In \bibinfo{booktitle}{\emph{Programming Languages and
  Systems - 33rd European Symposium on Programming, {ESOP} 2024, Held as Part
  of the European Joint Conferences on Theory and Practice of Software, {ETAPS}
  2024, Luxembourg City, Luxembourg, April 6-11, 2024, Proceedings, Part {II}}}
  \emph{(\bibinfo{series}{Lecture Notes in Computer Science},
  Vol.~\bibinfo{volume}{14577})}, \bibfield{editor}{\bibinfo{person}{Stephanie
  Weirich}} (Ed.). \bibinfo{publisher}{Springer}, \bibinfo{pages}{3--30}.
\newblock
\href{https://doi.org/10.1007/978-3-031-57267-8\_1}{doi:\nolinkurl{10.1007/978-3-031-57267-8\_1}}


\bibitem[Colledan and Dal~Lago(2025)]%
        {ColledanDalLago25}
\bibfield{author}{\bibinfo{person}{Andrea Colledan} {and} \bibinfo{person}{Ugo
  Dal~Lago}.} \bibinfo{year}{2025}\natexlab{}.
\newblock \showarticletitle{Flexible Type-Based Resource Estimation in Quantum
  Circuit Description Languages}.
\newblock \bibinfo{journal}{\emph{Proc. ACM Program. Lang.}}
  \bibinfo{volume}{9}, \bibinfo{number}{POPL}, Article \bibinfo{articleno}{47}
  (\bibinfo{date}{Jan.} \bibinfo{year}{2025}), \bibinfo{numpages}{31}~pages.
\newblock
\href{https://doi.org/10.1145/3704883}{doi:\nolinkurl{10.1145/3704883}}


\bibitem[Developers(2025)]%
        {Circ25}
\bibfield{author}{\bibinfo{person}{Cirq Developers}.}
  \bibinfo{year}{2025}\natexlab{}.
\newblock \bibinfo{booktitle}{\emph{Cirq}}.
\newblock
\href{https://doi.org/10.5281/zenodo.15191735}{doi:\nolinkurl{10.5281/zenodo.15191735}}


\bibitem[Elliott(2013)]%
        {Elliott13}
\bibfield{author}{\bibinfo{person}{Conal Elliott}.}
  \bibinfo{year}{2013}\natexlab{}.
\newblock \bibinfo{title}{Circuits as a bicartesian closed category}.
\newblock
  \bibinfo{howpublished}{\url{http://conal.net/blog/posts/circuits-as-a-bicartesian-closed-category}}.
\newblock
\newblock
\shownote{Blog post, Accessed: 2025-01-23}.


\bibitem[Fu et~al\mbox{.}(2020)]%
        {FuKRS20}
\bibfield{author}{\bibinfo{person}{Peng Fu}, \bibinfo{person}{Kohei Kishida},
  \bibinfo{person}{Neil~J. Ross}, {and} \bibinfo{person}{Peter Selinger}.}
  \bibinfo{year}{2020}\natexlab{}.
\newblock \showarticletitle{A Tutorial Introduction to Quantum Circuit
  Programming in Dependently Typed {Proto-Quipper}}. In
  \bibinfo{booktitle}{\emph{Reversible Computation - 12th International
  Conference, {RC} 2020, Oslo, Norway, July 9-10, 2020, Proceedings}}
  \emph{(\bibinfo{series}{Lecture Notes in Computer Science},
  Vol.~\bibinfo{volume}{12227})}, \bibfield{editor}{\bibinfo{person}{Ivan
  Lanese} {and} \bibinfo{person}{Mariusz Rawski}} (Eds.).
  \bibinfo{publisher}{Springer}, \bibinfo{pages}{153--168}.
\newblock
\href{https://doi.org/10.1007/978-3-030-52482-1\_9}{doi:\nolinkurl{10.1007/978-3-030-52482-1\_9}}


\bibitem[Fu et~al\mbox{.}(2022b)]%
        {FuKRS22}
\bibfield{author}{\bibinfo{person}{Peng Fu}, \bibinfo{person}{Kohei Kishida},
  \bibinfo{person}{Neil~J. Ross}, {and} \bibinfo{person}{Peter Selinger}.}
  \bibinfo{year}{2022}\natexlab{b}.
\newblock \showarticletitle{A biset-enriched categorical model for
  {Proto-Quipper} with dynamic lifting}. In
  \bibinfo{booktitle}{\emph{Proceedings 19th International Conference on
  Quantum Physics and Logic, {QPL} 2022, Wolfson College, Oxford, UK, 27 June -
  1 July 2022}} \emph{(\bibinfo{series}{{EPTCS}}, Vol.~\bibinfo{volume}{394})},
  \bibfield{editor}{\bibinfo{person}{Stefano Gogioso} {and}
  \bibinfo{person}{Matty Hoban}} (Eds.). \bibinfo{pages}{302--342}.
\newblock
\href{https://doi.org/10.4204/EPTCS.394.16}{doi:\nolinkurl{10.4204/EPTCS.394.16}}


\bibitem[Fu et~al\mbox{.}(2023)]%
        {FuKRS23}
\bibfield{author}{\bibinfo{person}{Peng Fu}, \bibinfo{person}{Kohei Kishida},
  \bibinfo{person}{Neil~J. Ross}, {and} \bibinfo{person}{Peter Selinger}.}
  \bibinfo{year}{2023}\natexlab{}.
\newblock \showarticletitle{{Proto-Quipper} with Dynamic Lifting}.
\newblock \bibinfo{journal}{\emph{Proc. {ACM} Program. Lang.}}
  \bibinfo{volume}{7}, \bibinfo{number}{{POPL}} (\bibinfo{year}{2023}),
  \bibinfo{pages}{309--334}.
\newblock
\href{https://doi.org/10.1145/3571204}{doi:\nolinkurl{10.1145/3571204}}


\bibitem[Fu et~al\mbox{.}(2022a)]%
        {FuKS22}
\bibfield{author}{\bibinfo{person}{Peng Fu}, \bibinfo{person}{Kohei Kishida},
  {and} \bibinfo{person}{Peter Selinger}.} \bibinfo{year}{2022}\natexlab{a}.
\newblock \showarticletitle{Linear Dependent Type Theory for Quantum
  Programming Languages}.
\newblock \bibinfo{journal}{\emph{Log. Methods Comput. Sci.}}
  \bibinfo{volume}{18}, \bibinfo{number}{3} (\bibinfo{year}{2022}).
\newblock
\href{https://doi.org/10.46298/LMCS-18(3:28)2022}{doi:\nolinkurl{10.46298/LMCS-18(3:28)2022}}


\bibitem[Gaboardi et~al\mbox{.}(2021)]%
        {GaboardiKOS21}
\bibfield{author}{\bibinfo{person}{Marco Gaboardi}, \bibinfo{person}{Shin{-}ya
  Katsumata}, \bibinfo{person}{Dominic Orchard}, {and} \bibinfo{person}{Tetsuya
  Sato}.} \bibinfo{year}{2021}\natexlab{}.
\newblock \showarticletitle{Graded Hoare Logic and its Categorical Semantics}.
  In \bibinfo{booktitle}{\emph{Programming Languages and Systems - 30th
  European Symposium on Programming, {ESOP} 2021, Held as Part of the European
  Joint Conferences on Theory and Practice of Software, {ETAPS} 2021,
  Luxembourg City, Luxembourg, March 27 - April 1, 2021, Proceedings}}
  \emph{(\bibinfo{series}{Lecture Notes in Computer Science},
  Vol.~\bibinfo{volume}{12648})}, \bibfield{editor}{\bibinfo{person}{Nobuko
  Yoshida}} (Ed.). \bibinfo{publisher}{Springer}, \bibinfo{pages}{234--263}.
\newblock
\href{https://doi.org/10.1007/978-3-030-72019-3\_9}{doi:\nolinkurl{10.1007/978-3-030-72019-3\_9}}


\bibitem[Green et~al\mbox{.}(2013)]%
        {GreenLRSV13}
\bibfield{author}{\bibinfo{person}{Alexander~S. Green},
  \bibinfo{person}{Peter~LeFanu Lumsdaine}, \bibinfo{person}{Neil~J. Ross},
  \bibinfo{person}{Peter Selinger}, {and} \bibinfo{person}{Beno{\^{\i}}t
  Valiron}.} \bibinfo{year}{2013}\natexlab{}.
\newblock \showarticletitle{Quipper: a scalable quantum programming language}.
  In \bibinfo{booktitle}{\emph{{ACM} {SIGPLAN} Conference on Programming
  Language Design and Implementation, {PLDI} '13, Seattle, WA, USA, June 16-19,
  2013}}, \bibfield{editor}{\bibinfo{person}{Hans{-}Juergen Boehm} {and}
  \bibinfo{person}{Cormac Flanagan}} (Eds.). \bibinfo{publisher}{{ACM}},
  \bibinfo{pages}{333--342}.
\newblock
\href{https://doi.org/10.1145/2491956.2462177}{doi:\nolinkurl{10.1145/2491956.2462177}}


\bibitem[H{\"{a}}ner et~al\mbox{.}(2020)]%
        {HanerHT20}
\bibfield{author}{\bibinfo{person}{Thomas H{\"{a}}ner},
  \bibinfo{person}{Torsten Hoefler}, {and} \bibinfo{person}{Matthias Troyer}.}
  \bibinfo{year}{2020}\natexlab{}.
\newblock \showarticletitle{Assertion-based optimization of Quantum programs}.
\newblock \bibinfo{journal}{\emph{Proc. {ACM} Program. Lang.}}
  \bibinfo{volume}{4}, \bibinfo{number}{{OOPSLA}} (\bibinfo{year}{2020}),
  \bibinfo{pages}{133:1--133:20}.
\newblock
\href{https://doi.org/10.1145/3428201}{doi:\nolinkurl{10.1145/3428201}}


\bibitem[Hasuo and Hoshino(2017)]%
        {HasuoHoshino17}
\bibfield{author}{\bibinfo{person}{Ichiro Hasuo} {and} \bibinfo{person}{Naohiko
  Hoshino}.} \bibinfo{year}{2017}\natexlab{}.
\newblock \showarticletitle{Semantics of higher-order quantum computation via
  geometry of interaction}.
\newblock \bibinfo{journal}{\emph{Ann. Pure Appl. Log.}} \bibinfo{volume}{168},
  \bibinfo{number}{2} (\bibinfo{year}{2017}), \bibinfo{pages}{404--469}.
\newblock
\href{https://doi.org/10.1016/J.APAL.2016.10.010}{doi:\nolinkurl{10.1016/J.APAL.2016.10.010}}


\bibitem[Javadi-Abhari et~al\mbox{.}(2024)]%
        {JavadiAbhariTKWLGMNBCJG24}
\bibfield{author}{\bibinfo{person}{Ali Javadi-Abhari}, \bibinfo{person}{Matthew
  Treinish}, \bibinfo{person}{Kevin Krsulich}, \bibinfo{person}{Christopher~J.
  Wood}, \bibinfo{person}{Jake Lishman}, \bibinfo{person}{Julien Gacon},
  \bibinfo{person}{Simon Martiel}, \bibinfo{person}{Paul~D. Nation},
  \bibinfo{person}{Lev~S. Bishop}, \bibinfo{person}{Andrew~W. Cross},
  \bibinfo{person}{Blake~R. Johnson}, {and} \bibinfo{person}{Jay~M. Gambetta}.}
  \bibinfo{year}{2024}\natexlab{}.
\newblock \bibinfo{title}{Quantum computing with Qiskit}.
\newblock
\showeprint[arxiv]{2405.08810}~[quant-ph]
\urldef\tempurl%
\url{https://arxiv.org/abs/2405.08810}
\showURL{%
\tempurl}


\bibitem[Jia et~al\mbox{.}(2022)]%
        {JiaKLMZ22}
\bibfield{author}{\bibinfo{person}{Xiaodong Jia}, \bibinfo{person}{Andre
  Kornell}, \bibinfo{person}{Bert Lindenhovius}, \bibinfo{person}{Michael~W.
  Mislove}, {and} \bibinfo{person}{Vladimir Zamdzhiev}.}
  \bibinfo{year}{2022}\natexlab{}.
\newblock \showarticletitle{Semantics for variational Quantum programming}.
\newblock \bibinfo{journal}{\emph{Proc. {ACM} Program. Lang.}}
  \bibinfo{volume}{6}, \bibinfo{number}{{POPL}} (\bibinfo{year}{2022}),
  \bibinfo{pages}{1--31}.
\newblock
\href{https://doi.org/10.1145/3498687}{doi:\nolinkurl{10.1145/3498687}}


\bibitem[Jiang(2024)]%
        {Jiang24}
\bibfield{author}{\bibinfo{person}{Hanru Jiang}.}
  \bibinfo{year}{2024}\natexlab{}.
\newblock \showarticletitle{Qubit Recycling Revisited}.
\newblock \bibinfo{journal}{\emph{Proc. {ACM} Program. Lang.}}
  \bibinfo{volume}{8}, \bibinfo{number}{{PLDI}} (\bibinfo{year}{2024}),
  \bibinfo{pages}{1264--1287}.
\newblock
\href{https://doi.org/10.1145/3656428}{doi:\nolinkurl{10.1145/3656428}}


\bibitem[Joyal and Street(1991)]%
        {JoyalStreet91}
\bibfield{author}{\bibinfo{person}{André Joyal} {and} \bibinfo{person}{Ross
  Street}.} \bibinfo{year}{1991}\natexlab{}.
\newblock \showarticletitle{The geometry of tensor calculus, I}.
\newblock \bibinfo{journal}{\emph{Advances in Mathematics}}
  \bibinfo{volume}{88}, \bibinfo{number}{1} (\bibinfo{year}{1991}),
  \bibinfo{pages}{55--112}.
\newblock
\showISSN{0001-8708}
\href{https://doi.org/10.1016/0001-8708(91)90003-P}{doi:\nolinkurl{10.1016/0001-8708(91)90003-P}}


\bibitem[Katsumata(2014)]%
        {Katsumata14}
\bibfield{author}{\bibinfo{person}{Shin{-}ya Katsumata}.}
  \bibinfo{year}{2014}\natexlab{}.
\newblock \showarticletitle{Parametric effect monads and semantics of effect
  systems}. In \bibinfo{booktitle}{\emph{The 41st Annual {ACM} {SIGPLAN-SIGACT}
  Symposium on Principles of Programming Languages, {POPL} '14, San Diego, CA,
  USA, January 20-21, 2014}}, \bibfield{editor}{\bibinfo{person}{Suresh
  Jagannathan} {and} \bibinfo{person}{Peter Sewell}} (Eds.).
  \bibinfo{publisher}{{ACM}}, \bibinfo{pages}{633--646}.
\newblock
\href{https://doi.org/10.1145/2535838.2535846}{doi:\nolinkurl{10.1145/2535838.2535846}}


\bibitem[Knill(2022)]%
        {Knill22}
\bibfield{author}{\bibinfo{person}{E. Knill}.} \bibinfo{year}{2022}\natexlab{}.
\newblock \bibinfo{title}{Conventions for Quantum Pseudocode}.
\newblock
\showeprint[arxiv]{2211.02559}~[quant-ph]
\urldef\tempurl%
\url{https://arxiv.org/abs/2211.02559}
\showURL{%
\tempurl}
\newblock
\shownote{[{A} version of LANL report LAUR-96-2724 (1996) with a modern
  formatting]}.


\bibitem[Levy(2019)]%
        {Levy19}
\bibfield{author}{\bibinfo{person}{Paul~Blain Levy}.}
  \bibinfo{year}{2019}\natexlab{}.
\newblock \bibinfo{title}{Locally graded categories}.
\newblock
  \bibinfo{howpublished}{\url{https://pblevy.github.io/papers/locgrade.pdf}}.
\newblock
\newblock
\shownote{Talk slides}.


\bibitem[Levy et~al\mbox{.}(2003)]%
        {LevyPT03}
\bibfield{author}{\bibinfo{person}{Paul~Blain Levy}, \bibinfo{person}{John
  Power}, {and} \bibinfo{person}{Hayo Thielecke}.}
  \bibinfo{year}{2003}\natexlab{}.
\newblock \showarticletitle{Modelling environments in call-by-value programming
  languages}.
\newblock \bibinfo{journal}{\emph{Inf. Comput.}} \bibinfo{volume}{185},
  \bibinfo{number}{2} (\bibinfo{year}{2003}), \bibinfo{pages}{182--210}.
\newblock
\href{https://doi.org/10.1016/S0890-5401(03)00088-9}{doi:\nolinkurl{10.1016/S0890-5401(03)00088-9}}


\bibitem[Lindenhovius et~al\mbox{.}(2018)]%
        {LindenhoviusMZ18}
\bibfield{author}{\bibinfo{person}{Bert Lindenhovius},
  \bibinfo{person}{Michael~W. Mislove}, {and} \bibinfo{person}{Vladimir
  Zamdzhiev}.} \bibinfo{year}{2018}\natexlab{}.
\newblock \showarticletitle{Enriching a Linear/Non-linear Lambda Calculus: {A}
  Programming Language for String Diagrams}. In
  \bibinfo{booktitle}{\emph{Proceedings of the 33rd Annual {ACM/IEEE} Symposium
  on Logic in Computer Science, {LICS} 2018, Oxford, UK, July 09-12, 2018}},
  \bibfield{editor}{\bibinfo{person}{Anuj Dawar} {and} \bibinfo{person}{Erich
  Gr{\"{a}}del}} (Eds.). \bibinfo{publisher}{{ACM}}, \bibinfo{pages}{659--668}.
\newblock
\href{https://doi.org/10.1145/3209108.3209196}{doi:\nolinkurl{10.1145/3209108.3209196}}


\bibitem[Liu et~al\mbox{.}(2019)]%
        {LiuZWYLLYZ19}
\bibfield{author}{\bibinfo{person}{Junyi Liu}, \bibinfo{person}{Bohua Zhan},
  \bibinfo{person}{Shuling Wang}, \bibinfo{person}{Shenggang Ying},
  \bibinfo{person}{Tao Liu}, \bibinfo{person}{Yangjia Li},
  \bibinfo{person}{Mingsheng Ying}, {and} \bibinfo{person}{Naijun Zhan}.}
  \bibinfo{year}{2019}\natexlab{}.
\newblock \showarticletitle{Quantum Hoare Logic}.
\newblock \bibinfo{journal}{\emph{Arch. Formal Proofs}}  \bibinfo{volume}{2019}
  (\bibinfo{year}{2019}).
\newblock
\urldef\tempurl%
\url{https://www.isa-afp.org/entries/QHLProver.html}
\showURL{%
\tempurl}


\bibitem[Malherbe et~al\mbox{.}(2013)]%
        {MalherbeSS13}
\bibfield{author}{\bibinfo{person}{Octavio Malherbe},
  \bibinfo{person}{Philip~J. Scott}, {and} \bibinfo{person}{Peter Selinger}.}
  \bibinfo{year}{2013}\natexlab{}.
\newblock \showarticletitle{Presheaf Models of Quantum Computation: An
  Outline}. In \bibinfo{booktitle}{\emph{Computation, Logic, Games, and Quantum
  Foundations. The Many Facets of Samson Abramsky - Essays Dedicated to Samson
  Abramsky on the Occasion of His 60th Birthday}}
  \emph{(\bibinfo{series}{Lecture Notes in Computer Science},
  Vol.~\bibinfo{volume}{7860})}, \bibfield{editor}{\bibinfo{person}{Bob
  Coecke}, \bibinfo{person}{Luke Ong}, {and} \bibinfo{person}{Prakash
  Panangaden}} (Eds.). \bibinfo{publisher}{Springer},
  \bibinfo{pages}{178--194}.
\newblock
\href{https://doi.org/10.1007/978-3-642-38164-5\_13}{doi:\nolinkurl{10.1007/978-3-642-38164-5\_13}}


\bibitem[Mellies(2012)]%
        {Mellies12}
\bibfield{author}{\bibinfo{person}{Paul-Andr{\'e} Mellies}.}
  \bibinfo{year}{2012}\natexlab{}.
\newblock \bibinfo{title}{Parametric monads and enriched adjunctions}.
\newblock
  \bibinfo{howpublished}{\url{https://www.irif.fr/~mellies/tensorial-logic/8-parametric-monads-and-enriched-adjunctions.pdf}}.
\newblock
\newblock
\shownote{Preprint}.


\bibitem[Moggi(1989)]%
        {Moggi89}
\bibfield{author}{\bibinfo{person}{Eugenio Moggi}.}
  \bibinfo{year}{1989}\natexlab{}.
\newblock \showarticletitle{Computational Lambda-Calculus and Monads}. In
  \bibinfo{booktitle}{\emph{Proceedings of the Fourth Annual Symposium on Logic
  in Computer Science {(LICS} '89), Pacific Grove, California, USA, June 5-8,
  1989}}. \bibinfo{publisher}{{IEEE} Computer Society},
  \bibinfo{pages}{14--23}.
\newblock
\href{https://doi.org/10.1109/LICS.1989.39155}{doi:\nolinkurl{10.1109/LICS.1989.39155}}


\bibitem[Moggi(1991)]%
        {Moggi91}
\bibfield{author}{\bibinfo{person}{Eugenio Moggi}.}
  \bibinfo{year}{1991}\natexlab{}.
\newblock \showarticletitle{Notions of Computation and Monads}.
\newblock \bibinfo{journal}{\emph{Inf. Comput.}} \bibinfo{volume}{93},
  \bibinfo{number}{1} (\bibinfo{year}{1991}), \bibinfo{pages}{55--92}.
\newblock
\href{https://doi.org/10.1016/0890-5401(91)90052-4}{doi:\nolinkurl{10.1016/0890-5401(91)90052-4}}


\bibitem[Orchard et~al\mbox{.}(2020)]%
        {OrchardWE20}
\bibfield{author}{\bibinfo{person}{Dominic Orchard}, \bibinfo{person}{Philip
  Wadler}, {and} \bibinfo{person}{Harley~Eades III}.}
  \bibinfo{year}{2020}\natexlab{}.
\newblock \showarticletitle{Unifying graded and parameterised monads}. In
  \bibinfo{booktitle}{\emph{Proceedings Eighth Workshop on Mathematically
  Structured Functional Programming, MSFP@ETAPS 2020, Dublin, Ireland, 25th
  April 2020}} \emph{(\bibinfo{series}{{EPTCS}}, Vol.~\bibinfo{volume}{317})},
  \bibfield{editor}{\bibinfo{person}{Max~S. New} {and} \bibinfo{person}{Sam
  Lindley}} (Eds.). \bibinfo{pages}{18--38}.
\newblock
\href{https://doi.org/10.4204/EPTCS.317.2}{doi:\nolinkurl{10.4204/EPTCS.317.2}}


\bibitem[Pagani et~al\mbox{.}(2014)]%
        {PaganiSV14}
\bibfield{author}{\bibinfo{person}{Michele Pagani}, \bibinfo{person}{Peter
  Selinger}, {and} \bibinfo{person}{Beno{\^{\i}}t Valiron}.}
  \bibinfo{year}{2014}\natexlab{}.
\newblock \showarticletitle{Applying quantitative semantics to higher-order
  quantum computing}. In \bibinfo{booktitle}{\emph{The 41st Annual {ACM}
  {SIGPLAN-SIGACT} Symposium on Principles of Programming Languages, {POPL}
  '14, San Diego, CA, USA, January 20-21, 2014}},
  \bibfield{editor}{\bibinfo{person}{Suresh Jagannathan} {and}
  \bibinfo{person}{Peter Sewell}} (Eds.). \bibinfo{publisher}{{ACM}},
  \bibinfo{pages}{647--658}.
\newblock
\href{https://doi.org/10.1145/2535838.2535879}{doi:\nolinkurl{10.1145/2535838.2535879}}


\bibitem[Paykin et~al\mbox{.}(2017)]%
        {PaykinRZ17}
\bibfield{author}{\bibinfo{person}{Jennifer Paykin}, \bibinfo{person}{Robert
  Rand}, {and} \bibinfo{person}{Steve Zdancewic}.}
  \bibinfo{year}{2017}\natexlab{}.
\newblock \showarticletitle{{QWIRE:} a core language for quantum circuits}. In
  \bibinfo{booktitle}{\emph{Proceedings of the 44th {ACM} {SIGPLAN} Symposium
  on Principles of Programming Languages, {POPL} 2017, Paris, France, January
  18-20, 2017}}, \bibfield{editor}{\bibinfo{person}{Giuseppe Castagna} {and}
  \bibinfo{person}{Andrew~D. Gordon}} (Eds.). \bibinfo{publisher}{{ACM}},
  \bibinfo{pages}{846--858}.
\newblock
\href{https://doi.org/10.1145/3009837.3009894}{doi:\nolinkurl{10.1145/3009837.3009894}}


\bibitem[Perdrix(2008)]%
        {Perdrix08}
\bibfield{author}{\bibinfo{person}{Simon Perdrix}.}
  \bibinfo{year}{2008}\natexlab{}.
\newblock \showarticletitle{Quantum Entanglement Analysis Based on Abstract
  Interpretation}. In \bibinfo{booktitle}{\emph{Static Analysis, 15th
  International Symposium, {SAS} 2008, Valencia, Spain, July 16-18, 2008.
  Proceedings}} \emph{(\bibinfo{series}{Lecture Notes in Computer Science},
  Vol.~\bibinfo{volume}{5079})},
  \bibfield{editor}{\bibinfo{person}{Mar{\'{\i}}a Alpuente} {and}
  \bibinfo{person}{Germ{\'{a}}n Vidal}} (Eds.). \bibinfo{publisher}{Springer},
  \bibinfo{pages}{270--282}.
\newblock
\href{https://doi.org/10.1007/978-3-540-69166-2\_18}{doi:\nolinkurl{10.1007/978-3-540-69166-2\_18}}


\bibitem[Power and Robinson(1997)]%
        {PowerRobinson97}
\bibfield{author}{\bibinfo{person}{John Power} {and} \bibinfo{person}{Edmund
  Robinson}.} \bibinfo{year}{1997}\natexlab{}.
\newblock \showarticletitle{Premonoidal Categories and Notions of Computation}.
\newblock \bibinfo{journal}{\emph{Math. Struct. Comput. Sci.}}
  \bibinfo{volume}{7}, \bibinfo{number}{5} (\bibinfo{year}{1997}),
  \bibinfo{pages}{453--468}.
\newblock
\href{https://doi.org/10.1017/S0960129597002375}{doi:\nolinkurl{10.1017/S0960129597002375}}


\bibitem[{Qiskit API reference}(2025)]%
        {HoareOptimizer}
\bibfield{author}{\bibinfo{person}{{Qiskit API reference}}.}
  \bibinfo{year}{2025}\natexlab{}.
\newblock \bibinfo{title}{{HoareOptimizer}}.
\newblock
  \bibinfo{howpublished}{\url{https://quantum.cloud.ibm.com/docs/en/api/qiskit/qiskit.transpiler.passes.HoareOptimizer}}.
\newblock
\newblock
\shownote{Accessed: 2025-07-10}.


\bibitem[Rennela and Staton(2020)]%
        {RennelaStaton19}
\bibfield{author}{\bibinfo{person}{Mathys Rennela} {and} \bibinfo{person}{Sam
  Staton}.} \bibinfo{year}{2020}\natexlab{}.
\newblock \showarticletitle{Classical Control, Quantum Circuits and Linear
  Logic in Enriched Category Theory}.
\newblock \bibinfo{journal}{\emph{Log. Methods Comput. Sci.}}
  \bibinfo{volume}{16}, \bibinfo{number}{1} (\bibinfo{year}{2020}).
\newblock
\href{https://doi.org/10.23638/LMCS-16(1:30)2020}{doi:\nolinkurl{10.23638/LMCS-16(1:30)2020}}


\bibitem[Rios and Selinger(2017)]%
        {RiosSelinger17}
\bibfield{author}{\bibinfo{person}{Francisco Rios} {and} \bibinfo{person}{Peter
  Selinger}.} \bibinfo{year}{2017}\natexlab{}.
\newblock \showarticletitle{A categorical model for a quantum circuit
  description language}. In \bibinfo{booktitle}{\emph{Proceedings 14th
  International Conference on Quantum Physics and Logic, {QPL} 2017, Nijmegen,
  The Netherlands, 3-7 July 2017}} \emph{(\bibinfo{series}{{EPTCS}},
  Vol.~\bibinfo{volume}{266})}, \bibfield{editor}{\bibinfo{person}{Bob Coecke}
  {and} \bibinfo{person}{Aleks Kissinger}} (Eds.). \bibinfo{pages}{164--178}.
\newblock
\href{https://doi.org/10.4204/EPTCS.266.11}{doi:\nolinkurl{10.4204/EPTCS.266.11}}


\bibitem[Román and Sobociński(2025)]%
        {RomanSobocinski25}
\bibfield{author}{\bibinfo{person}{Mario Román} {and} \bibinfo{person}{Paweł
  Sobociński}.} \bibinfo{year}{2025}\natexlab{}.
\newblock \showarticletitle{String Diagrams for Premonoidal Categories}.
\newblock \bibinfo{journal}{\emph{Logical Methods in Computer Science}}
  \bibinfo{volume}{Volume 21, Issue 2}, Article \bibinfo{articleno}{9}
  (\bibinfo{date}{Apr} \bibinfo{year}{2025}).
\newblock
\showISSN{1860-5974}
\href{https://doi.org/10.46298/lmcs-21(2:9)2025}{doi:\nolinkurl{10.46298/lmcs-21(2:9)2025}}


\bibitem[Sanders and Zuliani(2000)]%
        {SandersZ00}
\bibfield{author}{\bibinfo{person}{J.~W. Sanders} {and} \bibinfo{person}{P.
  Zuliani}.} \bibinfo{year}{2000}\natexlab{}.
\newblock \showarticletitle{Quantum Programming}. In
  \bibinfo{booktitle}{\emph{Mathematics of Program Construction}},
  \bibfield{editor}{\bibinfo{person}{Roland Backhouse} {and}
  \bibinfo{person}{Jos{\'e}~Nuno Oliveira}} (Eds.).
  \bibinfo{publisher}{Springer Berlin Heidelberg}, \bibinfo{address}{Berlin,
  Heidelberg}, \bibinfo{pages}{80--99}.
\newblock
\showISBNx{978-3-540-45025-2}


\bibitem[Selinger(2004)]%
        {Selinger04}
\bibfield{author}{\bibinfo{person}{Peter Selinger}.}
  \bibinfo{year}{2004}\natexlab{}.
\newblock \showarticletitle{Towards a quantum programming language}.
\newblock \bibinfo{journal}{\emph{Mathematical. Structures in Comp. Sci.}}
  \bibinfo{volume}{14}, \bibinfo{number}{4} (\bibinfo{date}{Aug.}
  \bibinfo{year}{2004}), \bibinfo{pages}{527–586}.
\newblock
\showISSN{0960-1295}
\href{https://doi.org/10.1017/S0960129504004256}{doi:\nolinkurl{10.1017/S0960129504004256}}


\bibitem[Selinger and Valiron(2005)]%
        {SelingerValiron05}
\bibfield{author}{\bibinfo{person}{Peter Selinger} {and}
  \bibinfo{person}{Beno{\^i}t Valiron}.} \bibinfo{year}{2005}\natexlab{}.
\newblock \showarticletitle{A Lambda Calculus for Quantum Computation with
  Classical Control}. In \bibinfo{booktitle}{\emph{Typed Lambda Calculi and
  Applications}}, \bibfield{editor}{\bibinfo{person}{Pawe{\l} Urzyczyn}} (Ed.).
  \bibinfo{publisher}{Springer Berlin Heidelberg}, \bibinfo{address}{Berlin,
  Heidelberg}, \bibinfo{pages}{354--368}.
\newblock
\showISBNx{978-3-540-32014-2}


\bibitem[Selinger and Valiron(2008)]%
        {SelingerValiron08}
\bibfield{author}{\bibinfo{person}{Peter Selinger} {and}
  \bibinfo{person}{Beno{\^{\i}}t Valiron}.} \bibinfo{year}{2008}\natexlab{}.
\newblock \showarticletitle{On a Fully Abstract Model for a Quantum Linear
  Functional Language: (Extended Abstract)}. In
  \bibinfo{booktitle}{\emph{Proceedings of the 4th International Workshop on
  Quantum Programming Languages, {QPL} 2006, Oxford, UK, July 17-19, 2006}}
  \emph{(\bibinfo{series}{Electronic Notes in Theoretical Computer Science},
  Vol.~\bibinfo{volume}{210})}, \bibfield{editor}{\bibinfo{person}{Peter
  Selinger}} (Ed.). \bibinfo{publisher}{Elsevier}, \bibinfo{pages}{123--137}.
\newblock
\href{https://doi.org/10.1016/J.ENTCS.2008.04.022}{doi:\nolinkurl{10.1016/J.ENTCS.2008.04.022}}


\bibitem[Shor(1994)]%
        {Shor94}
\bibfield{author}{\bibinfo{person}{Peter~W. Shor}.}
  \bibinfo{year}{1994}\natexlab{}.
\newblock \showarticletitle{Algorithms for Quantum Computation: Discrete
  Logarithms and Factoring}. In \bibinfo{booktitle}{\emph{35th Annual Symposium
  on Foundations of Computer Science, Santa Fe, New Mexico, USA, 20-22 November
  1994}}. \bibinfo{publisher}{{IEEE} Computer Society},
  \bibinfo{pages}{124--134}.
\newblock
\href{https://doi.org/10.1109/SFCS.1994.365700}{doi:\nolinkurl{10.1109/SFCS.1994.365700}}


\bibitem[Staton and Levy(2013)]%
        {StatonLevy13}
\bibfield{author}{\bibinfo{person}{Sam Staton} {and}
  \bibinfo{person}{Paul~Blain Levy}.} \bibinfo{year}{2013}\natexlab{}.
\newblock \showarticletitle{Universal properties of impure programming
  languages}. In \bibinfo{booktitle}{\emph{The 40th Annual {ACM}
  {SIGPLAN-SIGACT} Symposium on Principles of Programming Languages, {POPL}
  '13, Rome, Italy - January 23 - 25, 2013}},
  \bibfield{editor}{\bibinfo{person}{Roberto Giacobazzi} {and}
  \bibinfo{person}{Radhia Cousot}} (Eds.). \bibinfo{publisher}{{ACM}},
  \bibinfo{pages}{179--192}.
\newblock
\href{https://doi.org/10.1145/2429069.2429091}{doi:\nolinkurl{10.1145/2429069.2429091}}


\bibitem[Svore et~al\mbox{.}(2018)]%
        {SvoreGTAGHR18}
\bibfield{author}{\bibinfo{person}{Krysta Svore}, \bibinfo{person}{Alan
  Geller}, \bibinfo{person}{Matthias Troyer}, \bibinfo{person}{John Azariah},
  \bibinfo{person}{Christopher Granade}, \bibinfo{person}{Bettina Heim},
  \bibinfo{person}{Vadym Kliuchnikov}, \bibinfo{person}{Mariia Mykhailova},
  \bibinfo{person}{Andres Paz}, {and} \bibinfo{person}{Martin Roetteler}.}
  \bibinfo{year}{2018}\natexlab{}.
\newblock \showarticletitle{{Q\#: Enabling Scalable Quantum Computing and
  Development with a High-level DSL}}. In \bibinfo{booktitle}{\emph{Proceedings
  of the Real World Domain Specific Languages Workshop 2018}}
  \emph{(\bibinfo{series}{RWDSL2018})}. \bibinfo{publisher}{ACM}.
\newblock
\href{https://doi.org/10.1145/3183895.3183901}{doi:\nolinkurl{10.1145/3183895.3183901}}


\bibitem[Tsukada and Asada(2024)]%
        {TsukadaAsada24}
\bibfield{author}{\bibinfo{person}{Takeshi Tsukada} {and}
  \bibinfo{person}{Kazuyuki Asada}.} \bibinfo{year}{2024}\natexlab{}.
\newblock \showarticletitle{Enriched Presheaf Model of Quantum {FPC}}.
\newblock \bibinfo{journal}{\emph{Proc. {ACM} Program. Lang.}}
  \bibinfo{volume}{8}, \bibinfo{number}{{POPL}} (\bibinfo{year}{2024}),
  \bibinfo{pages}{362--392}.
\newblock
\href{https://doi.org/10.1145/3632855}{doi:\nolinkurl{10.1145/3632855}}


\bibitem[Valiron(2016)]%
        {Valiron16}
\bibfield{author}{\bibinfo{person}{Beno{\^{\i}}t Valiron}.}
  \bibinfo{year}{2016}\natexlab{}.
\newblock \showarticletitle{Generating Reversible Circuits from Higher-Order
  Functional Programs}. In \bibinfo{booktitle}{\emph{Reversible Computation -
  8th International Conference, {RC} 2016, Bologna, Italy, July 7-8, 2016,
  Proceedings}} \emph{(\bibinfo{series}{Lecture Notes in Computer Science},
  Vol.~\bibinfo{volume}{9720})}, \bibfield{editor}{\bibinfo{person}{Simon~J.
  Devitt} {and} \bibinfo{person}{Ivan Lanese}} (Eds.).
  \bibinfo{publisher}{Springer}, \bibinfo{pages}{289--306}.
\newblock
\href{https://doi.org/10.1007/978-3-319-40578-0\_21}{doi:\nolinkurl{10.1007/978-3-319-40578-0\_21}}


\bibitem[Wood(1978)]%
        {Wood78}
\bibfield{author}{\bibinfo{person}{R.~J. Wood}.}
  \bibinfo{year}{1978}\natexlab{}.
\newblock \showarticletitle{V-indexed categories}. In
  \bibinfo{booktitle}{\emph{Indexed Categories and Their Applications}}.
  \bibinfo{publisher}{Springer Berlin Heidelberg}, \bibinfo{address}{Berlin,
  Heidelberg}, \bibinfo{pages}{126--140}.
\newblock
\showISBNx{978-3-540-35762-9}


\bibitem[Ying(2016)]%
        {Ying16}
\bibfield{author}{\bibinfo{person}{Mingsheng Ying}.}
  \bibinfo{year}{2016}\natexlab{}.
\newblock \bibinfo{booktitle}{\emph{Foundations of Quantum Programming}
  (\bibinfo{edition}{1st} ed.)}.
\newblock \bibinfo{publisher}{Morgan Kaufmann Publishers Inc.},
  \bibinfo{address}{San Francisco, CA, USA}.
\newblock
\showISBNx{0128023066}


\bibitem[Ying et~al\mbox{.}(2017)]%
        {YingYW17}
\bibfield{author}{\bibinfo{person}{Mingsheng Ying}, \bibinfo{person}{Shenggang
  Ying}, {and} \bibinfo{person}{Xiaodi Wu}.} \bibinfo{year}{2017}\natexlab{}.
\newblock \showarticletitle{Invariants of quantum programs: characterisations
  and generation}. In \bibinfo{booktitle}{\emph{Proceedings of the 44th {ACM}
  {SIGPLAN} Symposium on Principles of Programming Languages, {POPL} 2017,
  Paris, France, January 18-20, 2017}},
  \bibfield{editor}{\bibinfo{person}{Giuseppe Castagna} {and}
  \bibinfo{person}{Andrew~D. Gordon}} (Eds.). \bibinfo{publisher}{{ACM}},
  \bibinfo{pages}{818--832}.
\newblock
\href{https://doi.org/10.1145/3009837.3009840}{doi:\nolinkurl{10.1145/3009837.3009840}}


\bibitem[Yoshimizu et~al\mbox{.}(2014)]%
        {YoshimizuHFL14}
\bibfield{author}{\bibinfo{person}{Akira Yoshimizu}, \bibinfo{person}{Ichiro
  Hasuo}, \bibinfo{person}{Claudia Faggian}, {and} \bibinfo{person}{Ugo~Dal
  Lago}.} \bibinfo{year}{2014}\natexlab{}.
\newblock \showarticletitle{Measurements in Proof Nets as Higher-Order Quantum
  Circuits}. In \bibinfo{booktitle}{\emph{Programming Languages and Systems -
  23rd European Symposium on Programming, {ESOP} 2014, Held as Part of the
  European Joint Conferences on Theory and Practice of Software, {ETAPS} 2014,
  Grenoble, France, April 5-13, 2014, Proceedings}}
  \emph{(\bibinfo{series}{Lecture Notes in Computer Science},
  Vol.~\bibinfo{volume}{8410})}, \bibfield{editor}{\bibinfo{person}{Zhong
  Shao}} (Ed.). \bibinfo{publisher}{Springer}, \bibinfo{pages}{371--391}.
\newblock
\href{https://doi.org/10.1007/978-3-642-54833-8\_20}{doi:\nolinkurl{10.1007/978-3-642-54833-8\_20}}


\bibitem[Yu and Palsberg(2021)]%
        {YuPalsberg21}
\bibfield{author}{\bibinfo{person}{Nengkun Yu} {and} \bibinfo{person}{Jens
  Palsberg}.} \bibinfo{year}{2021}\natexlab{}.
\newblock \showarticletitle{Quantum abstract interpretation}. In
  \bibinfo{booktitle}{\emph{{PLDI} '21: 42nd {ACM} {SIGPLAN} International
  Conference on Programming Language Design and Implementation, Virtual Event,
  Canada, June 20-25, 2021}}, \bibfield{editor}{\bibinfo{person}{Stephen~N.
  Freund} {and} \bibinfo{person}{Eran Yahav}} (Eds.).
  \bibinfo{publisher}{{ACM}}, \bibinfo{pages}{542--558}.
\newblock
\href{https://doi.org/10.1145/3453483.3454061}{doi:\nolinkurl{10.1145/3453483.3454061}}


\bibitem[Zhou et~al\mbox{.}(2019)]%
        {ZhouYY19}
\bibfield{author}{\bibinfo{person}{Li Zhou}, \bibinfo{person}{Nengkun Yu},
  {and} \bibinfo{person}{Mingsheng Ying}.} \bibinfo{year}{2019}\natexlab{}.
\newblock \showarticletitle{An applied quantum Hoare logic}. In
  \bibinfo{booktitle}{\emph{Proceedings of the 40th {ACM} {SIGPLAN} Conference
  on Programming Language Design and Implementation, {PLDI} 2019, Phoenix, AZ,
  USA, June 22-26, 2019}}, \bibfield{editor}{\bibinfo{person}{Kathryn~S.
  McKinley} {and} \bibinfo{person}{Kathleen Fisher}} (Eds.).
  \bibinfo{publisher}{{ACM}}, \bibinfo{pages}{1149--1162}.
\newblock
\href{https://doi.org/10.1145/3314221.3314584}{doi:\nolinkurl{10.1145/3314221.3314584}}


\end{thebibliography}
\clearpage
\appendix
\section{Soundness \& Adequacy}
\label{appx:st-proofs}
We prove the soundness and the computational adequacy of the interpretation given in Section~\ref{sec:simple-type}.
The soundness and the adequacy for the interpretation given in Section~\ref{sec:effect-system} can be shown by almost the same argument (and thus we do not repeat the argument).

\subsection{Soundness}
As usual, the soundness is proved by induction on the derivation of the evaluation relation.
The proof also uses the following standard substitution lemma.
\begin{lemma}[Substitution]
  \label{lem:subst-st}
  Suppose that \( \vjudgst {\pcontextOne, \contextOne_1}  \valOne \typeOne \).
  \begin{enumerate}
    \item If \( \vjudgst {\pcontextOne, \contextOne_2, x: \typeOne, \contextOne'_2} \valTwo \typeTwo \), then
      \begin{align*}
        \sem{\vjudgst {\pcontextOne, \contextOne_2, \contextOne_1, \contextOne_2'} {\valTwo \sub \valOne x} \typeTwo} =
        \begin{aligned}[t]
          &\sem \pcontextOne \otimes \sem {\contextOne_2} \otimes \sem {\contextOne_1} \otimes \sem {\contextOne_2'} \xrightarrow{(\Delta, \id) \otimes \id} \\
          & \sem \pcontextOne \otimes \sem \pcontextOne \otimes \sem {\contextOne_2} \otimes \sem {\contextOne_1} \otimes \sem {\contextOne_2'} \xrightarrow{\cong} \\
          & \sem \pcontextOne \otimes \sem {\contextOne_2} \otimes \sem \pcontextOne \otimes \sem {\contextOne_1} \otimes \sem {\contextOne_2'} \xrightarrow{\id \otimes \sem \valOne \otimes \id} \\
          & \sem \pcontextOne \otimes \sem {\contextOne_2} \otimes \sem \typeOne \otimes \sem {\contextOne_2'} \xrightarrow{\sem \valTwo} \sem \typeTwo.
        \end{aligned}
      \end{align*}
    \item If \( \cjudgst {\pcontextOne, \contextOne, x : \typeOne} \termOne \typeTwo \), then
      \begin{align*}
        \sem{\vjudgst {\pcontextOne, \contextOne, \contextOne_1} {\termOne \sub \valOne x} \typeTwo} =
        \begin{aligned}[t]
          &\sem \pcontextOne \otimes \sem \contextOne \otimes \sem {\contextOne_1}  \xrightarrow{\dup \otimes \id} \\
          &\sem \pcontextOne \otimes \sem \pcontextOne \otimes \sem \contextOne \otimes \sem {\contextOne_1} \xrightarrow{\cong} \\
          & \sem \pcontextOne \otimes \sem \contextOne \otimes \sem \pcontextOne \otimes \sem {\contextOne_1} \xrightarrow{\id \otimes J(\sem \valOne)} \\
          & \sem \pcontextOne \otimes \sem \contextOne \otimes  \sem \typeOne \xrightarrow{\sem \termOne} \sem \typeTwo.
        \end{aligned}
      \end{align*}

  \end{enumerate}
\end{lemma}
\begin{proof}
  By a simultaneous induction on the structure of the type derivation.
\end{proof}

\soundnessSt*
\begin{proof}
  By induction on the derivation of \( \config \circuitOne \termOne \eval \config {\circuitOne'} \valOne \).
  Throughout the proof we write \( \lcOne_{\termOne} \) (resp.~\(\lcOne_{\valOne} \)) for the label context that types the term \( \termOne \) (resp.~\( \valOne \)); this means that we have \( \lcOne = \lcOne_\termOne, \lcOne' \) and similarly for \( \valOne\).

  \begin{description}
    \item[Case \textit{app}:]
      We have \( \termOne = (\lambda x. \termTwo) \; \valTwo \) with \( \config \circuitOne {\termTwo \sub \valTwo x} \eval \config \circuitTwo \valOne \).
      It suffices to prove that \( \sem{(\lambda x. \termTwo) \, W} = \sem{\termTwo \sub \valTwo x} \) because then we can close this case by the induction hypothesis.
      By inversion on the typing rule, we must have
      \begin{gather*}
        \cjudgst {\lcOne_\termTwo, x : \typeTwo} \termTwo \typeOne
        \quad \text{and} \quad
        \vjudgst {\lcOne_\valTwo} \valTwo \typeTwo
      \end{gather*}
      with \( Q = Q_\termTwo, Q_\valTwo \) and \( \sharp \sem{Q_\valTwo} = \sharp \sem{\typeTwo} \).
      \begin{align*}
        &\sem{(\lambda x. \termTwo) \, W} \\
        &= (J((\Lambda(\sem{\cjudgst {\lcOne_\termTwo, x : \typeTwo} \termTwo \typeOne}), \id_{\semM {\sharp(\lcOne_\termTwo, x : \typeTwo)}})) \otimes J(\sem{\vjudgst {\lcOne_\valTwo} \valTwo \typeTwo})); \ev \tag{by def.} \\
        &= (J(\Lambda(\sem{\cjudgst {\lcOne_\termTwo, x : \typeTwo} \termTwo \typeOne}), \semM {\sharp(\lcOne_\termTwo, x : \typeTwo)}) \otimes J(\sem{\vjudgst {\lcOne_\valTwo} \valTwo \typeTwo})); \ev \tag{by def.~of the functor \(J(-, \mObjOne)\)}\\
        &= (J(\id_{\lcOne_\termTwo} \otimes \sem{\vjudgst {\lcOne_\valTwo} \valTwo \typeTwo})) ;  \sem{\cjudgst {\lcOne_\termTwo, x : \typeTwo} \termTwo \typeOne} \tag{by the universaity of \( \ev \)}\\
        &= \sem{\termTwo \sub \valTwo x} \tag{by substituion lemma (Lemma~\ref{lem:subst-st})}
      \end{align*}
    \item[Case \textit{dest}:]
      Similar to \textit{let}. %
    \item[Case \textit{force}:]
      Similar to the case \text{app}, it suffices to show that \( \sem{\force {(\lift \termTwo)}} = \sem{\termTwo} \), for some \( \termTwo \) such that \( \termOne = \force {(\lift \termTwo )}\).
      Since \( \lift \) is a special case of abstraction (i.e.~thunking) and \( \force \) is a special case for application, this can be proved as in the case of \textit{app}.
    \item[Case \textit{apply}:]
      In this case, we have \( \termOne = \apply{\boxedCirc{\struct\labOne}{\circuitThree}{\struct{\labTwo}}}{\struct{\labThree}} \) for some circuit \( \circuitThree \) and wire bundles \( \struct \labOne \), \( \struct \labTwo \) and \( \struct \labThree \).
      Moreover, we have \( \circuitOne ; (\circuitThree \ltensor \semM{\lcOne'}) = \circuitOne' \).
      By inversion on the typing rules, we must have
      \begin{gather*}
        \vjudgst \emptycontext {\boxedCirc{\struct\labOne}{\circuitThree}{\struct{\labTwo}}} {\circt {} \mtypeOne \mtypeTwo}
        \quad \text{and} \quad
        \vjudgst  {\struct \labThree : \mtypeOne} {\struct \labThree} \mtypeOne
      \end{gather*}
      with \( \vjudgst {\lcOne_\termOne}  {\struct \labThree} \mtypeOne \) and \( \typeOne = \mtypeTwo \) for some \( \mtypeOne \) and \( \mtypeTwo \).
      We first show that \( \sem{\apply{\boxedCirc{\struct\labOne}{\circuitThree}{\struct{\labTwo}}}{\struct{\labThree}}} = J(\id_1, \circuitThree)\).
      By definition, we have
      \begin{align*}
        \sem{\apply{\boxedCirc{\struct\labOne}{\circuitThree}{\struct{\labTwo}}}{\struct{\labThree}}}
        &= (\sem{\boxedCirc{\struct\labOne}{\circuitThree}{\struct{\labTwo}}} \otimes \id_{\lcOne_\termOne}); \applysem \\
        &= ( (J(\Lambda(\id_1, \semM{\boxedCirc{\struct\labOne}{\circuitThree}{\struct{\labTwo}}})) ;J(\boxsem, \id_I)) \otimes \id_{\lcOne_\termOne}); (J(\boxsem^{-1}, \id_I) \otimes \id_{\lcOne_\termOne});\ev \\
        &= ( (J(\Lambda(\id_1, \semM{\boxedCirc{\struct\labOne}{\circuitThree}{\struct{\labTwo}}})) \otimes \id_{\lcOne_\termOne});\ev \\
        &= J(\id_1, \semM{\boxedCirc{\struct\labOne}{\circuitThree}{\struct{\labTwo}}} ).
      \end{align*}
      Therefore,
      \begin{align*}
        \sem{\config \circuitOne \termOne}
        &= J(\id_1, \circuitOne);(\sem{\apply{\boxedCirc{\struct\labOne}{\circuitThree}{\struct{\labTwo}}}{\struct{\labThree}}} \otimes \id_{\lcOne'}) \\
        &= J(\id_1, \circuitOne);(J(\id_1, \semM{\boxedCirc{\struct\labOne}{\circuitThree}{\struct{\labTwo}}}) \otimes J(\id_1, \id_{\lcOne'})) \\
        &= J(\id_1, \circuitOne);J(\id_1, \semM{\boxedCirc{\struct\labOne}{\circuitThree}{\struct{\labTwo}}}  \ltensor \sem{\lcOne'}) \\
        &= J(\id_1, \circuitOne'; \perm) \\
        &= \sem{(\circuitOne', \struct \labTwo)}
      \end{align*}
      as desired. Here, \( \perm \) is the isomorphism induced by the codomain of \( \circuitOne' \) and type of \( \struct \labTwo \).
    \item[Case \textit{box}:]
      It must be the case that \( \termOne = \boxt \mtypeOne \valTwo \) and \( \valOne =  \boxedCirc{\struct\labOne}{\circuitTwo}{\struct{\labTwo}} \) for some value \( \valTwo\), circuit \( \circuitTwo \), and wire bundles \( \struct \labOne \) and \( \struct \labTwo \), and moreover, \( \config {\id_\lcOne} {\app \valTwo {\struct \labOne}} \eval \config \circuitTwo {\struct \labTwo }\).
      By inversion on the typing rules, we have
      \[
      \vjudgst \emptycontext \valTwo {\arrowst \mtypeOne \mtypeTwo {}} \quad \text{and} \quad  \vjudgst \emptycontext  {\boxedCirc{\struct\labOne}{\circuitTwo}{\struct{\labTwo}}} {\circt{} \mtypeOne \mtypeTwo}
      \]
      where \( \typeOne = \circt{} \mtypeOne \mtypeTwo \).
      Since the underlying circuit does not change during the reduction, it suffices to show that \( \sem{\boxt \mtypeOne \valTwo} = J(\sem{\boxedCirc{\struct\labOne}{\circuitTwo}{\struct{\labTwo}}})\).
      Moreover, since  \( \sem{\boxt \mtypeOne \valTwo} = J(\sem{\valTwo}); J(\boxsem, \id_I) \) and \( J(\sem{\boxedCirc{\struct\labOne}{\circuitTwo}{\struct{\labTwo}}}) = J(\widehat{\semM{\boxedCirc{\struct \labOne} \circuitTwo {\struct \labTwo}}}, \id_I); J(\boxsem, \id_I)\) by the definitions, we only need to show that \( \sem{\valTwo}  = (\widehat{\semM{\boxedCirc{\struct \labOne} \circuitTwo {\struct \labTwo}}}, \id_I) \).
      By the induction hypothesis, we have
      \begin{align*}
      \sem{\config {\id_\lcOne} {\app \valTwo {\struct \labOne}}}  = \sem{\app \valTwo {\struct \labOne}}  = J(\id_1, \semM{\boxedCirc{\struct \labOne} \circuitTwo {\struct \labTwo}}) = \sem{\config \circuitTwo {\struct \labTwo}}.
      \end{align*}
      Thus, we have
      \begin{align*}
        \sem{\valTwo} = (\Lambda(\sem{\app \valTwo {\struct \labOne}}), \id_I) = (\Lambda(J(\id_1, \semM{\boxedCirc{\struct \labOne} \circuitTwo {\struct \labTwo}})), \id_I) = ( \widehat{\semM{\boxedCirc{\struct \labOne} \circuitTwo {\struct \labTwo}}}, \id_I)
      \end{align*}
        as desired.
    \item[Case \textit{return}:]
      In this case, we have \( \termOne = \return \valOne \) and \( \circuitTwo = \circuitOne \).
      Since \( \circuitTwo = \circuitOne \), it suffices to show that \( \sem{\termOne} = J((\id_1, \perm); \sem{\valOne})\) for some suitable permutation isomorphism \( \perm \), and this is obvious by the definition of \( \sem{\return \valOne} \).
    \item[Case \textit{let}:]
      In this case, we must have
      \[
      \inferrule
      {\config{\circuitOne}{\termTwo_1} \eval \config{\circuitTwo}{\valTwo} \\
      \config{\circuitThree}{\termTwo_2\sub{\valTwo}{\varOne}} \eval \config{\circuitTwo}{\valOne}}
      {\config{\circuitOne}{\letin{\varOne}{\termTwo_1}{\termTwo_2}} \eval \config{\circuitOne'}{\valOne}}
      \]
      with \( \termOne =  \letin{\varOne}{\termTwo_1}{\termTwo_2} \) for some terms \( \termTwo_1, \termTwo_2 \), value \( \valTwo \) and circuit \( \circuitThree \).
      By the inversion on the typing rule, we must also have
      \begin{gather*}
        \cjudgst {\lcOne_{\termTwo_1}} {\termTwo_1} \typeTwo
        \quad \text{and} \quad
        \cjudgst {\lcOne_{\termTwo_2}, \varOne : \typeTwo} {\termTwo_2} \typeOne
      \end{gather*}
      for some \( \lcOne_{\termTwo_1} \) and \( \lcOne_{\termTwo_2} \) such that \( \lcOne_{\termOne} = \lcOne_{\termTwo_2}, \lcOne_{\termTwo_1} \).
      \begin{align*}
        \sem{\config \circuitOne \termOne}
        &= J(\id_1, \circuitOne); (\sem{\letin{\varOne}{\termTwo_1}{\termTwo_2}} \ltensor \sem{\lcOne'}) \\
        &= J(\id_1, \circuitOne); \left((\sem{\lcOne_{\termTwo_2}} \rtensor \sem{\termTwo_1}) ; \sem{\termTwo_2}\right)  \ltensor \sem{\lcOne'}) \\
        &= J(\id_1, \circuitOne); (\sem{\lcOne_{\termTwo_2}} \rtensor \sem{\termTwo_1} \ltensor \sem{\lcOne'}) ; (\sem{\termTwo_2}  \ltensor \sem{\lcOne'}) \\
        &= J(\id_1, \circuitTwo); (\id_{\sem{\lcOne_{\termTwo_2}}} \otimes J(\sem{\valTwo}) \ltensor \sem{\lcOne'}) ; (\sem{\termTwo_2}  \ltensor \sem{\lcOne'}) \tag{by I.H.}\\
        &= J(\id_1, \circuitTwo); \left(\left((\id_{\sem{\lcOne_{\termTwo_2}}} \otimes J(\sem{\valTwo})) ; \sem {\cjudgst {\lcOne_{\termTwo_2}, \varOne : \typeTwo} {\termTwo_2} \typeOne}\right)  \ltensor \sem{\lcOne'}\right) \\
        &= J(\id_1, \circuitTwo); \left( \sem{\termTwo_2 \sub \valTwo \varOne}  \ltensor \sem{\lcOne'}\right) \tag{by substitution lemma (Lemma~\ref{lem:subst-st})} \\
        &= J(\id_1, \circuitOne'); (J(\sem{\valOne}) \ltensor \sem{\lcOne'}) \tag{by I.H.} \\
        &= \sem{\config {\circuitOne'} \valOne}
      \end{align*}
 \end{description}
\end{proof}

\subsection{Computational Adequacy}
\newcommand*{\closedVals}[2]{\mathrm{CVal}_{#1 \vdash #2}}
\newcommand*{\closedTerms}[2]{\mathrm{CTerm}_{#1 \vdash #2}}
\newcommand*{\lRel}[2]{\mathrel{\mathcal{R}}_{#1 \vdash #2}}
\newcommand*{\lvRel}[2]{\mathrel{\mathcal{V}}_{#1 \vdash #2}}
\newcommand*{\lvRelCtx}[1]{\mathrel{\mathcal{V}}_{#1}}
\newcommand*{\semvalOne}{v}
\newcommand*{\semvalTwo}{w}
\renewcommand*{\unitv}{*}
\newcommand*{\subst}{\gamma}
\newcommand*{\dom}{\mathrm{dom}}
\par
The proof of computational adequacy also follows a standard strategy: we define a logical relation between semantics and syntax of \PQM{}.
Similar logical relations have been considered by Colledan and Dal Lago~\cite{ColledanDalLago24a,ColledanDalLago25} to prove the correctness of the type system of \PQR{}.
(These are defined by purely operational means.)

We say that a value (resp.~term) is \emph{closed} if the value (resp.~term) does not have any free variables (but it may have some labels).
The set of closed values of type (resp.~term) \( \typeOne \) that is well-typed under the labeled context \( \lcOne \) is denoted by \( \closedVals \lcOne \typeOne \) (resp.~\( \closedTerms \lcOne \typeOne \)).
We define relations  (indexed by judgment of the form \( \lcOne \vdash \typeOne \)) \( \lRel \lcOne \typeOne \subseteq (\flat \sem{\typeOne} \times \catM(\semM \lcOne, \sharp \sem \typeOne)) \times \closedTerms \lcOne \typeOne \) and \( \lvRel \lcOne \typeOne \subseteq \flat \sem{\typeOne}\times \closedVals \lcOne \typeOne\)  as the smallest relations satisfying the following conditions.

\begin{itemize}
  \item \( (\semvalOne, \circuitOne; \perm) \lRel \lcOne \typeOne \termOne\) if and only if \( \config {\id_{\semM {\lcOne}}} \termOne \eval \config \circuitOne \valOne \), \( \circuitOne \colon \lcOne \to \lcTwo' \) and \( \semvalOne \lvRel \lcTwo \typeOne \valOne \) for some label contexts \( \lcTwo' \) and \( \lcTwo \) such that \( \semM{\lcTwo'} \xrightarrow[\cong]{\perm} \semM{\lcTwo}\) where \( \perm \) is a permutation isomorphism.
  \item  \( * \lvRel \emptycontext \unitt \unitv \)
  \item \( n \lvRel \emptycontext \natt n \)
  \item \( * \lvRel {\labOne : \wtypeOne} \wtypeOne \labOne\)
  \item \( f \lvRel \lcOne {\arrowst \typeOne \typeTwo \mtypeOne} \valOne \) if and only if for all \( \valTwo \) such that \( \semvalTwo \lvRel \lcTwo \typeOne \valTwo \), we have \(\app f \semvalTwo \lRel {\lcOne, \lcTwo} \typeTwo \app \valOne \valTwo \). %
  \item  \( f \lvRel \emptycontext {\bang{\typeOne}} \valOne \) if and only if \(\app f * \lRel \emptycontext \typeTwo \force \valOne \)
  \item \( \sem{\boxedCirc{\struct \labOne}{\circuitOne}{\struct \labTwo}} \lvRel \emptycontext {\circt{} \mtypeOne \mtypeTwo} \boxedCirc {\struct\labOne} \circuitOne {\struct \labTwo} \)
  \item \( (\semvalOne, \semvalTwo) \lvRel {\lcOne_1, \lcOne_2} {\typeOne \otimes \typeTwo} \tuple \valOne \valTwo \) if and only if \( \semvalOne \lvRel {\lcOne_1} \valOne \) and \( \semvalTwo \lvRel {\lcOne_2} \valTwo \).
\end{itemize}

We extend the relation to a relation between pairs of a value and a substitution and typing contexts by
\[
  (\semvalOne, \subst) \lvRelCtx {\lcOne} (a_1 : \typeOne_1, \ldots, a_n : \typeOne_n) \iff  \semvalOne \lvRel \lcOne {\typeOne_1 \otimes \cdots \otimes \typeOne_n}{\langle \subst(a_1), \ldots, \subst(a_n) \rangle}
\]
where
\begin{itemize}
  \item \( \semvalOne\) is an element of \( \flat \sem{\typeOne_1} \times \cdots \times \flat \sem{\typeOne_n} \) and \item \( \subst \) is a map from \( \dom(\contextOne) \) to closed values such that, for each \( a_i \),  the only labels appearing in \( \subst(a_i) \) are those in \( \lcOne\).
\end{itemize}

The definition of \( \lRel \lcOne \typeOne \) evaluates \( \termOne\) with the identity circuit, but this does not loose generality because if the initial configuration has a circuit that is not the identity we can just concatenate the initial circuit to the circuit obtained by evaluating \( \termOne \) with the identity circuit.
\begin{lemma}
  \label{lem:config-id}
  Suppose that \( \circuitOne \colon \lcOne_0 \to \lcOne_1, \lcTwo, \lcOne_2 \) and \( \lcTwo \vdash \termOne \colon A \).
  If \( \config {\id_{\semM \lcTwo}} \termOne \eval \config \circuitTwo \valOne \) and \( \config \circuitOne \termOne \eval \config \circuitThree \valTwo \), then \( \valOne = \valTwo \) (up to renaming of labels) and \( \circuitThree = \circuitOne;(\semM{\lcOne_1} \rtensor (\circuitTwo \ltensor \semM{\lcOne_2})) \).
\end{lemma}
\begin{proof}
  By induction on the derivation of \( \config \circuitOne \termOne \eval \config \circuitThree \valTwo \).
\end{proof}

The fundamental property of the logical relations hold.
And, as usual, the theorem is a direct consequence of the fundamental property.
\begin{lemma}[Fundamental Property]
  Let \( \contextOne \) be a type environment, \( \lcOne \) be a label context such that \( \semM{\lcOne} = \sharp \sem \contextOne \).
  Suppose that \( (\semvalOne, \subst) \lvRelCtx \lcOne \contextOne \).
  Then the following holds.
  \begin{enumerate}
    \item \( \vjudgst \contextOne \valOne \typeOne \) implies \( \flat \sem \valOne(\semvalOne) \lvRel \lcOne \typeOne \subst(\valOne) \).
    \item \( \cjudgst \contextOne \termOne \typeOne \) implies, \( \sem{\termOne}(\semvalOne) \lRel \lcOne \typeOne \subst(\termOne)  \)
  \end{enumerate}
\end{lemma}
\begin{proof}
  By simultaneous induction on the type derivation.
  We show two of the most interesting cases: the case for \textit{box} and \textit{let}.
  The other cases can be proved in a similar manner.

  \begin{description}
    \item[Case \textit{box}:]
      In this case, we must have \( \termOne = \boxt{}{\valOne} \), \( \contextOne = \pcontextOne \) and \( \typeOne = \circt{}{\mtypeOne}{\mtypeTwo}{}\) for some value \( \valOne \), parameter context \( \pcontextOne \) and bundle types \( \mtypeOne \) and \( \mtypeTwo \).
      Moreover, we have \( \vjudgst \pcontextOne \valOne {\arrowst{\mtypeOne}{\mtypeTwo}{}} \).
      \par
      Suppose that \( (\semvalOne, \subst) \lvRelCtx \emptyset \pcontextOne \).
      Our goal is to show \( \sem{\termOne}(\semvalOne) \lRel{\emptyset}{\circt{}{\mtypeOne}{\mtypeTwo}} \subst(\termOne)\).
      By induction hypothesis, we have \( \flat \sem{\valOne}(\semvalOne) \lvRel{\emptyset}{\arrowst{\mtypeOne}{\mtypeTwo}{}} \subst(\valOne) \).
      Hence, we have \( \flat \sem{\valOne}(\semvalOne)(*) \lRel {\struct \labOne : \mtypeOne} {\mtypeTwo} \app {\subst(\valOne)} {\struct \labOne} \).
      By the definition of \( \lRel {\struct \labOne : \mtypeOne} {\mtypeTwo} \), we have
      \begin{gather}
        \config {\id_{\semM{\mtypeOne}}} {\app {\subst(\valOne)} {\struct \labOne}} \eval \config \circuitOne {\struct \labTwo} \label{eq:adequacy-circ-red}\\
        \flat \sem{\valOne}(\semvalOne)(*) = (*, \circuitOne; \perm) \label{eq:adequacy-circ-sem} \\
        * \lvRel {\lcTwo}{\mtypeTwo} \struct \labTwo \\
        \perm \colon \semM{\lcTwo'} \xrightarrow{\cong} \semM{\lcTwo}
      \end{gather}
      where \( C \colon \semM{\mtypeOne} \to \semM{\lcTwo'} \)
      By applying the \textit{box} rule to \eqref{eq:adequacy-circ-red}, we have \( \config{\id_{\munitt}}{\subst(\boxt {} {\valOne})} \eval \config{\id_{\munitt}}{\boxedCirc {\struct \labOne} \circuitOne {\struct \labTwo}}\).
      It remains to show that \( \sem{\termOne}(\semvalOne) = (\semM{\boxedCirc {\struct \labOne} \circuitOne {\struct \labTwo}}, \id_\munitt) \), and this is an immediate consequence of \eqref{eq:adequacy-circ-sem} and the interpretation of \( \boxt{}{}{} \).
    \item[Case \textit{let}:]
      It must be the case that
      \begin{gather*}
        \termOne = (\letin x {\termTwo} {\termThree}) \quad \contextOne = \pcontextOne, \contextOne_2, \contextOne_1 \\
        \cjudgst {\pcontextOne, \contextOne_1} {\termTwo} {\typeTwo} \\
        \cjudgst {\pcontextOne, \contextOne_2, x : \typeTwo} {\termThree} {\typeOne}
      \end{gather*}
      for some terms \( \termTwo \) and \( \termOne \), some type \( \typeTwo \) and  contexts \( \pcontextOne \), \( \contextOne_1 \) and \( \contextOne_2\).
      Let \( \lcOne_1 \) and \( \lcOne_2 \) be label contexts such that \( \semM{\lcOne_1} = \sharp \sem {\contextOne_1}\)
      and \( \semM{\lcOne_2} = \sharp \sem {\contextOne_2}\).

      Suppose that \( (\semvalOne, \subst) \lvRelCtx{\lcOne_2, \lcOne_1} \contextOne \).
      We need to show that \( \sem{\termOne}(\semvalOne) \lRel_{\lcOne_2, \lcOne_1}  \subst(\termOne) \).
      Let us write \( \semvalOne_0\) , \( \semvalOne_1 \) and \( \semvalOne_2 \) for the \( \pcontextOne\), \( \contextOne_1 \) and \( \contextOne_2 \) part of \( \semvalOne \), respectively, and let \( \subst_1 \defeq \subst \!\upharpoonright_{\dom (\pcontextOne, \contextOne_1)}\) and \( \subst_2 \defeq \subst \!\upharpoonright_{\dom (\pcontextOne, \contextOne_2)}\).
      By the induction hypothesis, we have
      \begin{align}
        \sem{\termTwo}(\semvalOne_0, \semvalOne_1) \lRel{\lcOne_1}{\typeTwo} \subst_1(\termTwo) \label{eq:adequacy-let-N-lrel}\\
        \sem{\termThree}(\semvalOne_0, \semvalOne_2, \semvalTwo) \lRel{\lcOne_2, \lcOne_1}{\typeOne} [\valTwo/x]\subst_2(\termThree) \label{eq:adequacy-let-P-lrel}
      \end{align}
      where
      \begin{align}
        \sem{\termTwo}(\semvalOne_0, \semvalOne_1) = (\semvalTwo, \circuitOne) \label{eq:adequacy-let-N-sem}\\
        \config {\id_{\semM{\lcOne_1}}}{\subst_1(\termTwo)} \eval \config{\circuitOne'}{\valTwo} \label{eq:adequacy-let-N-eval}
      \end{align}
  \end{description}
  By \eqref{eq:adequacy-let-N-lrel}, we must have \( \semvalTwo \lvRel {\lcOne_1}{\typeTwo} \valTwo\) and \( \circuitOne = \circuitOne' \) (up to permutation which we shall ignore for simplicity).
  By applying Lemma~\ref{lem:config-id} to \eqref{eq:adequacy-let-N-eval}, we also have
  \begin{align}
    \config{\id_{\semM{\lcOne_2, \lcOne_1}}}{\subst_1(\termTwo)} \eval \config{ \semM{\lcOne_2} \rtensor \circuitOne'} {\valTwo}. \label{eq:adequacy-let-N-eval-with-wires}
  \end{align}
  Now suppose that
  \begin{align}
         \sem{\termThree}(\semvalOne_0, \semvalOne_2, \semvalTwo) = (\semvalOne_{\mathrm{ret}}, \circuitTwo) \label{eq:adequacy-let-P-sem}\\
        \config {\id_{\semM{\lcOne_1}}}{\subst_1(\termTwo)} \eval \config{\circuitTwo'}{\valOne} \label{eq:adequacy-let-P-eval}\
  \end{align}
  By the definition of \eqref{eq:adequacy-let-P-lrel}, it must be the case that \( \semvalOne_{\mathrm{ret}} \lvRel {\lcOne_2, \lcOne_1}{\typeTwo} \valOne\) and \( \circuitTwo = \circuitTwo' \) (again, up to permutation which we ignore).
  Then, once again, by Lemma~\ref{lem:config-id}, we have
  \begin{align}
    \config{\semM{\lcOne_2} \rtensor \circuitOne'}{[\valTwo/x] \subst_2(\termThree)} \eval \config{(\semM{\lcOne_2} \rtensor \circuitOne') ;\circuitTwo'}{\valOne} \label{eq:adequacy-let-P-eval-with-wires}
  \end{align}
  It follows that
  \begin{align*}
    \config{\id_{\semM{\lcOne_2, \lcOne_1}}}{\subst(\termOne)} \eval \config{(\semM{\lcOne_2} \rtensor \circuitOne') ;\circuitTwo}{\valOne}
  \end{align*}
  by applying the rule \textit{let} to \eqref{eq:adequacy-let-N-eval-with-wires} and \eqref{eq:adequacy-let-P-eval-with-wires}.
  From the definition of \( \sem{\termOne} \), \eqref{eq:adequacy-let-N-sem}, and \eqref{eq:adequacy-let-P-sem} together with  \( \semvalOne_{\mathrm{ret}} \lvRel {\lcOne_2, \lcOne_1}{\typeTwo} \valOne\) we have \( \flat(\sem{\termOne}(\semvalOne_0, \semvalOne_2, \semvalOne_1)) \lvRel{\lcOne_2, \lcOne_1}{\typeOne} \valOne \).
  The equality (up to permutation) on the circuit part also holds.
  We, therefore, have \( \sem{\termOne}(\semvalOne) \lRel {\lcOne_2, \lcOne_1} {\typeOne} \subst(\termOne) \).
\end{proof}

\computationalAdequacy*
\begin{proof}
  The assumption \( \configjudgment \emptycontext \circuitOne \termOne \unitt \emptycontext \) means that there is a label context \( \lcOne \) such that \( \cjudgst \lcOne \termOne \unitt \).
  Since \( \termOne \) is closed, using the previous lemma, we obtain \( \sem \termOne (*) \lRel {\lcOne} \unitt\termOne \).
  By the definition of \( \lRel {\lcOne} \unitt \), we have \( \config{\id_{\semM \lcOne}} \termOne \eval \config \circuitThree  *  \) where \( \sem \termOne (*) = (*, \circuitThree)\).
  By  \( \sem{\config \circuitOne \termOne} = J(\sem{\config \circuitTwo \valOne})\), it follows that \( \circuitTwo = \circuitOne; \circuitThree \).
  From this and Lemma~\ref{lem:config-id}, we have \( \config \circuitOne \termOne \eval \config \circuitTwo \valOne \) as desired.
\end{proof}

\begin{remark}
  To prove the adequacy for terms typed in the type-and-effect system, we just need to add effect annotations to the logical relations.
  That is, the logical relation for computations, now becomes \( \lRel{\lcOne} {\typeOne}^{\effOne} \subseteq (\flat \sem{\typeOne} \times \catM^{\ple \effOne}(\semM \lcOne, \sharp \sem \typeOne)) \times \closedTerms \lcOne {\typeOne; \effOne} \) where \( \closedTerms \lcOne {\typeOne; \effOne} \) is the set of terms \( \termOne \) such that \( \cjudgeff \lcOne \termOne \typeOne \effOne \).
  Note that the set of circuits are now restricted to  \( \catM^{\ple \effOne}(\semM \lcOne, \sharp \sem \typeOne)) \).
  Relations \( \lvRel {\lcOne} {\typeOne} \) for arrow type, thunk type, and circuit type are changed accordingly.
\end{remark}

\clearpage
\section{Omitted Definitions from Section~\ref{sec:effect-system}}
\label{appx:interpretation-eff}
This section shows the typing rules and the interpretation that was omitted from Section~\ref{sec:effect-system}.

\subsection{Full Definition of the Typing Rules}
Here we list the full list of typing rules for reference.
The rules that were omitted from Section~\ref{sec:effect-system} are mostly identical to the corresponding rules of \PQC{}.
\begin{figure}[h]
	\centering
	\fbox{\begin{mathpar}
			\inference[\textit{unit}]
			{ }
			{\vjudgst{\pcontextOne}{\unitv}{\unitt}}
			\and
      \inference[\textit{nat}]
			{ }
			{\vjudgst{\pcontextOne}{n}{\natt}}
			\and
			\inference[\textit{lab}]
			{ }
			{\vjudgst{\pcontextOne, \labOne:\wtypeOne}{\labOne}{\wtypeOne}}
			\and
			\inference[\textit{var}]
			{ }
			{\vjudgst{\pcontextOne,\varOne:\typeOne}{\varOne}{\typeOne}}

			\and

			\inference[\textit{abs}]
			{\cjudgeff{\contextOne,\varOne:\typeOne}{\termOne}{\typeTwo} {\effOne \colon \eObjOne \to \eObjTwo}}
			{\vjudgst{\contextOne}{\abs{\varOne}{\typeOne}{\termOne}}{\arroweff{\typeOne}{\typeTwo}{\rcount{\contextOne}} {\effOne \colon
				\eObjOne \to \eObjTwo}}}
			\and
			\inference[\textit{app}]
			{\vjudgst{\pcontextOne,\contextOne_1}{\valOne}{\arroweff{\typeOne}{\typeTwo}{\mtypeOne}{\effOne \colon \eObjThree \to \eObjTwo}}
				&
				\vjudgst{\pcontextOne,\contextOne_2}{\valTwo}{\typeOne}}
			{\cjudgeff{\pcontextOne,\contextOne_1,\contextOne_2}{\app{\valOne}{\valTwo}}{\typeTwo}{\effOne \colon \eObjThree \to \eObjTwo}}
			\and
			\inference[\textit{lift}]
			{\cjudgeff{\pcontextOne}{\termOne}{\typeOne} {\effOne \colon \eObjOne \to \eObjTwo}}
			{\vjudgst{\pcontextOne}{\lift{\termOne}}{\bangeff{\typeOne} \effOne}}
			\and
			\inference[\textit{force}]
			{\vjudgst{\pcontextOne}{\valOne}{\bangeff{\typeOne} \effOne}}
			{\cjudgeff{\pcontextOne}{\force{\valOne}}{\typeOne} {\effOne \colon \eObjOne \to \eObjTwo}}
			\and
      \inference[\textit{circ}]
			{\circjudgment{\circuitOne}{\lcOne}{\lcTwo}
			  &
        \lcOne \permequiv \lcOne'
        &
        \lcTwo \permequiv \lcTwo'
			  \\
				\vjudgst{\lcOne'}{\struct\labOne}{\mtypeOne}
				\;\;
				\vjudgst{\lcTwo'}{\struct\labTwo}{\mtypeTwo}
        \;\;
			  \abstraction(\sem{\boxedCirc{\struct\labOne}{\circuitOne}{\struct \labTwo}}) = \effOne
				} {\vjudgst{\pcontextOne}{\boxedCirc{\struct\labOne}{\circuitOne}{\struct\labTwo}}{\circt{\effOne}{\mtypeOne}{\mtypeTwo}}}
        \and
        			\inference[\textit{box}]
			{\vjudgst{\pcontextOne}{\valOne}{{\arroweff{\mtypeOne}{\mtypeTwo}{I}{\effOne\colon \eObjOne \to \eObjTwo}}}}
			{\cjudgeff{\pcontextOne}{\boxt{\mtypeOne}{\valOne}}{\circt{\effOne\colon \eObjOne \to \eObjTwo}{\mtypeOne}{\mtypeTwo}} \nulleff}

      \and
      			\inference[\textit{apply}]
			{\vjudgst{\pcontextOne,\contextOne_1}{\valOne}{\circt{\effOne \colon \eObjOne \to \eObjTwo}{\mtypeOne}{\mtypeTwo}}
				\\
				\vjudgst{\pcontextOne,\contextOne_2}{\valTwo}{\mtypeOne}}
			{\cjudgeff{\pcontextOne,\contextOne_1,\contextOne_2}{\apply{\valOne}{\valTwo}}{\mtypeTwo} {\effOne \colon \eObjOne \to \eObjTwo}}
			\\

			\\
			\and
			\inference[\textit{dest}]
			{\vjudgst{\pcontextOne,\contextOne_1}{\valOne}{\tensor{\typeOne}{\typeTwo}}
				\\\
				\cjudgeff{\pcontextOne,\contextOne_2,\varOne:\typeOne,\varTwo:\typeTwo}{\termOne}{\typeThree}{\effOne}}
			{\cjudgeff{\pcontextOne,\contextOne_2,\contextOne_1}{\dest{\varOne}{\varTwo}{\valOne}{\termOne}}{\typeThree}{\effOne}}
      \and
			\inference[\textit{ifz}]
			{\vjudgst{\pcontextOne}{\valOne}{\natt}
				\\\
				\cjudgeff{\pcontextOne,\contextOne}{\termOne}{\typeOne}{\effOne}
        \and
        \cjudgeff{\pcontextOne,\contextOne}{\termTwo}{\typeOne}{\effOne}}
			{\cjudgeff{\pcontextOne,\contextOne}{\ifz{\valOne}{\termOne}{\termTwo}}{\typeOne}{\effOne}}
			\and
			\inference[\textit{pair}]
			{\vjudgst{\pcontextOne,\contextOne_1}{\valOne}{\typeOne}
				\vjudgst{\pcontextOne,\contextOne_2}{\valTwo}{\typeTwo}}
			{\vjudgst{\pcontextOne,\contextOne_1,\contextOne_2}{\tuple{\valOne}{\valTwo}}{\tensor{\typeOne}{\typeTwo}}}
			\and
			\inference[\textit{return}]
			{\vjudgst{\contextOne}{\valOne}{\typeOne} \quad \abstraction(\sharp \sem \typeOne) = \eObjOne }
			{\cjudgeff{\contextOne}{\return{\valOne}}{\typeOne} {\id_{\eObjOne} \colon \eObjOne \to \eObjOne}}
			\and
			\inference[\textit{let}]
			{
			  \cjudgeff{\pcontextOne,\contextOne_1}{\termOne}{\typeOne} {\effOne_1 \colon \eObjOne_1 \to \eObjOne_1'}
				\\
				\cjudgeff{\pcontextOne,\contextOne_2,\varOne:\typeOne}{\termTwo}{\typeTwo}  {\effOne_2 \colon \eObjOne_2 \to \eObjOne_2'} \\
				\abstraction(\sem{\sharp \contextOne_i})= {\eObjTwo_i}  & \effOne = (\id_{\eObjTwo_2} \rtensor \effOne_1); \effOne_2
			 }
			{\cjudgeff{\pcontextOne,\contextOne_2,\contextOne_1}{\letin{\varOne}{\termOne}{\termTwo}}{\typeTwo} \effOne} \and
  	\inference[\textit{sub}]
			{
			  \cjudgeff{\contextOne}{\termOne}{\typeOne} {\effOne_1 \colon \eObjOne \to \eObjTwo} &
        \effOne_1 \ple \effOne_2
      }
      {\cjudgeff{\contextOne}{\termOne}{\typeOne} {\effOne_2 \colon \eObjOne \to \eObjTwo}}
      \and
      \inference[\textit{ex}]
			{ \mathbf{perm} : \sem{\contextOne_1} \otimes \sem{\typeOne} \otimes \sem{\typeTwo} \otimes \sem{\contextOne_2} \xrightarrow{\cong}  \sem{\contextOne_1} \otimes \sem{\typeTwo} \otimes \sem{\typeOne} \otimes \sem{\contextOne_2} \\ \cjudgeff{\contextOne_1, \lvOne : \typeOne, \lvTwo : \typeTwo, \contextOne_2}{\termOne}{\typeThree}{\effOne \colon \eObjOne \to \eObjTwo}}
			{\cjudgeff{\contextOne_1, \lvTwo : \typeTwo, \lvOne : \typeOne, \contextOne_2}{\termOne}{\typeThree}{(\alpha(\mathbf{perm});\effOne})}
			\end{mathpar}}
      \caption{Typing rules for the effect system of \PQR{} (full definition).}
	\label{fig:appx-eff-typing-rules}
\end{figure}

\subsection{Interpretation of Value Judgments}
We first show how value judgments are interpreted.
As explained, value judgments are interpreted in \( \catV ( = \Set \times \disc(\catM))\) and the interpretation is essentially the same as that of the simply-typed system.
Here we only show the interpretation of values with types that have effect annotations.
\begin{align*}
  &\sem{\vjudgst{\contextOne}{\abs{\varOne}{\typeOne}{\termOne}}{\arroweff{\typeOne}{\typeTwo}{\rcount{\contextOne}} {\effOne \colon
				\eObjOne \to \eObjTwo}}}
  \defeq (\Lambda_{\effOne}(\sem{\cjudgeff{\contextOne,\varOne:\typeOne}{\termOne}{\typeTwo}{\effOne}}),  \id_{\semM{\rcount{\contextOne}}})) \\
  &\sem{\vjudgst{\pcontextOne}{\lift{\termOne}}{\bangeff{\typeOne}{\effOne}}} \defeq  (\Lambda_{\effOne}(\sem{\cjudgst{\pcontextOne}{\termOne}{\typeOne}}), \id_I)  \\
  &\sem{\vjudgst{\pcontextOne}{\boxedCirc{\struct\labOne}{\circuitOne}{\struct\labTwo}}{\circt{\effOne}{\mtypeOne}{\mtypeTwo}}}  \defeq  (!_{\sem \pcontextOne}; \widehat {\sem{\boxedCirc{\struct \labOne} \circuitOne {\struct \labTwo}}}, \id_I)
\end{align*}
The above interpretations are almost identical to those of Figure~\ref{fig:interpretation-val-judgment-st}.
The only difference is the natural bijection \( \Lambda_\effOne \), which is now indexed with \( \effOne \):
\begin{align*}
  \Set(X \times Y, Z \times \catM^{\ple \effOne}(\mObjOne, \mObjTwo)) \xrightarrow[\cong]{\Lambda_{\effOne, X, Y, Z}}\Set(X, \functionSet Y {Z \times \catM^{\ple \effOne}(\mObjOne, \mObjTwo))}.
\end{align*}
The morphism \( \hat C \colon 1 \to \catM^{\ple \effOne}(\semM{\mObjOne}, \semM{\mObjTwo}) \) is the global element \( \hat C (*) = C \) (which in turn can be defined using \( \Lambda_{e} \)).
\subsection{Interpretation of Computational Judgments}
Recall that \( \cjudgeff \contextOne \termOne \typeOne \effOne  \) is interpreted as a morphism in \( \Set \) from \( \flat \sem \contextOne \) to \( \monad^{\effOne \colon \sharp \sem{\contextOne} \to \sharp \sem{\typeOne}}(\flat \sem{\typeOne})\).
\subsubsection*{Application and forcing}
\begin{align*}
  &\left\llbracket
  \begin{matrix}
    \infer
    {\vjudgst{\pcontextOne,\contextOne_1}{\valOne}{\arroweff{\typeOne}{\typeTwo}{\mtypeOne}{\effOne \colon \eObjThree \to \eObjTwo}}
    \quad
    \vjudgst{\pcontextOne,\contextOne_2}{\valTwo}{\typeOne}}
    {\cjudgeff{\pcontextOne,\contextOne_1,\contextOne_2}{\app{\valOne}{\valTwo}}{\typeTwo}{\effOne \colon \eObjThree \to \eObjTwo}}
  \end{matrix}
  \right\rrbracket \\
  &\defeq \begin{aligned}[t]
  \semP \pcontextOne  \times \flat \sem {\contextOne_1} \times \flat \sem {\contextOne_2}
    &\xrightarrow{\Delta_{\semP{\pcontextOne}}\times \id} \semP \pcontextOne \times \sem \pcontextOne  \times \flat \semP {\contextOne_1} \times \flat \sem {\contextOne_2}  \\
        &\xrightarrow{\mathmakebox[4em]{\cong}} (\semP \pcontextOne  \times \flat \sem {\contextOne_1}) \times (\semP \pcontextOne \times \flat \sem {\contextOne_2}) \\
    &\xrightarrow{\flat \sem \valOne \times \flat \sem \valTwo} \flat \sem{\arrowst \typeOne  \typeTwo {\mtypeOne} } \times \flat \sem \typeOne \\
    &\xrightarrow{\mathmakebox[4em]{\ev}} \monad^{\effOne}(\flat \sem{\typeTwo})
      \end{aligned}
  \displaybreak[1] \\
&\left\llbracket
  \begin{matrix}
    \infer{\vjudgst{\pcontextOne}{\valOne}{\bangeff{\typeOne} \effOne}}
		{\cjudgeff{\pcontextOne}{\force{\valOne}}{\typeOne} {\effOne \colon \eObjOne \to \eObjTwo}}
  \end{matrix}
  \right\rrbracket \defeq
  \semP \pcontextOne
    \xrightarrow{\flat \sem \valOne} \flat \sem{\bangeff{\typeOne}{\effOne}} \times 1
    \xrightarrow{\mathmakebox[2em]{\ev}} \monad^{\effOne}(\flat \sem{\typeOne})
\end{align*}
Here, unlike in Section~\ref{sec:categorica-semantics-st}, \( \ev \) is the evaluation morphism for the exponential objects in \( \Set \).

\subsubsection*{Circuit operations}
\begin{align*}
  &\left\llbracket
  \begin{matrix}
    \infer
    {\vjudgst{\pcontextOne}{\valOne}{{\arroweff{\mtypeOne}{\mtypeTwo}{I}{\effOne\colon \eObjOne \to \eObjTwo}}}}
    {\cjudgeff{\pcontextOne}{\boxt{\mtypeOne}{\valOne}}{\circt{\effOne\colon \eObjOne \to \eObjTwo}{\mtypeOne}{\mtypeTwo}} \nulleff}
  \end{matrix}
  \right\rrbracket \\
  &\defeq \begin{aligned}[t]
    \semP \pcontextOne &\xrightarrow{\mathmakebox[2em]{\sem{\valOne}}} \flat \sem{\arroweff \mtypeOne \mtypeTwo {}{\effOne}} \\
    &\xrightarrow{\mathmakebox[2em]{\boxsem}} \catM^{\ple \effOne}(\semM{\mtypeOne}, \semM{\mtypeTwo}) \\
    &\xrightarrow{\mathmakebox[2em]{\unit_\munitt}} \monad^{\nulleff}( \catM^{\ple \effOne}(\semM{\mtypeOne}, \semM{\mtypeTwo}))
    \end{aligned} \\ \\
    &\left\llbracket
      \begin{matrix}
        \infer
        {\vjudgst{\pcontextOne,\contextOne_1}{\valOne}{\circt{\effOne \colon \eObjOne \to \eObjTwo}{\mtypeOne}{\mtypeTwo}}
        \\
        \vjudgst{\pcontextOne,\contextOne_2}{\valTwo}{\mtypeOne}}
        {\cjudgeff{\pcontextOne,\contextOne_1,\contextOne_2}{\apply{\valOne}{\valTwo}}{\mtypeTwo} {\effOne \colon \eObjOne \to \eObjTwo}}
      \end{matrix}
      \right\rrbracket \\
  &\defeq
    \begin{aligned}[t]
      \semP{\pcontextOne} \times \flat \sem{\contextOne_1} \times \flat \sem{\contextOne_2}
      &\xrightarrow{\Delta_{\semP \pcontextOne} \times \id}  \semP{\pcontextOne} \times \semP{\pcontextOne} \times \flat \sem{\contextOne_1} \times \flat \sem{\contextOne_2} \\
      &\xrightarrow{\mathmakebox[3em]{\cong}}  (\semP \pcontextOne  \times \flat \sem {\contextOne_1}) \times (\semP \pcontextOne \times \flat \sem {\contextOne_2}) \\
      &\xrightarrow{\flat \sem \valOne \times \flat \sem \valTwo} \catM^{\ple \effOne}(\semM{\mtypeOne}, \semM{\mtypeTwo}) \\
      &\xrightarrow{\mathmakebox[3em]{\unit_{\semM \mtypeOne}}} \monad^{\id_{\sem \mObjOne}} (\catM^{\ple \effOne}(\semM{\mtypeOne}, \semM{\mtypeTwo}))) \\
      &\xrightarrow{\mathmakebox[3em]{\mathbf{comp}}} \monad^{\effOne}(1).
    \end{aligned}
\end{align*}
Here \( \mathbf{comp}_{\mObjOne_1, \mObjOne_2, \mObjOne_3} \colon \catM(\mObjOne_1, \mObjOne_2) \times \catM(\mObjOne_2, \mObjOne_3) \to \catM(\mObjOne_1, \mObjOne_3)\) is the composition morphism.
\subsubsection*{Return and let}

\begin{align*}
  &\left\llbracket
  \begin{matrix}
    \infer
    {\vjudgst{\contextOne}{\valOne}{\typeOne} \quad \abstraction(\sharp \sem \typeOne) = \eObjOne }
    {\cjudgeff{\contextOne}{\return{\valOne}}{\typeOne} {\id_{\eObjOne} \colon \eObjOne \to \eObjOne}}
  \end{matrix}
  \right\rrbracket
    \defeq \flat \sem{\contextOne} \xrightarrow{\flat \sem{\valOne}} \flat \sem{\typeOne} \xrightarrow{\mathmakebox[2em]{\unit_{\sharp \sem \typeOne}}} \monad^{\id_{\sharp \sem \typeOne}}(\flat \sem \typeOne) \\
  \displaybreak[1] \\
  &\left\llbracket
    \begin{matrix} %
      \infer{\cjudgeff{\pcontextOne,\contextOne_1}{\termOne}{\typeOne} {\effOne_1 \colon \eObjOne_1 \to \eObjOne_1'}
      \\
      \cjudgeff{\pcontextOne,\contextOne_2,\varOne:\typeOne}{\termTwo}{\typeTwo}  {\effOne_2 \colon \eObjOne_2 \to \eObjOne_2'} \\
        \refine {\sem{\sharp \contextOne_i}} {\eObjTwo_i}  \quad \effOne = (\id_{\eObjTwo_2} \rtensor \effOne_1); \effOne_2
      }
  {\cjudgeff{\pcontextOne,\contextOne_2,\contextOne_1}{\letin{\varOne}{\termOne}{\termTwo}}{\typeTwo} \effOne}
    \end{matrix}
    \right\rrbracket \\
  &\defeq
\begin{aligned}[t]
  \sem \pcontextOne \times \flat \sem {\contextOne_2} \times \flat \sem{\contextOne_1}
  &\xrightarrow{(\Delta \times \id); \cong} \sem \pcontextOne \times \flat \sem {\contextOne_2} \times \sem \pcontextOne \times \flat \sem{\contextOne_1} \\
  &\xrightarrow{\id \times \unit_{\eObjTwo_2} \times \sem \termOne} \sem \pcontextOne \times \circmonadEff {\id_{\eObjTwo_2}} \flat \sem{\contextOne_2} \times \circmonadEff {\effOne_1 } \flat \sem{\typeOne} \\
  &\xrightarrow{\id \times \rtensorTwo; \strength}
    \circmonadEff {\id \rtensor \effOne_1}(\sem \pcontextOne \times \flat \sem{\contextOne_2} \times \flat \sem{\typeOne} ) \\
  &\xrightarrow{\circmonadEff {\id \rtensor \effOne_1} \sem \termTwo}
     \circmonadEff {\id \rtensor \effOne_1}(\circmonadEff{\effOne_2}(\flat \sem \typeTwo) ) \\
  &\xrightarrow {\mult_{\id \rtensor \effOne_1, \effOne_2}} \circmonadEff {(\id \rtensor \effOne_1); \effOne_2}(\flat \sem \typeTwo).
\end{aligned}
\end{align*}

{%

\subsubsection*{Subsumption}
\begin{align*}
  &\left\llbracket
  \begin{matrix}
  	\infer{\cjudgeff{\contextOne}{\termOne}{\typeOne} {\effOne_1 \colon \eObjOne \to \eObjTwo} \quad \effOne_1 \ple \effOne_2}
    {\cjudgeff{\contextOne}{\termOne}{\typeOne} {\effOne_2 \colon \eObjOne \to \eObjTwo}}
  \end{matrix}
  \right\rrbracket
  \defeq \flat \sem{\contextOne} \xrightarrow{\sem{M}} \monad^{\effOne_1 \colon \sharp \sem{\contextOne} \to \sharp \sem{\typeOne}}(\flat \sem{\typeOne}) \xrightarrow{\monad^{\effOne_1 \ple \effOne_2}_{\flat \sem{\typeOne}}} \monad^{\effOne_2}(\flat \sem{\typeOne})
\end{align*}

\section{Supplementary Materials on Indexed Monads and Parameterized Freyd Categories}
\label{appx:monad-Freyd}
As supplementary material, we include a brief review of indexed monads and the parameterized Freyd category to make the paper self-contained. Most of the definitions are drawn from~\cite{Atkey09}.
We present only the minimal definitions required to understand this paper, with explanations tailored to our purposes.

\subsection{Parameterized Freyd category}

We give the precise definition of parameterized Freyd category that was omitted from the body of the paper.

\begin{definition}[{\cite{Atkey09}}]
  \label{def:parametrized-Freyd}
  A \emph{parameterized Freyd category} consists of three functors \( J : \catC  \times \catS \to \catK \),  \( \rtensorF : \catC \times \catK \to \catK \) and \( \ltensorF : \catK \times \catC  \to \catK \)  such that
\begin{itemize}
\item \( J \) is identity on objects,
\item the cartesian product of \( \catC \) is respected: \( X \rtensorF J(Y, S) = J(X, S) \ltensorF Y = J(X \times Y, S)\),
  \item for each \( S \in \catS \), the transformations given by the associativity \( J(\alpha, S) \), the left unitor \( J(\lambda, S) \), the right unitor  \( J(\rho, S) \) and the symmetry \( J(\symm, S) \) of the symmetric monoidal structure arising from the finite products of \( \catC \) must be natural in the variables in all combinations of \( \times \), \( \rtensorF \) and \( \ltensorF \) that make up their domain and codomain.
\end{itemize}
As is clear from the above conditions, we are assuming that \( \catC \) has finite products.
The pair of functors \( \rtensorF \) and \( \ltensorF \) is often called the premonoidal structure of the parameterized Freyd category (with respect to \( \catC \)).
\end{definition}

The most important example (and the only example appearing in the paper) of a parameterized Freyd category is the one given by the Kleisli construction of an indexed monad.
As already explained, the Kleisli category of \(\monad : \catS^\op \times \catS  \times \catC \to \catC \) has pairs of objects of \( \catC \) and \( \catS \) as its objects and the homset \( \catC_\monad((X, S), (Y, T) \) is defined as \( \catC(X, \monad (S, T, Y))\).
Compositions of morphisms \( f \colon (X, S) \to (Y, T)\) and \( g \colon (Y, T) \to (Z, U) \) are defined by the Kleisli composition:
\begin{align*}
  X \xrightarrow{f} \monad(S, T, Y) \xrightarrow{\monad(S, T, g)} \monad(S, T, \monad(T, U, Z)) \xrightarrow{\mult_{S,T,U, Z}} \monad(S, U, Z).
\end{align*}
The (morphism part of the) functor \( \rtensorF \) is defined by
\begin{align}
  \label{eq:Freyd-premonoidal}
  f \rtensorF c \defeq X \times W \xrightarrow{f \times c} Y \times \monad(S, T, Z) \xrightarrow{\strength_{Y, S ,T , Z}} \monad(S, T, Y \times Z).
\end{align}
for \( f: X \to Y \), a morphism in \( \catC \), and \( c: (W, S) \to (Z, T) \).

\subsection{Lifting of the Premonoidal Structure}
\label{appx:lifitng}
In the main part of the paper, we claimed that the premonoidal structure is lifted to the indexed monad, parameterized Freyd category, and the category-graded monad.
Here we make precise what we mean by that.
This is just a review of an existing, but perhaps not so known, notion~\cite{Atkey09}, and we believe that having this supplementary section would help readers understand the technical details.
However, \emph{readers not interested in the technical definitions may safely ignore this section}.
To understand the interpretation, it suffices to understand that there is an operation, called the lifting of \( - \ltensor \mObjOne \), that acts on computations by modifying the underlying circuit \( \circuitOne \) to \( \circuitOne \ltensor \mObjOne \).
It may be worth mentioning that our interpretation, as well as the soundness result, does not use any specific property of the circuit monad except for circuit related operations such as \( \boxoperator \) and \( \liftoperator\).
The interpretation of the other constructs only relies on the parameterized Freyd structure \emph{and the existence of the premonoidal lifting}.

We start by reviewing the notion of lifting for indexed monads.
\begin{definition}[Lifting for indexed monads~\cite{Atkey09}]
  Let \( \monad \) be a \( \catS \)-indexed monad over a cartesian category \( \catC \) and \( F \colon \catS \to \catS \) be an endofunctor.
  A \emph{lifting of \( F \)} to \( \monad \) is a natural transformation \( F^\dagger_{S, T, X} \colon \monad(S, T, X) \to \monad (F S, F T, X) \) that commutes with the unit, multiplication (and strength of the monad):
  \[
  \begin{tikzcd}[row sep=large, column sep=large,wire types={n,n}]
	X \arrow[rd,"\unit_{FS, X}"']\arrow[r,"\unit_{S, X}"]  & \monad(S, S, X) \arrow[d,"F^{\dagger}_{S, S, X}"]\\
	& \monad(FS, FS, X)
  \end{tikzcd}
  \]
  \[
  \begin{tikzcd}[row sep=large, column sep=huge,wire types={n,n}]
	\monad(S, T, \monad(T, U, X)) \arrow[d,"\mult_{S, T, U, X}"']\arrow[rr,"F^\dagger_{S, T, \monad(T, U, X)}"] && \monad(FS, FT, \monad(T, U, X)) \arrow[d,"{\monad(FS, FT, F^{\dagger})}"] \\
	\monad(S, U, X) \arrow[r,"F^\dagger_{S, U, X}"'] & \monad(FS, S U, X) &  \arrow[l,"\mult_{FS, FT, FU, X}"]\monad(FS, FT, \monad(FT, FU, X))
  \end{tikzcd}
\]

A natural transformation \( \theta \colon F \to G \) is \emph{natural for liftings \( F^\dagger \) and \( G^\dagger \)} if the following diagram commutes.
\[
  \begin{tikzcd}[row sep=large, column sep=large,wire types={n,n}]
	\monad(S, T, X) \arrow[d,"{G^\dagger_{S, T, X}}"']\arrow[r,"{F^\dagger_{S, T, X}}"] & \monad(FS, FT, X) \arrow[d,"{\monad(FS, \theta_T, X)}"]\\
	\monad(GS, GT, X) \arrow[r,"{\monad(\theta_S, GT, X)}"'] & \monad(FS, GT, X)
  \end{tikzcd}
\]
\end{definition}

\begin{definition}
\label{def:premonoidal-lifting}
We say that an \( \catS \)-indexed monad \( \monad \) has \emph{premonoidal lifting} if there are liftings for the functors \( - \ltensor S \) and \( S \rtensor - \), written \(  (- \ltensor S )^\dagger \) and \( (S \rtensor -)^\dagger \), for every \( S \in \catS \), such that all the  associativity and left and
right unitors are natural for them.
Furthermore, the indexed monad is said to have a symmetric premonoidal lifting if the
symmetry natural transformations are also natural.
\end{definition}
As we briefly explained, the circuit monad \( \circmonad* \) has premonoidal lifting.

Next we review the notion of liftings for the parameterized Freyd category.
The liftings for an indexed monad and the liftings for the parameterized Freyd category obtained as the Kleisli category of that indexed monad are mutually related.

\begin{definition}[Lifting for parameterized Freyd categories~{\cite{Atkey09}}]
  Let \( F \colon \catS \to \catS \) be an endofunctor and \( (J, \rtensorF, \ltensorF) \) be a parametric Freyd category where \( J \colon \catC \times \catS \to \catK \).
  The parameterized Freyd category has a lifting of \( F \) if it has a functor \( F^\star \colon \catK \to \catK\) such that
  \begin{itemize}
    \item \( F^\star (J(X, S)) = J(X, FS)\) and \( F^\star(J(f, s)) = J(f, Fs) \) for \( f \in \catC(X, Y) \) and \(s \in \catS(S, T)\)
    \item \( F^\star \) respects the premonoidal structure: \( F^\star(X \rtensorF (Y, S)) = F^\star((X, S) \ltensorF Y)  = (X \times Y, F S)\) and \(F^\star(f \rtensorF c) = f \rtensor F^\star c \) (and similarly for \( \ltensorF \)).
  \end{itemize}

  A natural transformation \( \theta \colon F \to G\) is natural for liftings \( F^\star \) and \( G^\star \) if the following diagram commutes for all \( c \colon (X, S) \to (Y, T) \).
  \[
    \begin{tikzcd}[row sep=large, column sep=large,wire types={n,n}]
      J(X, FS)\arrow[d,"{F^\dagger f}"']\arrow[r,"{J(X, \theta_S)}"] & J(X, GS) \arrow[d,"{G^\dagger f}"]\\
	J(Y, FT) \arrow[r,"{J(Y, \theta_T)}"'] & J(Z, G T)
  \end{tikzcd}
  \]
\end{definition}
\begin{proposition}[{\cite[Theorem~3]{Atkey09}}]
Let \( \monad \) be an indexed monad over a cartesian category \( \catC \).
If \( F^\dagger \) is a lifting of an endofunctor \( F : \catS  \to \catS \), then we can construct a lifting \( F^\star \) on the parameterized Freyd category \( \catC_\monad \) and vice versa.
These operations are inverse.
If a natural transformation from \( F \) to \( G \) is natural for liftings \( F^\dagger \) and \( G^\dagger \), then it is also natural for liftings \( F^\star \)and \( G^\star \), and vice versa.
\end{proposition}

The premonoidal lifting for parameterized Freyd category is defined as in Definition~\ref{def:premonoidal-lifting}, and it is easy to check that if \( \monad \) has a premonoidal lifting so does its Kleisli category.
We note that a parametrized Freyd category with a premonoidal lifting, in a sense, has two premonoidal structures one respect to \( \catC \) and the other respect to \( \catS \).
This allows us to define a binoidal functor \( (X, S) \rtensor - \colon \catK \to \catK \), for each \( (X, S) \in \catK \), defined as \( (X \rtensorF -); (S \rtensor -)^\star \).
The functor \(  - \ltensor (X, S)\colon \catK \to \catK \) is defined analogously.
It is not hard to see that this is makes \( \catK \) a premonoidal category.
In particular, for \( \Set_{\circmonad*} \), we have
\[
  ((X, S) \rtensor f)(x, y) = ((x, z), S \rtensor \circuitOne)
\]
provided that \( f \colon Y \to \circmonad T U Z\) maps \( y \) to \( (z, \circuitOne)\).
In other words, \(  (X, S) \rtensor f\) just performs the computation against the second element and, at the same time, augments wires to the underlying circuit of the computation.
It is this premonoidal structure that is used in the interpretation of the computational judgments given in Section~\ref{sec:categorica-semantics-st}.
We note that \( J(f, s) \) is a central morphism if \( s \) is with respect to the premonoidal structure of \( \catS \).
Therefore, the values are interpreted as central morphisms in the interpretation.

The notion of lifting of functors can be naturally extended to category-graded monads.

\begin{definition}[Lifting for category-graded monad]
  Let \( \monad \) be a \( \catA \)-graded monad over a cartesian category \( \catC \) and \( F \colon \catA \to \catA \) be an endofunctor.
  A \emph{lifting of \( F \)} to \( \monad \) is a family of natural transformation \( F^\dagger_{f, X} \colon \monad^f X \to \monad^{Ff} X \) indexed by morphisms \( f \) in \(\catA \) that commutes with the unit, multiplication (and strength of the monad):
  \[
  \begin{tikzcd}[row sep=large, column sep=large,wire types={n,n}]
	X \arrow[rd,"\unit_{F a, X}"']\arrow[r,"\unit_{a, X}"]  & \monad^{\id_a} \arrow[d,"F^{\dagger}_{\id_a, X}"]\\
	& \monad^{\id_{Fa}}
  \end{tikzcd}
  \quad
  \begin{tikzcd}[row sep=large, column sep=huge,wire types={n,n}]
	\monad^f(\monad^g X) \arrow[d,"\mult_{f, g X}"']\arrow[rr,"F^\dagger_{f, \monad^g X}"] && \monad^{F f} (\monad^{g} X) \arrow[d,"{\monad^{F f}(F^{\dagger})}"] \\
	\monad^{f; g} X  \arrow[r,"F^\dagger_{f;g, X}"'] & \monad^{Ff; Fg} X &  \arrow[l,"\mult_{Ff, Fg, X}"]\monad^{Ff}(\monad^{Fg} X)
  \end{tikzcd}
\]

A natural transformation \( \theta \colon F \to G \) is \emph{natural for liftings \( F^\dagger \) and \( G^\dagger \)} if there is a generalized unit~\cite{OrchardWE20} \(\unit_{\theta_a, X} \colon X \to \monad^{\theta_a} X \) for each component \( \theta_a \) such that
\[
  \begin{tikzcd}[row sep=large, column sep=large,wire types={n,n}]
    \monad^f  X  \arrow[dd,"{G^\dagger_{f, X}}"']\arrow[r,"{F^\dagger_{f, X}}"] & \monad^{F f} \arrow[r,"{\monad^{F f} \unit_{\theta_b, X}}"] & \monad^{F f} \monad^{\theta_b} X \arrow[d,"\mult"] & \\
    && \monad^{F f; \theta_b} \arrow[d,equal] \\
	\monad^{G f} X \arrow[r,"{\unit_{a, \monad^{G f} X}}"'] & \monad^{\theta_a} \monad^{G f} \arrow[r,"{\mult_{\theta_a; Gf}}"'] X  & \monad^{\theta_a; Gf} X
\end{tikzcd}
\]
commutes for each \( f \colon a \to b\).
\end{definition}
For the \( \circmonad*^{\effOne} \) the lift of the premonoidal product \( T \rtensor - \) is once again just the operation that maps the circuit part of a computation \( \circuitOne \) to \( T \rtensor C \).
The only natural transformation we are interested in is the symmetry (as we are working in a strict premonoidal category), and the generalized unit for the symmetry \( \symm \) is just the map \( x \mapsto (x, \symm) \).

\end{document}